\documentclass[12pt]{article}

\usepackage{amsmath}
\usepackage[round]{natbib}
\usepackage{amssymb,color}
\usepackage{appendix}
\usepackage{amsthm}
\usepackage{enumerate}
\usepackage{hyperref}
\usepackage{mathtools} 
\mathtoolsset{showonlyrefs=true}  

\usepackage{caption}

\usepackage[margin=0.7in]{geometry}

\numberwithin{equation}{section}
\usepackage{comment}

\usepackage{titling}
\newtheorem{thm}{Theorem}[section]

\newtheorem{defi}{Definition}[section]
\newtheorem{assume}{Assumption}[section]
\newtheorem{prop}{Proposition}[section]

\newtheorem*{nota}{Notation}

\newtheorem{lemma}{Lemma}[section]
\newtheorem{ex}{Example}[section]
\newtheorem{remark}{Remark}[section]

\begin{document}

\author{Hyungbin Park \thanks{hyungbin@snu.ac.kr, hyungbin2015@gmail.com} \\ \\ \normalsize{Department of Mathematical Sciences} \\ 
	\normalsize{Seoul National University}\\
	\normalsize{1, Gwanak-ro, Gwanak-gu, Seoul, Republic of Korea} 
}
\title{Convergence rates of large-time sensitivities with the Hansen--Scheinkman decomposition}

\maketitle

\abstract{This paper investigates the large-time asymptotic behavior of the sensitivities of cash flows. In quantitative finance, the price of a cash flow  is expressed in terms of a pricing operator of a Markov diffusion process.  We study the extent to which the pricing operator is affected by small changes of the underlying Markov diffusion. The main idea is a partial differential equation (PDE) representation of the pricing operator
	by incorporating  the Hansen--Scheinkman decomposition  method. The sensitivities of the cash flows and their large-time convergence rates can be represented via simple expressions in terms of  eigenvalues and eigenfunctions of the pricing operator. Furthermore, compared to the work of Park (Finance Stoch. 4:773-825, 2018), more detailed convergence rates are provided. In addition, we discuss the application of our results to three practical problems: utility maximization, entropic risk measures, and bond prices. Finally, as examples, explicit results for several market models such as the Cox--Ingersoll--Ross (CIR) model, 3/2 model and  constant elasticity of variance (CEV) model are presented.}

\section{Introduction}

In financial mathematics, sensitivity analysis is used to demonstrate how  changes of parameters affect cash flows.
A cash flow is expressed in expectation form as  
\begin{equation}\label{eqn:p_T}
	p_T:=\mathbb{E}_\xi^{\mathbb{P}}[e^{-\int_0^T r(X_s)\,ds} h(X_T)],  
\end{equation}
where $\mathbb{E}_\xi^\mathbb{P}$ is an expectation, $r$ and $h$ are suitable measurable functions, and $X=(X_t)_{t\geq0}$ is an underlying stochastic process with $X_0=\xi.$
This paper deals with the 
sensitivities of the expectation $p_T$  with respect to changes of the underlying process $X$ as well as their large-time asymptotic behaviors as $T\to\infty.$ 
The underlying process $X$ is assumed to be a   Markov diffusion, and the Markov process  $X$ with killing rate $r$ generates a pricing operator. 
The main idea is a partial differential equation (PDE) representation of the pricing operator  
by incorporating  the Hansen--Scheinkman decomposition method.
We conclude that the large-time behavior of the sensitivities is expressed in terms of eigenvalues and eigenfunctions of the pricing operator.

One of the core concepts of this paper is the Hansen--Scheinkman decomposition. 
Under suitable conditions, the expectation $p_T$ is expressed as 
\begin{equation}
	\label{eqn:intor_p_T}
	p_T=\phi(\xi)e^{-\lambda T}
	f(T,\xi)
\end{equation} 
for a positive measurable function  $\phi(\cdot),$    a positive number $\lambda$,
and a measurable function $f(T,\cdot).$
The function $f(T,\xi)$ converges to a nonzero constant, which is independent of $\xi$, as $T\to\infty.$ 
The key aspect of this decomposition is that the function $f$ is regarded as a  negligible term as $T\to\infty$ so that the behavior of $p_T$ is determined by the two factors $\phi(\xi)$ and $e^{-\lambda T}.$
In particular,  the large-time behavior of $p_T$ 
satisfies $|p_T|\leq c e^{-\lambda T}$ for a positive constant $c,$ independent of $T.$

This paper essentially investigates two types of sensitivities. The first is the sensitivity with respect to the initial value $X_0=\xi,$ which will be discussed in Section \ref{sec:long_term_sen_analy}. For the first-order sensitivity, known as the delta value,   the derivative $\partial_\xi p_T$ and its large-time behavior as $T\to\infty$ are of interest to us.  
From the Hansen--Scheinkman decomposition presented in Eq.\eqref{eqn:intor_p_T}, it follows that
$$\frac{\partial_\xi p_T}{p_T}=\frac{\phi'(\xi)}{\phi(\xi)}+\frac{f_x(T,\xi) }{f(T,\xi)}\,.$$
We  show that 
the derivative $f_x(T,\xi)$ can  also be expressed in expectation form in Eq.\eqref{eqn:p_T}.
This is achieved by representing the function $f(T,\xi)$
into a PDE form using the Feynman–Kac formula, then differentiating the PDE,
and then reconstructing the expectation for the derivative.
Applying  
the Hansen--Scheinkman decomposition repeatedly to the derivative $f_x(T,\xi),$ 
we have the expression, which is similar to Eq.\eqref{eqn:intor_p_T}, 
\begin{equation}
	\begin{aligned}
		f_x(T,\xi)
		&=\hat{\phi}(\xi)e^{-\hat{\lambda} T}\hat{f}(T,\xi)
	\end{aligned}
\end{equation}
for a positive measurable function $\hat{\phi}(\cdot),$    
a positive number $\hat{\lambda}$,
and a measurable function $\hat{f}(T,\cdot)$
converging to a nonzero constant as $T\to\infty.$ 
Finally, we have
\begin{equation}\label{eqn:intro_delta}
	\left|\frac{\partial_\xi p_T}{p_T}
	-\frac{\phi'(\xi)}{\phi(\xi)}\right|
	=\left|\frac{f_x(T,\xi) }{f(T,\xi)}\right|\leq  ce^{-\hat{\lambda} T}\,,\;T\ge0 
\end{equation} 
for some positive constant $c.$
Thus, $\frac{\partial_\xi p_T}{p_T}$ converges to $\frac{\phi'(\xi)}{\phi(\xi)}$ as $T\to\infty$ and its exponential convergence rate is $\hat{\lambda}.$ Further details are discussed in Theorem  \ref{thm:initial_main_thm}.
For the second-order sensitivity, known as the gamma value,  we present a similar argument for $\partial_{\xi\xi}p_T$ and its asymptotic behavior. We show that
\begin{equation}
	\left|\frac{{\partial_{\xi\xi}} p_T}{p_T}-\frac{\phi''(\xi)}{\phi(\xi)}
	\right|\leq ce^{-\hat{\lambda} T}
\end{equation}
for some positive constant $c.$
The asymptotic behavior is discussed in further detail in Theorem \ref{thm:2nd_sen}.

The second type of sensitivity
includes the drift and diffusion sensitivities,
which are known as the rho value and the vega value, respectively.  
Let $(X_t^{(\epsilon)})_{t\ge0}$ be the 
underlying process with perturbed drift or diffusion terms. The precise meaning of  perturbation is given in
Assumption \ref{assume:epsilon}. Here, $\epsilon$ can be understood as a perturbation parameter.
Let
$$p_T^{(\epsilon)}:=\mathbb{E}_{\xi}^{\mathbb{P}}
[e^{-\int_{0}^{T}r(X_{s}^{\epsilon})ds}
f(X_{T}^{\epsilon})]$$
be the expectation corresponding to the perturbed underlying process. We want to investigate the large-time behavior of $$\frac{\partial}{\partial\epsilon}\Big|_{\epsilon=0}p_T^{(\epsilon)}\,.$$
Under the assumption that the Hansen--Scheinkman decomposition is applicable to each $\epsilon$ 
(Assumption  \ref{assume:epsilon_a}), 
we obtain
$$p_T^{(\epsilon)}=\phi^{(\epsilon)}(\xi)\,e^{-\lambda^{(\epsilon)} T}f^{(\epsilon)}(T,\xi)\,,$$
which is analogous to Eq.\eqref{eqn:intor_p_T}. 
We verify that under some circumstances,
$$\left|\frac{1}{T}\frac{\partial}{\partial \epsilon}\Big|_{\epsilon=0} \ln p_T^{(\epsilon)}+\frac{\partial}{\partial \epsilon}\Big|_{\epsilon=0}\lambda^{(\epsilon)}\right|
\leq  \frac{c}{T}$$
for a positive constant $c.$
Thus, $\frac{1}{T}\frac{\partial}{\partial \epsilon}|_{\epsilon=0} \ln p_T^{(\epsilon)}$ converges to $-\frac{\partial}{\partial \epsilon}|_{\epsilon=0}\lambda^{(\epsilon)}$ as $T\to\infty$ and its  convergence rate is $O(1/T).$

Compared to the previous work of \cite{park2018sensitivity}, this paper has three distinguishing features.
First, for the first-order sensitivity with respect to the initial value $X_0=\xi$, the exponential convergence rate is demonstrated. 
Eq.\eqref{eqn:intro_delta} (or Theorem \ref{thm:initial_main_thm}) 
implies that $\frac{\partial_\xi p_T}{p_T}$ converges to $\frac{\phi'(\xi)}{\phi(\xi)}$ as $T\to\infty$ and its exponential convergence rate is $\hat{\lambda}.$
The previous work also showed that $\frac{\partial_\xi p_T}{p_T}$ converges to $\frac{\phi'(\xi)}{\phi(\xi)}$ as $T\to\infty$; however, its exponential convergence rate was not provided.
Second, the second-order sensitivity with respect to the initial value $X_0=\xi$ is analyzed in this paper (Section \ref{sec:2nd_sen}), whereas it was not addressed  in the previous work  at all.
Third, for the drift and diffusion perturbations, 
the current paper adopts more relaxed assumptions compared to the previous work. As the previous work relies on Malliavin calculus, it requires strong conditions such as continuous differentiability with bounded derivatives and  uniform ellipticity on drift and diffusion functions.  
As a specific example, if the underlying process is the constant elasticity of variance (CEV) model 
$$\frac{dX_t}{X_t}=\mu\,dt+\sigma {X_t}^{\beta}\,dB_t\,,\;X_0=\xi\,,$$
this paper can address the sensitivity with respect to  the  leverage effect parameter $\beta$ (Sections \ref{sec:CEV_n} and \ref{sec:CEV_p}), whereas it cannot be analyzed by the method used in the previous work.

Sensitivity analysis has been studied for various topics in quantitative finance.  
\cite{fournie1999applications}  
presented an original probabilistic method for numerical computation of the sensitivities. They employed Malliavin calculus and 
demonstrated option price sensitivities for hedging purposes.
\cite{gobet2005sensitivity} derived an expectation form of the sensitivity of the expected cost by employing three methods: the Malliavin calculus approach, the adjoint approach, and the martingale approach.  
\cite{kramkov2006sensitivity} developed sensitivity analysis of the optimal expected utility with respect to initial capital perturbations. \cite{mostovyi2017sensitivity} and \cite{mostovyi2018asymptotic} conducted sensitivity analysis of  the optimal expected utility with respect to small changes of the underlying market models.
\cite{park2019sensitivity} investigated
the sensitivities of the long-term expected utility of optimal portfolios for an investor with constant
relative risk aversion under incomplete markets.

Many authors have investigated the behavior of long-term cash flows.
\cite{fleming1995risk} studied the  long-term growth of expected utility with constant relative risk aversion and reformulated it as an infinite-time-horizon risk-sensitive control problem. 
\cite{hansen2012dynamic}, \cite{hansen2009long}, and \cite{hansen2012pricing}  
demonstrated
a long-term risk-return trade-off by employing the Hansen--Scheinkman decomposition.
\cite{leung2017long} studied the long-term growth rate of the expected utility from holding the leveraged exchange traded funds. They also used the Hansen--Scheinkman decomposition for the analysis. However, the parameter sensitivities were not covered in that paper.
\cite{liu2013portfolio}
demonstrated a computational method for evaluating optimal portfolios with special emphasis on long-horizon asymptotics.
\cite{robertson2015large} analyzed the large-time asymptotic behavior of solutions to semi-linear Cauchy problems with 
direct applications to  long-term portfolio choice problems.

The remainder of this paper is organized as follows. 
Section \ref{sec:HS} introduces the Hansen--Scheinkman decomposition
method as an essential tool of this paper.
Section \ref{sec:long_term_sen_analy}
investigates the sensitivities with respect to the initial value, the drift term, and the diffusion term  
and demonstrates their large-time behaviors.
Section \ref{sec:appli} discusses direct applications to three topics: utility maximization, entropic risk measures, and bond prices.
Section \ref{sec:ex} presents three specific examples:
the Cox--Ingersoll--Ross (CIR) model, the 3/2 model, and the CEV model.
Finally, Section \ref{sec:con} summarizes the results. The proofs and detailed calculations are provided in the appendices.

\section{Hansen--Scheinkman decomposition}
\label{sec:HS}

The  Hansen--Scheinkman decomposition (\cite{hansen2009long}) is one of  the main  techniques employed in this paper.
Given a Markov diffusion $(X_t)_{t\ge0}$ and a function $r(\cdot),$
this decomposition provides an expression 
of the  operator
$$h\mapsto\mathbb{E}_\xi^{\mathbb{P}}[e^{-\int_0^T r(X_s)\,ds} h(X_T)]  $$
in terms of  eigenvalues and eigenfunctions of this  operator.
In the following sections,  we will describe its mathematical formulation.

\begin{nota}
	Regarding the initial value $X_0$ as a given  constant, we use the notation $\xi$ so that 
   $\mathbb{E}_\xi^{\mathbb{P}}[\,\cdots]$ means the expectation when $X_0=\xi.$ Occasionally, we need to regard the initial value as a variable, in which case, the notation $x$ is used instead of $\xi.$ 
   This new notation $x$ makes us much convenient while dealing with PDEs. For example, see Eq.\eqref{eqn:1st_order_f} and Eq.\eqref{eqn:f_x_PDE}.
\end{nota}

\subsection{Consistent family of probability measures}

We begin with the notion of a consistent family of probability measures.
While considering a filtered probability space $(\Omega,\mathcal{F},(\mathcal{F}_t)_{t\ge0},\mathbb{P}),$ the  probability measure $\mathbb{P}$ is an object defined on the sigma-algebra $\mathcal{F}.$ The probability measure $\mathbb{P}$ is universal in the sense that the sigma-algebra contains all sub-sigma-algebras $(\mathcal{F}_t)_{t\ge0}.$ Instead of such a universal probability measure, we introduce  a family of probability measures  $(\mathbb{P}_t)_{t\ge0}$ where each $\mathbb{P}_t$ is defined on the sub-sigma-algebra $\mathcal{F}_t$.

The main reason for introducing this concept is that  
defining such a universal probability measure 
is occasionally impossible.
As a specific case, the change of measure in the Girsanov theorem holds for finite time horizon $[0,T]$ but it may not hold for infinite time horizon $[0,\infty)$. Thus, to use the  Girsanov theorem for studying large-time behavior as $T\to\infty,$ it is more convenient to deal with a family of probability measures instead of a universal probability measure.

\begin{defi}
	Let $(\Omega,\mathcal{F})$ be a measurable space and $(\mathcal{F}_t)_{t\ge0}$ be a  filtration. 
	We say that a family $(\mathbb{P}_t)_{t\ge0}$ of probability measures is {\em consistent} if each  probability measure $\mathbb{P}_t$ is defined on $\mathcal{F}_t$ and if
	$$\mathbb{P}_{t}(A)=\mathbb{P}_{t'}(A)$$ for any $0\le t \le t'$ and any $A\in\mathcal{F}_{t}.$
	We say that $(\Omega,\mathcal{F},(\mathcal{F}_t)_{t\ge0},(\mathbb{P}_t)_{t\ge0})$
	is   a {\em consistent probability space}.
\end{defi}
\noindent In this paper, we abuse the notations $\mathbb{P}$ and $\mathbb{P}_t$ without ambiguity. For a $\mathcal{F}_t$-measurable random variable $X,$ the expectation $\mathbb{E}^{\mathbb{P}_t}(X)$ is denoted as $\mathbb{E}^{\mathbb{P}}(X).$ 
This notation is not confusing because the family $(\mathbb{P}_t)_{t\ge0}$ is consistent so that $\mathbb{E}^{\mathbb{P}_t}(X)=\mathbb{E}^{\mathbb{P}_{t'}}(X)$ for any $t'\ge t.$

We present basic definitions of several probabilistic concepts. These definitions are straightforward.
\begin{defi}\label{defi:consistent_prob}
	Let $(\Omega,\mathcal{F},(\mathcal{F}_t)_{t\ge0},(\mathbb{P}_t)_{t\ge0})$
	be  a   consistent probability space.
	\begin{enumerate}[(i)]
		\item We say that a process $B=(B_t)_{t\ge0}$ 
		is a $d$-dimensional Brownian motion on the consistent probability space if for each $T\ge0$,  the process 
		$(B_t)_{0\le t\le T}$ is a usual $d$-dimensional Brownian motion on the filtered probability space $(\Omega,\mathcal{F}_T,(\mathcal{F}_t)_{0\le t\le T},\mathbb{P}_T).$ 
		\item We say that a process $X=(X_t)_{t\ge0}$ is a Markov process  (respectively, a martingale) on the consistent probability space if for each $T\ge0$, the process $(X_t)_{0\le t\le T}$ is a Markov process  (respectively, a martingale)  on  the filtered probability space $(\Omega,\mathcal{F}_T,(\mathcal{F}_t)_{0\le t\le T},\mathbb{P}_T).$ 
		\item We say that a process $X=(X_t)_{t\ge0}$ is
		a unique strong solution of the  SDE   (respectively, satisfies the SDE)
		$$dX_t=b(X_t)\,dt+\sigma(X_t)\,dB_t\,,\;X_0=\xi$$  on the consistent probability space 
		if for each $T\ge 0$, the process $(X_t)_{0\le t\le T}$ is a unique strong solution of the   SDE (respectively, satisfies the SDE) on  the filtered probability space $(\Omega,\mathcal{F}_T,(\mathcal{F}_t)_{0\le t\le T},\mathbb{P}_T).$ 
		\item Let $\mathcal{D}$ be an open subset of $ \mathbb{R}^d$ and  $X$ be a Markov process with state space $\mathcal{D}$ on the consistent probability space. We say that the process $X$ is recurrent   if  
		$$R(x,A):=\int_0^\infty \mathbb{E}_x^\mathbb{P}[\mathbb{I}_A(X_t)]\,dt=0\;\textnormal{ or }R(x,A)= \infty$$
		for any Borel set $A\subseteq\mathcal{D}$ and any $x\in \mathcal{D}.$ Refer to \cite[Theorem 4.3.6 on page 150]{pinsky1995positive} and \cite[Definition 3.1]{qin2016positive}.
	\end{enumerate} 
\end{defi}

\subsection{Recurrent eigenpairs}

Let  $(\Omega,\mathcal{F},(\mathcal{F}_t)_{t\ge0},(\mathbb{P}_t)_{t\ge0})$ be a consistent probability space that has  
a one-dimensional  Brownian motion $(B_t)_{t\ge0}.$ 
We consider a quadruple of functions $(b,\sigma,r,h)$ satisfying Assumptions \ref{assume:a}--\ref{assume:h} below on the consistent probability space $(\Omega,\mathcal{F},(\mathcal{F}_t)_{t\ge0},(\mathbb{P}_t)_{t\ge0})$ with given Brownian motion $(B_t)_{t\ge0}.$

\begin{assume}\label{assume:a} Let $\mathcal{D}$ be an open interval in $\mathbb{R}$ and $b:\mathcal{D}\to\mathbb{R}$ and $\sigma:\mathcal{D}\to\mathbb{R}$  be continuously differentiable functions with $\sigma>0.$ 
	For each $\xi\in\mathcal{D},$ the SDE
	$$dX_t=b(X_t)\,dt+\sigma(X_t)\,dB_t\,,\;X_0=\xi$$
	has a unique strong solution on $\mathcal{D},$ i.e.,  $\mathbb{P}_T(X_t\in\mathcal{D} \textnormal{ for all } 0\le t \le T)=1$  for each $T\ge0.$ 
\end{assume}

\begin{assume}\label{assume:r}
	The function $r:\mathcal{D}\to\mathbb{R}$ is a continuous function.
\end{assume}

For given $b,\sigma$ and $r$ satisfying the above-mentioned assumptions, we can 
define a pricing operator $\mathcal{P}$ as $$\mathcal{P}_Th(x)=\mathbb{E}_x^{\mathbb{P}}[e^{-\int_0^T r(X_s)\,ds} h(X_T)]\,,\;x\in\mathcal{D}$$
so that $p_T=\mathcal{P}_Th(\xi).$
For a real number $\lambda$ and a positive measurable function $\phi,$ we say that
a pair $(\lambda,\phi)$
is an {\em eigenpair} of $\mathcal{P}$ if
\begin{equation}\label{eqn:eigen_semigp}
\mathcal{P}_T\phi(x)=e^{-\lambda T}\phi(x)\quad\textnormal{for all }\; T>0\,,\,x\in\mathcal{D}\,. 
\end{equation}
For each eigenpair $(\lambda,\phi),$ the process
\begin{equation}\label{eqn:M}
M_t^\phi:=e^{\lambda t-\int_{0}^{t} r(X_{s})\,ds}\,\frac{\phi(X_{t})}{\phi(\xi)}\,,\;t\ge0 
\end{equation}
is a positive martingale.
For the proof, observe that 
$\mathbb{E}_\xi^\mathbb{P}[e^{-\int_{t}^{T} r(X_{s})\,ds} \phi(X_{T})|\mathcal{F}_t]=e^{-\lambda(T-t)}\phi(X_t)$ for $0\le t\le T$ from the Markov property of $X.$
Then 
\begin{equation}
\begin{aligned}
\mathbb{E}_\xi^\mathbb{P}[M_T^\phi|\mathcal{F}_t]
&=\mathbb{E}_\xi^\mathbb{P}[e^{-\int_{t}^{T} r(X_{s})\,ds} \phi(X_{T})|\mathcal{F}_t] e^{\lambda T-\int_0^tr(X_s)\,ds}\frac{1}{\phi(\xi)}=M_t^\phi\,,
\end{aligned}
\end{equation}
which implies that $M^\phi$ is a martingale.

One can define a measure $\hat{\mathbb{P}}_t^\phi$ on
each $\mathcal{F}_t$
by
$$\frac{d\hat{\mathbb{P}}_t^\phi}{d\mathbb{P}_{t}}=M_t^\phi\,.$$
Then, it can be easily checked that the family $(\hat{\mathbb{P}}_t^\phi)_{t\ge0}$  is consistent, i.e., 
$$\hat{\mathbb{P}}_{t}^\phi(A)=\hat{\mathbb{P}}_{t'}^\phi(A)$$
for any $A\in\mathcal{F}_{t}$ and $0\leq t\le t'.$  The family $(\hat{\mathbb{P}}_t^\phi)_{t\ge0}$
is called the consistent family of the {\em eigen-measures} with respect to $\phi.$

\begin{assume}\label{assume:HS}
	There exists an eigenpair $(\lambda,\phi)$ of the operator $\mathcal{P}$ with $\lambda>0$ 
	such that  
	the process $X$ is recurrent on the consistent probability space 
	$(\Omega,\mathcal{F},(\mathcal{F}_t)_{t\ge0},(\hat{\mathbb{P}}_t^\phi)_{t\ge0}).$
\end{assume}

\noindent In this case, the positive number $\lambda$, the measurable function $\phi$, and the pair $(\lambda,\phi)$ are called the {\em recurrent eigenvalue}, the {\em recurrent eigenfunction}, and the {\em recurrent eigenpair}, respectively. Several studies have investigated the existence of the recurrent eigenpair, e.g., \cite[Section 9]{hansen2009long} and
\cite[Section 5]{qin2016positive}.
Hereafter, we use the notations $M,$  $\hat{\mathbb{P}},$ $(\hat{\mathbb{P}}_t)_{t\ge0},$ and  $(\hat{B}_t)_{t\ge0}$  instead of 
$M^\phi,$  $\hat{\mathbb{P}}^\phi,$ $(\hat{\mathbb{P}}_t^\phi)_{t\ge0},$ and $(\hat{B}_t^\phi)_{t\ge0},$ respectively, without $\phi.$
These notations are not confusing because 
the recurrent eigenpair $(\lambda,\phi)$ is unique if it exists (Proposition 7.2 in \cite{hansen2009long}
and Theorem 3.1 in \cite{qin2016positive}).

\begin{remark}\label{rem:HS}
	In Assumption \ref{assume:HS}, we assume that the recurrent eigenvalue $\lambda$ is positive. This does not hold in general, however it is known that if   $r\ge0$ and $r\not\equiv 0$, then $\lambda>0.$ Refer to \cite[Theorem 3.3 (iii) on page 148]{pinsky1995positive}.
\end{remark}

The next statement is about some regularity condition on the recurrent eigenfunction.
\begin{assume} \label{assume:diff_phi}
	The recurrent eigenfunction $\phi$ is twice continuously differentiable on $\mathcal{D}.$
\end{assume}
\noindent We provide a sufficient condition to satisfy this assumption.
Assumption \ref{assume:diff_phi} holds if 
$r$ satisfies the
Lipschitz condition on every open subinterval $\mathcal{D}'$ with $\overline{\mathcal{D}'}\subseteq\mathcal{D}.$
This is obtained from \cite[Theorem 3.2 (ii) on page 146]{pinsky1995positive}, which implies that
the recurrent eigenfunction $\phi$ is twice continuously differentiable on $\mathcal{D}$ (in fact, $\phi$ is in $C_L(D)$) if Assumption $\tilde{H}_{\textnormal{loc}}$ holds.
See \cite[pages 123 and 144]{pinsky1995positive} for definitions of Assumption $\tilde{H}_{\textnormal{loc}}$ and $C_L(D),$ respectively.
Since
$b,\sigma:\mathcal{D}\to\mathbb{R}$  are continuously differentiable and $\sigma>0$ 
by Assumption \ref{assume:a}, 
Assumption $\tilde{H}_{\textnormal{loc}}$ holds if $r$ satisfies the
Lipschitz condition on every open subinterval $\mathcal{D}'$ with $\overline{\mathcal{D}'}\subseteq\mathcal{D}.$


We now describe the dynamics of $X$ under the recurrent eigen-measure.
By the Ito formula, 
$$M_T^\phi=e^{\int_0^T(\sigma\phi'/\phi)(X_s)\,dB_s-\frac{1}{2}\int_0^T(\sigma\phi'/\phi)^2(X_s)\,ds}\,.$$ 
From the Girsanov theorem, the process
$$\hat{B}_t^\phi=-\int_0^t(\sigma\phi'/\phi)(X_s)\,ds+B_t\,,\;0\le t\le T$$
is a Brownian motion on the filtered probability space $(\Omega,\mathcal{F}_T,(\mathcal{F}_t)_{0\le t\le T},\hat{\mathbb{P}}_T^\phi)$ for each $T\ge0.$ In other words, 
the process $(\hat{B}_t^\phi)_{t\ge0}$ is a Brownian motion
on the consistent probability space 
$(\Omega,\mathcal{F},(\mathcal{F}_t)_{t\ge0},(\hat{\mathbb{P}}_t^\phi)_{t\ge0})$ (see Definition 
\ref{defi:consistent_prob}).
Similarly,  
the process
$X$ satisfies
$$dX_t=(b+\sigma^2\phi'/\phi)(X_t)\,dt+\sigma(X_t)\,d\hat{B}_t^\phi $$
on the consistent probability space 
$(\Omega,\mathcal{F},(\mathcal{F}_t)_{t\ge0},(\hat{\mathbb{P}}_t^\phi)_{t\ge0}).$

In this paper, a two-variable function $f=f(t,x)$ is said to be $C^{1,2}$ if
$f$ is once continuously  differentiable with respect to $t$ and twice continuously differentiable with respect to $x.$ 

\begin{assume}\label{assume:h}
	The function $h:\mathcal{D}\to\mathbb{R}$ is continuously differentiable. The function 
	\begin{equation}
	\label{eqn:1st_order_f}
	f(t,x):=\mathbb{E}_x^{\mathbb{P}}
	[M_t(h/\phi)(X_t)]=\mathbb{E}_x^{\hat{\mathbb{P}}}
	[(h/\phi)(X_t)]
	\end{equation}   is  $C^{1,2}$  and  converges to a nonzero constant as $t\to\infty$ for each $x\in\mathcal{D}.$  This function $f$ is referred to as the {\em remainder function}.  
\end{assume}
\noindent A sufficient condition for the function $f$ to be $C^{1,2}$
is given at \cite[Theorem 1]{heath2000martingales}.
One can also find a sufficient condition for the convergence of $f(t,x)$ to a nonzero constant as $t\to\infty$ from \cite[Theorem 5.2]{qin2016positive}.

For a given quadruple of functions $(b,\sigma,r,h)$ satisfying   Assumptions \ref{assume:a}--\ref{assume:h}
on the consistent probability space $(\Omega,\mathcal{F},(\mathcal{F}_t)_{t\ge0},(\mathbb{P}_t)_{t\ge0})$ with  Brownian motion $(B_t)_{t\ge0},$
we have constructed the process $X,$ the pricing operator $\mathcal{P}$, the recurrent eigenpair $(\lambda,\phi),$ 
the martingale $M,$ the remainder function $f,$ 
the consistent family  $(\hat{\mathbb{P}}_t)_{t\ge0}$ of recurrent eigen-measures, and the Brownian motion $(\hat{B}_t)_{t\ge0}.$
Hereafter,   these objects
$$X,\,\mathcal{P},\,(\lambda,\phi),\,M,\,f,\,(\hat{\mathbb{P}}_t)_{t\ge0},\,(\hat{B}_t)_{t\ge0}$$
appear frequently and the notations are self-explanatory.

Under these assumptions, the {\em Hansen--Scheinkman decomposition} is a very useful tool for large-time analysis.
From Eq.\eqref{eqn:M},
the discount factor $e^{-\int_0^T r(X_t)\,dt}$ can be written as  
$$e^{-\int_0^T r(X_s)\,ds}=M_T e^{-\lambda T}\frac{\phi(\xi)}{\phi(X_T)}\,.$$
This expression is referred to as the Hansen--Scheinkman decomposition. 
The expectation $p_T$ satisfies
\begin{equation}  \label{eqn:HS_tranform}
\begin{aligned}
p_T=\mathbb{E}_\xi^{\mathbb{P}}[e^{-\int_0^T r(X_s)\,ds} h(X_T)]
&=\phi(\xi)\,e^{-\lambda T}\,\mathbb{E}_\xi^{\mathbb{P}}
[M_T(h/\phi)(X_T)]\\
&=\phi(\xi)\,e^{-\lambda T}\,
\mathbb{E}_\xi^{\hat{\mathbb{P}}}
[(h/\phi)(X_T)]\\
&=\phi(\xi)e^{-\lambda T}
f(T,\xi)\,.
\end{aligned}
\end{equation}
Because $f(T,\xi)$ converges to a  nonzero constant as $T\to\infty,$ we obtain the inequality
\begin{equation}
\label{eqn:decay}
|p_T|\leq c e^{-\lambda T}  
\end{equation}  
for some positive constant $c,$ which is independent of $T.$
This implies that the large-time behavior of $p_T$ is governed 
by the recurrent eigenvalue.

\section{Long-term sensitivity analysis}
\label{sec:long_term_sen_analy}

This section develops the Hansen--Scheinkman decomposition presented above to investigate the large-time behavior of the sensitivities.

\subsection{Sensitivity of the initial value}

We begin with the initial-value sensitivity, i.e., the extent to which the expectation \begin{equation}\label{eqn:main_p_T}
	p_T:=\mathbb{E}_\xi^{\mathbb{P}}[e^{-\int_0^T r(X_s)\,ds} h(X_T)]  
\end{equation}
is
affected by small changes of the initial value $\xi=X_0$ of the underlying Markov diffusion.
The first-order and second-order sensitivities
with respect to the initial value are demonstrated below.

\subsubsection{First-order sensitivity}
\label{sec:1st_sen}

In this section, we develop the first-order initial-value sensitivity for large time $T.$
For a given 
quadruple of functions $(b,\sigma,r,h)$ satisfying Assumptions \ref{assume:a}--\ref{assume:h}, the large-time asymptotic behavior of the partial derivative 
$$\partial_\xi p_T=\partial_\xi \mathbb{E}_\xi^{\mathbb{P}}[e^{-\int_0^T r(X_s)\,ds} h(X_T)]$$
is of interest to us.
For notational simplicity, we define
$$\kappa:=b+\sigma^2\phi'/\phi\,.$$
If
$(\kappa+\sigma'\sigma,\sigma,-\kappa',(h/\phi)')$
also satisfies Assumptions \ref{assume:a}--\ref{assume:h} on the consistent probability space 
$(\Omega,\mathcal{F},(\mathcal{F}_t)_{t\ge0},(\hat{\mathbb{P}}_t)_{t\ge0})$ with Brownian motion $(\hat{B}_t)_{t\ge0},$
we can construct the corresponding objects $$\hat{X},\,\hat{\mathcal{P}},\,(\hat{\lambda},\hat{\phi}),\,\hat{M},\,\hat{f},\,(\tilde{\mathbb{P}}_t)_{t\ge0}\,,(\tilde{B}_t)_{t\ge0}$$
and these notations are self-explanatory.
For example, $\hat{X}$ satisfies
\begin{equation}
	\label{eqn:hat_X}
	d\hat{X}_t=(\kappa+\sigma'\sigma)(\hat{X}_t)\,dt+\sigma(\hat{X}_t)\,d\hat{B}_t\,.
\end{equation} 
These objects will be used in the statement and the proof of the following theorem.

\begin{thm}\label{thm:initial_main_thm}
	Suppose that the partial derivative $f_{tx}$ exists and  the following conditions hold.  
	\begin{enumerate}[(i)]
		\item The quadruple of functions $(b,\sigma,r,h)$ satisfies Assumptions \ref{assume:a}--\ref{assume:h} on the consistent probability space $(\Omega,\mathcal{F},(\mathcal{F}_t)_{t\ge0},(\mathbb{P}_t)_{t\ge0})$ with  Brownian motion $(B_t)_{t\ge0}.$
		\item The quadruple of functions  $(\kappa+\sigma'\sigma,\sigma,-\kappa',(h/\phi)')$ satisfies Assumptions \ref{assume:a}--\ref{assume:h}
		on the consistent probability space $(\Omega,\mathcal{F},(\mathcal{F}_t)_{t\ge0},(\hat{\mathbb{P}}_t)_{t\ge0})$ with Brownian motion $(\hat{B}_t)_{t\ge0}.$
	\end{enumerate}  
	Then,  for each $T>0$, 
	the process $(f_x(T-t,\hat{X}_t)e^{\int_0^t \kappa'(\hat{X}_s)\,ds})_{0\le t\le T}$ is a local martingale under the probability measure  $\hat{\mathbb{P}}_T.$  
	For each $T>0$, if   this process is a martingale under the probability measure $\hat{\mathbb{P}}_T$   or if  
	the function $f_x$ satisfies
	\begin{equation}
		\label{eqn:f_x_exp_repn}
		f_x(T,x)=\mathbb{E}^{\hat{\mathbb{P}}}[ e^{\int_0^T \kappa'(\hat{X}_s)\,ds}(h/\phi)'(\hat{X}_T)|\hat{X}_0=x]\,,
	\end{equation}
	then   
	\begin{equation}
		\label{eqn:1st_sen_conv_rate}
		\left|\frac{\partial_\xi p_T}{p_T}
		-\frac{\phi'(\xi)}{\phi(\xi)}\right|\leq  ce^{-\hat{\lambda} T}\,,\;T\ge0
	\end{equation} 
	for some positive constant $c,$ which is dependent of $\xi$ but independent of $T.$
\end{thm}

\begin{proof}   
	Applying the Feynman--Kac formula 
	to Eq.\eqref{eqn:1st_order_f},
	it follows that
	\begin{equation}
		\label{eqn:1st_initial_ode}
		-f_t+\frac{1}{2}\sigma^2(x)f_{xx}+\kappa(x)f_x=0\,,\;f(0,x)=(h/\phi)(x)\,.
	\end{equation}
	Observe that  $f$ is $C^{1,2}$ from Assumption \ref{assume:h}.
	Since every coefficient is continuously differentiable in $x$ in the above PDE, the partial derivatives $f_x$ is also in $C^{1,2}$ by \cite[Section 3.5]{friedman2008partial}.  
	Taking the differentiation in $x,$ we get
	\begin{equation}
		\label{eqn:f_x_PDE}
		-f_{xt}+\frac{1}{2}\sigma^2(x)f_{xxx}+(\kappa+\sigma'\sigma)(x)f_{xx}+\kappa'(x)f_x=0\,,\;f_x(0,x)=(h/\phi)'(x)\,.
	\end{equation}

	Meanwhile, since $(\kappa+\sigma'\sigma,\sigma,-\kappa',(h/\phi)')$ satisfies Assumptions \ref{assume:a}--\ref{assume:h} on the consistent probability space
	$(\Omega,\mathcal{F},(\mathcal{F}_t)_{t\ge0},(\hat{\mathbb{P}}_t)_{t\ge0})$ having Brownian motion $(\hat{B}_t)_{t\ge0},$
	we can construct the corresponding objects $$\hat{X},\,\hat{\mathcal{P}},\,(\hat{\lambda},\hat{\phi}),\,\hat{M},\,\hat{f},\,(\tilde{\mathbb{P}}_t)_{t\ge0}\,,(\tilde{B}_t)_{t\ge0}$$
	and these notations are self-explanatory.
	Since the process $\hat{X}$ satisfies
	$$d\hat{X}_t=(\kappa+\sigma'\sigma)(\hat{X}_t)\,dt+\sigma(\hat{X}_t)\,d\hat{B}_t\,,$$
	by applying the Ito formula to $f_x(T-t,\hat{X}_t)e^{\int_0^t \kappa'(X_s)\,ds},$ we have
	\begin{equation}\label{eqn:1st_order_Ito_formula_f_x}
		\begin{aligned}
			&d\big(f_x(T-t,\hat{X}_t)e^{\int_0^t \kappa'(\hat{X}_s)\,ds}\big)\\
			=&\;e^{\int_0^t \kappa'(\hat{X}_s)\,ds}\Big(-f_{xt}+\frac{1}{2}\sigma^2(\hat{X}_t)f_{xxx}+(\kappa+\sigma'\sigma)(\hat{X}_t)f_{xx}+\kappa'(\hat{X}_s)f_x\Big)\,dt
			+e^{\int_0^t \kappa'(\hat{X}_s)\,ds}\sigma(\hat{X}_t) f_{xx}\,d\hat{B}_t\\
			=&\;e^{\int_0^t \kappa'(\hat{X}_s)\,ds}\sigma(\hat{X}_t) f_{xx}\,d\hat{B}_t.
		\end{aligned}
	\end{equation}
	Thus, for each $T>0$, the  process  $(f_x(T-t,\hat{X}_t)e^{\int_0^t \kappa'(\hat{X}_s)\,ds})_{0\le t\le T}$ is a  local martingale  under the probability measure $\hat{\mathbb{P}}_T.$

	Suppose that $f_x$ satisfies
	\begin{equation}
	\label{eqn:1st_order_f_x}
	f_x(T,x)=\mathbb{E}^{\hat{\mathbb{P}}}[ e^{\int_0^T \kappa'(\hat{X}_s)\,ds}(h/\phi)'(\hat{X}_T)|\hat{X}_0=x]\,.
	\end{equation} 
	Note that this equality holds if   $(f_x(T-t,\hat{X}_t)e^{\int_0^t \kappa'(\hat{X}_s)\,ds})_{0\le t\le T}$ is a martingale since
	\begin{equation}\label{eqn:f_x_prob_repn}
		f_x(T,x)=\mathbb{E}^{\hat{\mathbb{P}}}[ e^{\int_0^T \kappa'(\hat{X}_s)\,ds}f_x(0,\hat{X}_T)|\hat{X}_0=x]=\mathbb{E}^{\hat{\mathbb{P}}}[ e^{\int_0^T \kappa'(\hat{X}_s)\,ds}(h/\phi)'(\hat{X}_T)|\hat{X}_0=x]\,.
	\end{equation}
	Replacing $T$ by $t,$ 
	it follows that
	$$f_x(t,x)=\mathbb{E}^{\hat{\mathbb{P}}}\big[e^{\int_0^t\kappa'(\hat{X}_s)\,ds}(h/\phi)'(\hat{X}_t)\big|\hat{X}_0=x\big]=\hat{\mathcal{P}}_t(h/\phi)'(x)\,.$$
	Now, we apply the Hansen--Scheinkman decomposition here.
	Since $(\hat{\lambda},\hat{\phi})$ is the recurrent eigenpair  and $$\hat{f}(t,x)=\mathbb{E}_x^{\hat{\mathbb{P}}}[\hat{M}_t(h/\phi)'(\hat{X}_t)]=\mathbb{E}_x^{\tilde{\mathbb{P}}}[((h/\phi)'/\hat{\phi})(\hat{X}_t)]$$ is the remainder function,
 we have
	\begin{equation}
		\label{eqn:f_x_general}
		\begin{aligned}
			f_x(t,x)
			&=\hat{f}(t,x)e^{-\hat{\lambda} t}\hat{\phi}(x)\,.
		\end{aligned}
	\end{equation}
	Since $f(t,x)$ and $\hat{f}(t,x)$ converge to nonzero constants as $t\to\infty,$  Eq.\eqref{eqn:HS_tranform} 
	implies that
	\begin{equation}
		\label{eqn:first_order_sensi}
		\left|\frac{\partial_\xi p_T}{p_T}
		-\frac{\phi'(\xi)}{\phi(\xi)}\right|=\left|\frac{f_x(T,\xi)}{f(T,\xi)}\right|=\left|\frac{\hat{f}(T,\xi)}{f(T,\xi)}\right|e^{-\hat{\lambda} T}\hat{\phi}(\xi)\leq ce^{-\hat{\lambda} T}
	\end{equation} 
	for some positive constant $c,$ which is dependent of $\xi$ but independent of $T.$	
\end{proof}

\begin{remark} Theorem \ref{thm:initial_main_thm} has an interesting implication. 
By Assumption \ref{assume:HS}, the recurrent eigenvalue $\hat{\lambda}$ is implicitly assumed to be positive (and this holds if $\kappa'\le 0$ and $\kappa'\not\equiv0$ by Remark \ref{rem:HS}).		
	Eq.\eqref{eqn:1st_sen_conv_rate} says  $$\lim_{T\to\infty}\frac{\partial_\xi p_T}{p_T}=\frac{\phi'(\xi)}{\phi(\xi)}\,.$$
	This means that the large-time behavior of the sensitivity  $\frac{\partial_\xi p_T}{p_T}$ as $T\to\infty$ is expressed in terms of the recurrent eigenfunction $\phi$  induced by $(b,\sigma,r,h).$ However, as one can observe from Eq.\eqref{eqn:1st_sen_conv_rate},  its exponential convergence rate  to the limit is determined
	by the recurrent eigenvalue $\hat{\lambda}$ induced  by $(\kappa+\sigma'\sigma,\sigma,-\kappa',(h/\phi)').$ 
\end{remark}

\begin{remark}
The methodology used in this paper is not applicable for
multi-dimensional cases.
The main idea of the methodology is 
the PDE representation in  Eq.\eqref{eqn:f_x_PDE}.
Since the Feynman-Kac formula can be used to this form of PDE, we derived Eq.\eqref{eqn:1st_order_f_x} and then applied the Hansen--Scheinkman decomposition.
If the underling process $X$ is multi-dimensional, 
we can also derive a PDE corresponding to   Eq.\eqref{eqn:f_x_PDE} by following the same procedure, however the Feynman-Kac formula cannot be used to this form of PDE. As a result, our analysis is not applicable if the underling process $X$ is multi-dimensional.
\end{remark}

\begin{remark}	Under the same hypothesis of Theorem \ref{thm:initial_main_thm}, 
	the probabilistic representation 
	$$	f_x(T,x)=\mathbb{E}^{\hat{\mathbb{P}}}[ e^{\int_0^T \kappa'(\hat{X}_s)\,ds}(h/\phi)'(\hat{X}_T)|\hat{X}_0=x]$$
	stated in Eq.\eqref{eqn:f_x_exp_repn} holds  if
	$$\mathbb{E}_\xi^{\hat{\mathbb{P}}}\Big[\int_0^Te^{2\int_0^t \kappa'(\hat{X}_s)\,ds}\sigma^2(\hat{X}_t) f_{xx}^2(T-t,\hat{X}_t)  \,dt\Big]<\infty\,.$$ 
	This is evident from Eq.\eqref{eqn:1st_order_Ito_formula_f_x} since 	
	the  process  $(f_x(T-t,\hat{X}_t)e^{\int_0^t \kappa'(\hat{X}_s)\,ds})_{0\le t\le T}$ is a  martingale  under the probability measure $\hat{\mathbb{P}}_T.$
\end{remark}

Recall  the hypothesis in Theorem \ref{thm:initial_main_thm}:
the process $(f_x(T-t,\hat{X}_t)e^{\int_0^t \kappa'(\hat{X}_s)\,ds})_{0\le t\le T}$ is a  $\mathbb{P}_T$-martingale  or  	the function $f_x$ satisfies Eq.\eqref{eqn:f_x_exp_repn}.
It can be easily shown that these two statements are equivalent.
In the remainder of this section,
we investigate a sufficient condition  to guarantee that this hypothesis holds, i.e., the function $f_x$ satisfies 
\begin{equation}
	\label{eqn:delta_condi_mart}
	f_x(T,x)=\mathbb{E}^{\hat{\mathbb{P}}}[ e^{\int_0^T \kappa'(\hat{X}_s)\,ds}(h/\phi)'(\hat{X}_T)|\hat{X}_0=x]
\end{equation}
presented in Eq.\eqref{eqn:f_x_exp_repn}.
The sufficient condition is based on the Feynman--Kac formula and is given in Proposition \ref{prop:1st_initial_sen_linear}.

Before stating  Proposition \ref{prop:1st_initial_sen_linear}, we present a slight modification of the standard Feynman--Kac formula
\cite[Theorem 5.7.6 on page 366]{karatzas1991brownian}
in  the remark below. 
Recall that $\mathcal{D}$ is an open interval in $\mathbb{R}.$
We say that a function $f$  on $\mathcal{D}$ has polynomial growth (respectively, linear growth) if there is a constant $C>0$ and  $m\in\mathbb{N}$ (respectively, $m=1$) such that for all $x\in\mathcal{D}$,
$$|f(x)|\leq C (1+|x|^m)\,.$$

\begin{remark}\label{remark:FK}
	(Feynman--Kac formula)  Consider a quadruple of functions $(b,\sigma,r,h)$ defined on an open interval $\mathcal{D}\subseteq\mathbb{R}.$ Suppose that the following conditions hold.
	\begin{enumerate}[(i)]
		\item Two functions $b$ and $\sigma$ are continuous and have linear growth.
		\item For each $x\in\mathcal{D},$ the stochastic differential equation (SDE)  
		$$dX_t=b(X_t)\,dt+\sigma(X_t)\,dB_t\,,\;X_0=x$$	
		has a unique strong solution on $\mathcal{D},$ i.e., $\mathbb{P}_T(X_t\in\mathcal{D} \textnormal{ for all } 0\le t\le T)=1$ for each $T\ge0.$
		\item Three functions $r,$  $h$ and $g$   are continuous. 
		In addition,  $r$ is bounded below,  $h$ has polynomial growth or is nonnegative, and $\max_{0\le t\le T}|g(t,\cdot)|$ has polynomial growth  or $g$ is nonnegative. 
	\end{enumerate}	    
	If $f(t,x)$ is $C^{1,2}$ and satisfies
	\begin{equation}
		\label{eqn:PDE_FK_usual}
		-f_t(t,x)+\frac{1}{2}\sigma^2(x)f_{xx}(t,x)+b(x)f_x(t,x)-r(x)f(t,x)+g(t,x)=0\,,\;f(0,x)=h(x)
	\end{equation} 
	as well as the polynomial growth condition
	$$\max_{0\le t\le T}|f(t,x)|\leq C(1+|x|^m)$$
	for some $C>0$ and $m\ge1,$ then
	\begin{equation}
		\label{eqn:FK_usual_form}
		f(t,x)=	\mathbb{E}\Big[e^{-\int_0^tr(X_s)\,ds}h(X_t)+\int_0^tg(T-s,X_s)e^{-\int_0^sr(X_u)\,du}\,ds\Big|X_0=x\Big]
	\end{equation} 
	for $x\in\mathcal{D}$ and $0\le t\le T.$
\end{remark}

The Feynman--Kac formula stated above differs from the standard statement \cite[Theorem 5.7.6 on page 366]{karatzas1991brownian}  in two ways.
First, the domain in this case is an open interval $\mathcal{D}$, whereas the domain is the whole real line $\mathbb{R}$ in the case of the standard statement. This is easily verified 
because the proof of the standard statement  can be directly applied.
Second, the time-derivative term $-f_t$  in the PDE \eqref{eqn:PDE_FK_usual} has a negative sign. 
The PDE in the standard statement is expressed in the time-reverse order, i.e., the final-time condition $f(T,\cdot)$ is given.  
However, in this paper, we are interested in the time order, i.e., the initial-time condition $f(0,\cdot)$ is given.
The time-reverse order can be easily changed to the time order by using the Markov property.
For fixed $T>0,$ define $F(t,x)=f(T-t,x).$ Then, Eq.\eqref{eqn:PDE_FK_usual} becomes
$$F_t(t,x)+\frac{1}{2}\sigma^2(x)F_{xx}(t,x)+b(x)F_x(t,x)-r(x)F(t,x)+g(T-t,x)=0\,,\;F(T,x)=h(x).$$
From the standard  Feynman--Kac formula, we know that
\begin{equation}
	\begin{aligned}
		F(t,x)
		&=\mathbb{E}\Big[e^{-\int_t^Tr(X_s)\,ds}h(X_T)+\int_t^Tg(T-s,X_s)e^{-\int_t^sr(X_u)\,du}\,ds\Big|X_t=x\Big]. 
	\end{aligned}
\end{equation}
Thus,
\begin{equation}
	\begin{aligned}
		f(T,x)=F(0,x)=
		\mathbb{E}\Big[e^{-\int_0^Tr(X_s)\,ds}h(X_T)+\int_0^Tg(T-s,X_s)e^{-\int_0^sr(X_u)\,du}\,ds\Big|X_0=x\Big]. 
	\end{aligned}
\end{equation}
Replacing $T$ by $t,$ we obtain Eq.\eqref{eqn:FK_usual_form}.


The Feynman--Kac formula in the above-mentioned remark occasionally cannot be applied to obtain the representation
$$f_x(T,x)=\mathbb{E}^{\hat{\mathbb{P}}}[ e^{\int_0^T \kappa'(\hat{X}_s)\,ds}(h/\phi)'(\hat{X}_T)|\hat{X}_0=x]$$
because it requires the function $\max_{0\le t\le T}|f_x(t,x)|$   to have polynomial growth in $x.$ Some financial models do not satisfy this condition (e.g., the CIR model in Appendix \ref{app:CIR_model}). 
One way to overcome this problem is to consider 
$f_x(t,x)\phi(x)$ instead of $f_x(t,x).$
The main aspect of Proposition \ref{prop:1st_initial_sen_linear}
is that the   polynomial growth condition on $\max_{0\le t\le T}|f_x(t,x)|$ can be replaced by  the   polynomial growth condition on $\max_{0\le t\le T}|f_x(t,x)|\phi(x).$
Indeed, in the CIR model, the function $\max_{0\le t\le T}|f_x(t,x)|\phi(x)$
has polynomial growth in $x,$ whereas $\max_{0\le t\le T}|f_x(t,x)|$ does not.

\begin{prop}\label{prop:1st_initial_sen_linear}
	Suppose that the following conditions hold.  
	\begin{enumerate}[(i)]
		\item The quadruple of functions $(b,\sigma,r,h)$ satisfies Assumptions \ref{assume:a}--\ref{assume:h} on the consistent probability space $(\Omega,\mathcal{F},(\mathcal{F}_t)_{t\ge0},(\mathbb{P}_t)_{t\ge0})$ with  Brownian motion $(B_t)_{t\ge0}.$
		\item The quadruple of functions  $(\kappa+\sigma'\sigma,\sigma,-\kappa',(h/\phi)')$ satisfies Assumptions \ref{assume:a}--\ref{assume:h}
		on the consistent probability space $(\Omega,\mathcal{F},(\mathcal{F}_t)_{t\ge0},(\hat{\mathbb{P}}_t)_{t\ge0})$ with Brownian motion $(\hat{B}_t)_{t\ge0}.$
	\end{enumerate}  		
	Furthermore, assume the following conditions for given $T>0.$
	\begin{enumerate}[(i)]
		\item Two functions $\kappa+\sigma'\sigma-\sigma^2\phi'/\phi$ and $\sigma$ have linear growth.
		\item The function   ${(\hat{\mathcal{L}}\phi)}/{\phi}-\kappa'$ is bounded below, where
		\begin{equation}
		\label{eqn:hat_L}
		\hat{\mathcal{L}}=\frac{1}{2}\sigma^2(x)\partial_{xx}+(\kappa+\sigma'\sigma-\sigma^2\phi'/\phi)(x)\partial_x\,.
		\end{equation}
		\item The function $h'-h\phi'/\phi$ has polynomial growth or is nonnegative.
		\item The function $\max_{0\le t\le T}|f_x(t,x)|\phi(x)$
		has polynomial growth in $x.$ 
		\item A local martingale 
		\begin{equation}
		\label{eqn:f}
		\frac{\phi(\hat{X}_0)}{\phi(\hat{X}_t)}e^{\int_0^t\frac{\hat{\mathcal{L}}\phi(\hat{X}_s)}{\phi(\hat{X}_s)}\,ds}\,,\;0\le t\le T
		\end{equation} 	
		is a martingale under the probability measure $\hat{\mathbb{P}}_T.$
	\end{enumerate}
	Then,  the process $(f_x(T-t,\hat{X}_t)e^{\int_0^t \kappa'(\hat{X}_s)\,ds})_{0\le t\le T}$ is a martingale  under the probability measure $\hat{\mathbb{P}}_T.$ In particular, 
	$$f_x(T,x)=\mathbb{E}^{\hat{\mathbb{P}}}[ e^{\int_0^T \kappa'(\hat{X}_s)\,ds}(h/\phi)'(\hat{X}_T)|\hat{X}_0=x]\,.$$
\end{prop}

\begin{proof}
	Recall that 
	$(\Omega,\mathcal{F},(\mathcal{F}_t)_{t\ge0},(\hat{\mathbb{P}}_t)_{t\ge0})$ is a consistent probability space having  Brownian motion $(\hat{B}_t)_{t\ge0}$
	and the process $(\hat{X}_t)_{t\ge0}$ satisfies Eq.\eqref{eqn:hat_X}.
	One can show that 
	$$\frac{\phi(\hat{X}_0)}{\phi(\hat{X}_t)}e^{\int_0^t\frac{\hat{\mathcal{L}}\phi(\hat{X}_s)}{\phi(\hat{X}_s)}\,ds}\,,\;0\le t\le T$$
	is a $\hat{\mathbb{P}}_T$-local martingale by applying the Ito formula and checking that the $dt$-term vanishes.
	Since this process is assumed to be a martingale,
	we can define a new measure $\mathbb{Q}_T$ on $\mathcal{F}_T$ as
	\begin{equation}
		\label{eqn:FK_RN}
		\frac{d\mathbb{Q}_T}{d\hat{\mathbb{P}}_T}=\frac{\phi(\hat{X}_0)}{\phi(\hat{X}_T)}e^{\int_0^T\frac{\hat{\mathcal{L}}\phi(\hat{X}_s)}{\phi(\hat{X}_s)}\,ds}\,.
	\end{equation} 
	It is easy to check that the family $(\mathbb{Q}_t)_{t>0}$ 
	of probability measures is consistent and 
	the process 
	$$B_t^\mathbb{Q}:=\hat{B}_t+(\sigma \phi'/\phi)(\hat{X}_t)\,,\;t\ge0$$
	is a Brownian motion on the consistent probability space $(\Omega,\mathcal{F},(\mathcal{F}_t)_{t\ge0},(\mathbb{Q}_t)_{t\ge0}).$ 
	The process  $\hat{X}$ satisfies
	\begin{equation}
		\begin{aligned}
			d\hat{X}_t
			&=(\kappa+\sigma'\sigma)(\hat{X}_t)\,dt+\sigma(\hat{X}_t)\,d\hat{B}_t\\
			&=(\kappa+\sigma'\sigma-\sigma^2\phi'/\phi)(\hat{X}_t)\,dt+\sigma(\hat{X}_t)\,d{B}_t^{\mathbb{Q}} \,.
		\end{aligned}
	\end{equation}
	Note that 
	$\hat{\mathcal{L}}$
	is the generator of $\hat{X}$ under the family of probability measures $(\mathbb{Q}_t)_{t\ge0}.$

	Define
	$$g(t,x):=f_x(t,x)\phi(x)\,.$$
	Then, Eq.\eqref{eqn:f_x_PDE} gives
	\begin{equation}
		\begin{aligned}
			&-g_t+\frac{1}{2}\sigma^2(x)g_{xx}+(\kappa+\sigma'\sigma-\sigma^2\phi'/\phi)(x)g_x+(\kappa'-{(\hat{\mathcal{L}}\phi)}/{\phi})(x)g=0\\
			&g(0,x)=(h/\phi)'(x)\phi(x)=(h'-h\phi'/\phi)(x)\,.
		\end{aligned}
	\end{equation}
	The Feynman--Kac formula (Remark \ref{remark:FK})
	states that
	\begin{equation}
		\begin{aligned}
			f_x(t,x)\phi(x)=g(t,x)&=\mathbb{E}^{{\mathbb{Q}}}\big[
			e^{\int_0^t(\kappa'-\frac{\hat{\mathcal{L}}\phi}{\phi})(\hat{X}_s)\,ds}\phi(\hat{X}_{t})(h/\phi)'(\hat{X}_t)\big|\hat{X}_0=x\big]\\
			&=\mathbb{E}^{{\mathbb{Q}}}\Big[
			e^{-\int_0^t\frac{\hat{\mathcal{L}}\phi}{\phi}(\hat{X}_s)\,ds}\,\frac{\phi(\hat{X}_{t})}{\phi(\hat{X}_{0})}
			e^{\int_0^t\kappa'(\hat{X}_s)\,ds}(h/\phi)'(\hat{X}_t)\Big|\hat{X}_0=x\Big]\phi(x)\\
			&=\mathbb{E}^{{\mathbb{Q}}}\Big[
			\frac{d\hat{\mathbb{P}}_t}{d\mathbb{Q}_t}
			e^{\int_0^t\kappa'(\hat{X}_s)\,ds}(h/\phi)'(\hat{X}_t)\Big|\hat{X}_0=x\Big]\phi(x)\\
			&=\mathbb{E}^{\hat{\mathbb{P}}}\big[
			e^{\int_0^t\kappa'(\hat{X}_s)\,ds}(h/\phi)'(\hat{X}_t)\big|\hat{X}_0=x\big]\phi(x)\,,
		\end{aligned}
	\end{equation}
	which implies that
	$$f_x(t,x)=\mathbb{E}^{\hat{\mathbb{P}}}\big[e^{\int_0^t\kappa'(\hat{X}_s)\,ds}(h/\phi)'(\hat{X}_t)\big|\hat{X}_0=x\big]\,.$$  
	From the time-homogeneous Markov property, 
	\begin{equation}
		\begin{aligned}
			f_x(T-t,x)
			&=\mathbb{E}^{\hat{\mathbb{P}}}\big[e^{\int_0^{T-t}\kappa'(\hat{X}_s)\,ds}(h/\phi)'(\hat{X}_{T-t})\big|\hat{X}_0=x\big]\\
			&=\mathbb{E}^{\hat{\mathbb{P}}}\big[e^{\int_t^T\kappa'(\hat{X}_s)\,ds}(h/\phi)'(\hat{X}_T)\big|\hat{X}_t=x\big]
		\end{aligned}
	\end{equation}
	so that
	$$f_x(T-t,\hat{X}_t)e^{\int_0^t \kappa'(\hat{X}_s)\,ds}
	=\mathbb{E}^{\hat{\mathbb{P}}}\big[e^{\int_0^T\kappa'(\hat{X}_s)\,ds}(h/\phi)'(\hat{X}_T)\big|\mathcal{F}_t\big]\,.$$
	In conclusion, the process
	$(f_x(T-t,\hat{X}_t)e^{\int_0^t \kappa'(\hat{X}_s)\,ds})_{0\le t\le T}$ is a martingale.
\end{proof}

\begin{remark}\label{rem:mart}
	For condition (v) in the above proposition,
	Theorem 5.1.8 of \cite{pinsky1995positive} provides a useful criterion  to check whether Eq.\eqref{eqn:f} is 
	a martingale or not.
This process is a martingale if and only if
\begin{equation} 
\begin{aligned}
&\int_a^{x_0} \frac{1}{\sigma^2(x)}e^{-\int_{x_0}^x\frac{2\kappa(s)}{\sigma^2(s)}\,ds}\int_x^{x_0} e^{\int_{x_0}^y\frac{2\kappa(s)}{\sigma^2(s)}\,ds}\,dy\,dx =\infty\,,\\
&\int_{x_0}^b \frac{1}{\sigma^2(x)}e^{-\int_{x_0}^x\frac{2\kappa(s)}{\sigma^2(s)}\,ds}\int_{x_0}^{x} e^{\int_{x_0}^y\frac{2\kappa(s)}{\sigma^2(s)}\,ds}\,dy\,dx =\infty\,,
\end{aligned}
\end{equation}
where  $\mathcal{D}=(a,b)$ and $x_0$ is any reference point in $(a,b).$ 
\end{remark}

The proof is as follows. 
	Define a operator $\overline{\mathcal{L}}$ as
	\begin{equation}
	\begin{aligned}
	\overline{\mathcal{L}}h(x)
	&=\hat{\mathcal{L}}h(x)-\frac{\hat{\mathcal{L}}\phi(x)}{\phi(x)}h(x)\\
	&=\frac{1}{2}\sigma^2(x)h''(x)+(\kappa+\sigma'\sigma-\sigma^2\phi'/\phi)(x)h'(x)-\frac{\hat{\mathcal{L}}\phi(x)}{\phi(x)}h(x)\,,\;h\in C^2(\mathcal{D}) 
	\end{aligned}
	\end{equation}
	for $\hat{\mathcal{L}}$ in Eq.\eqref{eqn:hat_L}. It is clear that $\overline{\mathcal{L}}\phi=0.$
	By  Theorem 5.1.8 of \cite{pinsky1995positive}, Eq.\eqref{eqn:f} is a martingale  under the probability measure $\hat{\mathbb{P}}_T$ if and only if
	\begin{equation} 
	\begin{aligned}
	&\int_a^{x_0} \frac{1}{\phi^2(x)}e^{-\int_{x_0}^x\frac{2}{\sigma^2(s)}(\kappa+\sigma'\sigma-\sigma^2\frac{\phi'}{\phi})(s)\,ds}\int_x^{x_0}\frac{\phi^2(y)}{\sigma^2(y)} e^{\int_{x_0}^y\frac{2}{\sigma^2(s)}(\kappa+\sigma'\sigma-\sigma^2\frac{\phi'}{\phi})(s)\,ds}\,dy\,dx =\infty\,,\\
	&\int_{x_0}^b \frac{1}{\phi^2(x)}e^{-\int_{x_0}^x\frac{2}{\sigma^2(s)}(\kappa+\sigma'\sigma-\sigma^2\frac{\phi'}{\phi})(s)\,ds}\int_{x_0}^{x} \frac{\phi^2(y)}{\sigma^2(y)} e^{\int_{x_0}^y\frac{2}{\sigma^2(s)}(\kappa+\sigma'\sigma-\sigma^2\frac{\phi'}{\phi})(s)\,ds}\,dy\,dx =\infty\,,
	\end{aligned}
	\end{equation}	
	where  $\mathcal{D}=(a,b)$ and $x_0$ is any reference point in $(a,b).$ 
	By simplifying these equations, we obtain the desired result.

\subsubsection{Second-order sensitivity}
\label{sec:2nd_sen}

In this section, we develop the second-order initial-value sensitivity for large-time $T.$
For a given 
quadruple of functions $(b,\sigma,r,h)$ satisfying Assumptions \ref{assume:a}--\ref{assume:h}, the large-time asymptotic behavior of the partial derivative 
$$\partial_{\xi\xi} p_T=\partial_{\xi\xi} \mathbb{E}_{\xi}^{\mathbb{P}}[e^{-\int_0^T r(X_s)\,ds} h(X_T)]$$
is of interest to us.

Toward this end, we need to study two   quadruples in addition to the original quadruple $(b,\sigma,r,h).$
Recall that $\kappa=b+\sigma^2\phi'/\phi,$ and suppose that 
$(\kappa+\sigma'\sigma,\sigma,-\kappa',(h/\phi)')$
satisfies Assumptions \ref{assume:a}--\ref{assume:h}
on the consistent probability space $(\Omega,\mathcal{F},(\mathcal{F}_t)_{t\ge0},(\hat{\mathbb{P}}_t)_{t\ge0})$ having  Brownian motion $(\hat{B}_t)_{t\ge0}$
as stated in Theorem \ref{thm:initial_main_thm}.
Then,
we can construct the corresponding objects
$$\hat{X},\,\hat{\mathcal{P}},\,(\hat{\lambda},\hat{\phi}),\,\hat{M},\,\hat{f},\,(\tilde{\mathbb{P}}_t)_{t\ge0}\,,(\tilde{B}_t)_{t\ge0}\,.$$
Define
$$\gamma:=\kappa+\sigma'\sigma+\sigma^2\hat{\phi}'/\hat{\phi}\,.$$

If $(\gamma+\sigma'\sigma,\sigma,-\gamma',((h/\phi)'/\hat{\phi})')$ also satisfies Assumptions \ref{assume:a} -- \ref{assume:h}
on the consistent probability space $(\Omega,\mathcal{F},(\mathcal{F}_t)_{t\ge0},(\tilde{\mathbb{P}}_t)_{t\ge0})$ having  Brownian motion $(\tilde{B}_t)_{t\ge0},$
we can construct the corresponding objects
$$\tilde{X},\,\tilde{\mathcal{P}},\,(\tilde{\lambda},\tilde{\phi}),\,\tilde{M},\,\tilde{f},\,(\overline{\mathbb{P}}_t)_{t\ge0},\,(\overline{B}_t)_{t\ge0}\,.$$
We summarize three quadruples and their corresponding objects in the following table. 
\begin{center}
	\begin{tabular}{ | l | c |c | r | }
	\hline
	\rule{0pt}{12pt}	Quadruples & 	$(b,\sigma,r,h)$   & $(\kappa+\sigma'\sigma,\sigma,-\kappa',(h/\phi)')$    &  $(\gamma+\sigma'\sigma,\sigma,-\gamma',((h/\phi)'/\hat{\phi})')$ \\ \hline	
	\rule{0pt}{12pt} Underlying space & $(\Omega,\mathcal{F},(\mathcal{F}_t)_{t\ge0},(\mathbb{P}_t)_{t\ge0})$   & $(\Omega,\mathcal{F},(\mathcal{F}_t)_{t\ge0},(\hat{\mathbb{P}}_t)_{t\ge0})$   & $(\Omega,\mathcal{F},(\mathcal{F}_t)_{t\ge0},(\tilde{\mathbb{P}}_t)_{t\ge0})$ \\ \hline				
	\rule{0pt}{12pt} Induced objects & $X,\mathcal{P},(\lambda,\phi),M,f,(\hat{\mathbb{P}}_t)_{t\ge0}$ &  $\hat{X},\hat{\mathcal{P}},(\hat{\lambda},\hat{\phi}),\hat{M},\hat{f},(\tilde{\mathbb{P}}_t)_{t\ge0}$ &$\tilde{X},\tilde{\mathcal{P}},(\tilde{\lambda},\tilde{\phi}),\tilde{M},\tilde{f},(\overline{\mathbb{P}}_t)_{t\ge0}$  \\ \hline
	\rule{0pt}{12pt}	Auxiliary functions 	  &   & $\kappa:=b+\sigma^2\phi'/\phi$ & 	$\gamma:=\kappa+\sigma'\sigma+\sigma^2\hat{\phi}'/\hat{\phi}$  \\ \hline
	\rule{0pt}{12pt} Sections used & \ref{sec:1st_sen}, \ref{sec:2nd_sen}, \ref{sec:drift}, \ref{sec:diffusion}  & \ref{sec:1st_sen}, \ref{sec:2nd_sen}, \ref{sec:drift}, \ref{sec:diffusion} 
	&\ref{sec:2nd_sen}, \ref{sec:diffusion}   \\
	\hline
\end{tabular} 
	\captionof{table}{Three quadruples \label{Tab:Tcr}}
\end{center}

Before stating the rigorous results, 
we provide   heuristic arguments to understand the motivation for Eq.\eqref{eqn:delta_second_order_conclusion} in Theorem \ref{thm:2nd_sen}.  Suppose that  three quadruples of functions $(b,\sigma,r,h),$  $(\kappa+\sigma'\sigma,\sigma,-\kappa',(h/\phi)')$ and $(\gamma+\sigma'\sigma,\sigma,-\gamma',((h/\phi)'/\hat{\phi})')$ satisfy Assumptions \ref{assume:a}--\ref{assume:h}.
Using the Hansen--Scheinkman decomposition, we have
$p_T=e^{-\lambda T}\phi(\xi)f(T,\xi)$
and 
$f_x(t,x)
=\hat{f}(t,x)e^{-\hat{\lambda} t}\hat{\phi}(x)$ in Eq.\eqref{eqn:f_x_general}.
Direct calculation gives  
\begin{equation}
\label{eqn:second_order_deriva_ori}
\frac{{\partial_{\xi\xi}} p_T}{p_T}-\Big(\frac{\partial_\xi p_T}{p_T}\Big)^2-\frac{\phi''(\xi)}{\phi(\xi)}+\Big(\frac{\phi'(\xi)}{\phi(\xi)}\Big)^2
-\frac{\hat{\phi}'(\xi)}{\hat{\phi}(\xi)} \frac{\partial_\xi p_T}{p_T} +\frac{\hat{\phi}'(\xi)}{\hat{\phi}(\xi)} \frac{\phi'(\xi)}{\phi(\xi)}
=\frac{\hat{f}_x(T,\xi)\hat{\phi}(\xi)e^{-\hat{\lambda} T}}{f(T,\xi)}-\Big(\frac{f_x(T,\xi)}{f(T,\xi)}\Big)^2\,.
\end{equation}
We shift our attention to the two terms on the right-hand side. 
For the  first term on the right-hand side,
we derive the equality
\begin{equation} 
\begin{aligned}
\hat{f}_x(t,x)
&=\tilde{f}(t,x)e^{-\tilde{\lambda} t}\tilde{\phi}(x),
\end{aligned}
\end{equation}
which is similar to Eq.\eqref{eqn:f_x_general}.
Thus,
$$\frac{\hat{f}_x(T,\xi)\hat{\phi}(\xi)e^{-\hat{\lambda} T}}{f(T,\xi)}= \frac{\tilde{f}(T,\xi)e^{-(\tilde{\lambda}+\hat{\lambda}) T}\tilde{\phi}(\xi)\hat{\phi}(\xi)}{f(T,\xi)}\simeq  e^{-(\tilde{\lambda}+\hat{\lambda}) T}\,.$$
Here, for two nonzero functions $h_1(T)$
and $h_2(T)$, the notation $h_1(T)\simeq h_2(T)$ implies that the limit $\lim_{T\to\infty}\frac{h_1(T)}{h_2(T)}$ 
converges to
a nonzero constant.
For the second term on the right-hand side of Eq.\eqref{eqn:second_order_deriva_ori},  
observe that 
$f_x(T,\xi)\simeq e^{-\hat{\lambda} T}$
from Eq.\eqref{eqn:f_x_general}. Thus, $$\Big(\frac{f_x(T,\xi)}{f(T,\xi)}\Big)^2\simeq e^{-2\hat{\lambda} T}\,.$$
In conclusion,  Eq.\eqref{eqn:second_order_deriva_ori} satisfies
$$\frac{{\partial_{\xi\xi}} p_T}{p_T}-\Big(\frac{\partial_\xi p_T}{p_T}\Big)^2-\frac{\phi''(\xi)}{\phi(\xi)}+\Big(\frac{\phi'(\xi)}{\phi(\xi)}\Big)^2
-\frac{\hat{\phi}'(\xi)}{\hat{\phi}(\xi)} \frac{\partial_\xi p_T}{p_T} +\frac{\hat{\phi}'(\xi)}{\hat{\phi}(\xi)} \frac{\phi'(\xi)}{\phi(\xi)} \simeq (e^{-\tilde{\lambda} T}+e^{-\hat{\lambda} T})e^{-\hat{\lambda} T}\,, $$
which is the motivation for Eq.\eqref{eqn:delta_second_order_conclusion}.

A rigorous estimation of the second-order partial derivative $\partial_{\xi\xi} p_T$ for large $T$ is obtained by analyzing the three quadruples stated above, as described in the following theorem.

\begin{thm} \label{thm:2nd_sen}
	Assume that  $\sigma$ is twice continuously differentiable, and the partial derivatives $f_{tx}$ and $\hat{f}_{tx}$ exist.
	Suppose the following conditions hold.
	\begin{enumerate}[(i)]
		\item The quadruple of functions $(b,\sigma,r,h)$ satisfies Assumptions \ref{assume:a}--\ref{assume:h} on the consistent probability space $(\Omega,\mathcal{F},(\mathcal{F}_t)_{t\ge0},(\mathbb{P}_t)_{t\ge0})$ having  Brownian motion $(B_t)_{t\ge0}.$
		\item The quadruple of functions  $(\kappa+\sigma'\sigma,\sigma,-\kappa',(h/\phi)')$ satisfies Assumptions \ref{assume:a}--\ref{assume:h}
		on the consistent probability space $(\Omega,\mathcal{F},(\mathcal{F}_t)_{t\ge0},(\hat{\mathbb{P}}_t)_{t\ge0})$ having Brownian motion $(\hat{B}_t)_{t\ge0}.$
		\item The quadruple of functions   $(\gamma+\sigma'\sigma,\sigma,-\gamma',((h/\phi)'/\hat{\phi})')$ satisfies Assumptions \ref{assume:a}--\ref{assume:h}
		on the consistent probability space $(\Omega,\mathcal{F},(\mathcal{F}_t)_{t\ge0},(\tilde{\mathbb{P}}_t)_{t\ge0})$ having Brownian motion $(\tilde{B}_t)_{t\ge0}.$
	\end{enumerate} 
	Then, for each $T>0$, two processes $(f_x(T-t,\hat{X}_t)e^{\int_0^t \kappa'(\hat{X}_s)\,ds})_{0\le t\le T}$ and
	$(\hat{f}_x(T-t,\tilde{X}_t)e^{\int_0^t \gamma'(\tilde{X}_s)\,ds})_{0\le t\le T}$ are local martingales under the probability measures $\hat{\mathbb{P}}_T$ and $\tilde{\mathbb{P}}_T,$ respectively. 
	For each $T>0$, if these are martingales 
	or if
	two functions $f_x$ and $\hat{f}_x$ satisfy
	\begin{equation}
	\label{eqn:f_xx_exp_repn}
	\begin{aligned}
	&f_x(T,x)=\mathbb{E}^{\hat{\mathbb{P}}}[ e^{\int_0^T \kappa'(\hat{X}_s)\,ds}(h/\phi)'(\hat{X}_T)|\hat{X}_0=x]\,,\\
	&\hat{f}_x(T,x)=\mathbb{E}^{\tilde{\mathbb{P}}}[e^{\int_0^T\gamma'(\tilde{X}_s)\,ds}((h/\phi)'/\hat{\phi})'(\tilde{X}_T)|\tilde{X}_0=x] \,,
	\end{aligned} 
	\end{equation}
then  
	\begin{equation}
	\label{eqn:delta_second_order_conclusion}
	\left|\frac{{\partial_{\xi\xi}} p_T}{p_T}-\Big(\frac{\partial_\xi p_T}{p_T}\Big)^2-\frac{\phi''(\xi)}{\phi(\xi)}+\Big(\frac{\phi'(\xi)}{\phi(\xi)}\Big)^2
 -\frac{\hat{\phi}'(\xi)}{\hat{\phi}(\xi)} \frac{\partial_\xi p_T}{p_T} +\frac{\hat{\phi}'(\xi)}{\hat{\phi}(\xi)} \frac{\phi'(\xi)}{\phi(\xi)}\right|\leq c(e^{-\tilde{\lambda} T}+e^{-\hat{\lambda} T})e^{-\hat{\lambda} T}
	\end{equation}
	for some positive constant $c,$ which is dependent of $\xi$ but independent of $T.$
	In particular,   we have
\begin{equation}
\left|\frac{{\partial_{\xi\xi}} p_T}{p_T}-\frac{\phi''(\xi)}{\phi(\xi)}
\right|\leq c'e^{-\hat{\lambda} T}
\end{equation}	
	for some positive constant $c',$ which is dependent of $\xi$ but independent of $T.$
\end{thm}

\begin{proof}  
	Using the same argument as that presented in the proof of Theorem \ref{thm:initial_main_thm}, it can be shown that
	two processes $(f_x(T-t,\hat{X}_t)e^{\int_0^t \kappa'(\hat{X}_s)\,ds})_{0\le t\le T}$ and
	$(\hat{f}_x(T-t,\tilde{X}_t)e^{\int_0^t \gamma'(\tilde{X}_s)\,ds})_{0\le t\le T}$ are local martingales  under the probability measures $\hat{\mathbb{P}}_T$ and $\tilde{\mathbb{P}}_T,$ respectively.  Thus, we omit the proof.
	
	Now, assume that two functions $f_x$ and $\hat{f}_x$ satisfy  Eq.\eqref{eqn:f_xx_exp_repn}.
	If the two local martingales above are martingales, then Eq.\eqref{eqn:f_xx_exp_repn} holds by the same argument as that in Eq.\eqref{eqn:f_x_prob_repn}.
	Since  $(\kappa+\sigma'\sigma,\sigma,-\kappa',(h/\phi)')$ satisfies Assumptions \ref{assume:a}--\ref{assume:h} on $(\Omega,\mathcal{F},(\mathcal{F}_t),(\hat{\mathbb{P}}_t)_{t\ge0})$ having Brownian motion $(\hat{B}_t)_{t\ge0},$ we can construct the corresponding objects
	$$\hat{X},\,\hat{\mathcal{P}},\,(\hat{\lambda},\hat{\phi}),\,\hat{M},\,\hat{f},\,(\tilde{\mathbb{P}}_t)_{t\ge0}\,,(\tilde{B}_t)_{t\ge0}$$
	and these notations are self-explanatory. 
	The remainder function
	\begin{equation}
	\label{eqn:hat_f_exp}
	\hat{f}(t,x)=\mathbb{E}_x^{\tilde{\mathbb{P}}}[((h/\phi)'/\hat{\phi})(\hat{X}_t)]
	\end{equation} 
	satisfies
	$f_x(t,x)=\hat{f}(t,x)e^{-\hat{\lambda} t}\hat{\phi}(x)$
	by Eq.\eqref{eqn:f_x_general}, and the dynamics of $\hat{X}$ is
	\begin{equation}
	\begin{aligned}
	d\hat{X}_t
	&=(\kappa+\sigma'\sigma)(\hat{X}_t)\,dt+\sigma(\hat{X}_t)\,d\hat{B}_t\\
	&=\gamma(\hat{X}_t)\,dt+\sigma(\hat{X}_t)\,d\tilde{B}_t
	\end{aligned}
	\end{equation} 
	for
	$\gamma=\kappa+\sigma'\sigma+\sigma^2\hat{\phi}'/\hat{\phi}.$
	Applying the Feynman--Kac formula to Eq.\eqref{eqn:hat_f_exp}, we have
	$$-\hat{f}_t+\frac{1}{2}\sigma^2(x)\hat{f}_{xx}+\gamma(x)\hat{f}_x=0\,,\;\hat{f}(0,x)=((h/\phi)'/\hat{\phi})(x)\,.$$
	Observe that the function $\hat{f}$  is $C^{1,2}$ by Assumption \ref{assume:h}.
Since every coefficient is continuously differentiable in $x$ in the above PDE, the partial derivatives $\hat{f}_x$ is also in $C^{1,2}$ by \cite[Section 3.5]{friedman2008partial}.  	
	Taking the differentiation in $x,$ we get
	\begin{equation}
	\label{eqn:hat_f_x}
	-\hat{f}_{xt}+\frac{1}{2}\sigma^2(x)\hat{f}_{xxx}+(\gamma+\sigma'\sigma)(x)\hat{f}_{xx}+\gamma'(x)\hat{f}_x=0\,,\;\hat{f}_x(0,x)=((h/\phi)'/\hat{\phi})'(x)\,.
	\end{equation}

	Since $(\gamma+\sigma'\sigma,\sigma,-\gamma',((h/\phi)'/\hat{\phi})')$ satisfies Assumptions \ref{assume:a}--\ref{assume:h}, we can construct the corresponding objects
	$$\tilde{X},\,\tilde{\mathcal{P}},\,(\tilde{\lambda},\tilde{\phi}),\,\tilde{M},\,\tilde{f},\,(\overline{\mathbb{P}}_t)_{t\ge0},\,(\overline{B}_t)_{t\ge0}\,.$$
	Note that the process $\tilde{X}$ satisfies
	$$d\tilde{X}_t=(\gamma+\sigma'\sigma)(\tilde{X}_t)\,dt+\sigma(\tilde{X}_t)\,d\tilde{B}_t\,.$$
	The process $(\hat{f}_x(T-t,\tilde{X}_t)e^{\int_0^t \gamma'(\tilde{X}_s)\,ds})_{0\le t\le T}$ is assumed to be a martingale; thus,
	$$\hat{f}_x(t,x)=\mathbb{E}_x^{\tilde{\mathbb{P}}}[e^{\int_0^t\gamma'(\tilde{X}_s)\,ds}((h/\phi)'/\hat{\phi})'(\tilde{X}_t)]=\tilde{\mathcal{P}}_T((h/\phi)'/\hat{\phi})'(x)\,.$$
	We apply the Hansen--Scheinkman decomposition here.
	Since $(\tilde{\lambda},\tilde{\phi})$ is the recurrent eigenpair  and $$\tilde{f}(t,x)=\mathbb{E}_x^{\tilde{\mathbb{P}}}[\tilde{M}_T(((h/\phi)'/\hat{\phi})'/\tilde{\phi})(\tilde{X}_t)]
	=\mathbb{E}_x^{\overline{\mathbb{P}}}[(((h/\phi)'/\hat{\phi})'/\tilde{\phi})(\tilde{X}_t)]$$ is the remainder function,  
	it follows, by Eq.\eqref{eqn:HS_tranform},  that 
	\begin{equation}
	\label{eqn:hat_f_x_general}
	\begin{aligned}
	\hat{f}_x(t,x)
	&=\tilde{f}(t,x)e^{-\tilde{\lambda} t}\tilde{\phi}(x)\,.
	\end{aligned}
	\end{equation}

	Meanwhile, direct calculation gives
	\begin{equation} \label{eqn:second-order_expansion}
	\begin{aligned}
	\frac{{\partial_{\xi\xi}} p_T}{p_T}-\Big(\frac{\partial_\xi p_T}{p_T}\Big)^2-\frac{\phi''(\xi)}{\phi(\xi)}+\Big(\frac{\phi'(\xi)}{\phi(\xi)}\Big)^2
	=\frac{f_{xx}(T,\xi)}{f(T,\xi)}-\Big(\frac{f_x(T,\xi)}{f(T,\xi)}\Big)^2\,.
	\end{aligned}
	\end{equation}
	We estimate the two terms on the right-hand side.
	Eq.\eqref{eqn:first_order_sensi} states that
	\begin{equation}
	\label{eqn:f_x_squared}\Big(\frac{f_x(T,\xi)}{f(T,\xi)}\Big)^2\leq c_1e^{-2\hat{\lambda} T}
	\end{equation}
	for some positive constant $c_1.$
	To estimate the term $\frac{f_{xx}(T,\xi)}{f(T,\xi)}$ on the right-hand side of Eq.\eqref{eqn:second-order_expansion},
	observe that
	Eq.\eqref{eqn:f_x_general}
	gives
	\begin{equation}\label{eqn:second-order_f_xx}
	\begin{aligned}
	f_{xx}(t,x)
	&= \hat{f}_x(t,x)e^{-\hat{\lambda} t}\hat{\phi}(x)+ \hat{f}(t,x)e^{-\hat{\lambda} t}\hat{\phi}'(x)
	=\hat{f}_x(t,x)e^{-\hat{\lambda} t}\hat{\phi}(x)+\frac{\hat{\phi}'(x)}{\hat{\phi}(x)}f_x(t,x)\,.
	\end{aligned}
	\end{equation} 
	Combined with $\frac{f_x(T,\xi)}{f(T,\xi)}=\frac{\partial_\xi p_T}{p_T}-\frac{\phi'(\xi)}{\phi(\xi)},$
	Eq.\eqref{eqn:second-order_expansion} becomes
	\begin{equation} 
	\begin{aligned}
	&\,\frac{{\partial_{\xi\xi}} p_T}{p_T}-\Big(\frac{\partial_\xi p_T}{p_T}\Big)^2-\frac{\phi''(\xi)}{\phi(\xi)}+\Big(\frac{\phi'(\xi)}{\phi(\xi)}\Big)^2
	-\frac{\hat{\phi}'(\xi)}{\hat{\phi}(\xi)} \frac{\partial_\xi p_T}{p_T} +\frac{\hat{\phi}'(\xi)}{\hat{\phi}(\xi)} \frac{\phi'(\xi)}{\phi(\xi)}\\
	=&\,\frac{\hat{f}_x(T,\xi)\hat{\phi}(\xi)e^{-\hat{\lambda} T}}{f(T,\xi)}-\Big(\frac{f_\xi(T,\xi)}{f(T,\xi)}\Big)^2\\
	=&\,\frac{\tilde{f}(T,\xi)e^{-(\tilde{\lambda}+\hat{\lambda}) T}\tilde{\phi}(\xi)\hat{\phi}(\xi)}{f(T,\xi)}-\Big(\frac{f_\xi(T,\xi)}{f(T,\xi)}\Big)^2\,.
	\end{aligned}
	\end{equation}
	Since $f(T,\xi)$ and $\tilde{f}(T,\xi)$ converge  to nonzero constants as $T\to\infty,$ 
	combined with Eq.\eqref{eqn:f_x_squared}, we conclude that
	\begin{equation} 
	\left|\frac{{\partial_{\xi\xi}} p_T}{p_T}-\Big(\frac{\partial_\xi p_T}{p_T}\Big)^2-\frac{\phi''(\xi)}{\phi(\xi)}+\Big(\frac{\phi'(\xi)}{\phi(\xi)}\Big)^2
	-\frac{\hat{\phi}'(\xi)}{\hat{\phi}(\xi)} \frac{\partial_\xi p_T}{p_T} +\frac{\hat{\phi}'(\xi)}{\hat{\phi}(\xi)} \frac{\phi'(\xi)}{\phi(\xi)}\right|\leq c(e^{-\tilde{\lambda} T}+e^{-\hat{\lambda} T})e^{-\hat{\lambda} T}
	\end{equation}
	for some positive constant $c,$ which is dependent of $\xi$ but independent of $T.$
	In particular, by using Eq.\eqref{eqn:first_order_sensi}, we have
	\begin{equation}
	\left|\frac{{\partial_{\xi\xi}} p_T}{p_T}-\frac{\phi''(\xi)}{\phi(\xi)}
	\right|\leq c'e^{-\hat{\lambda} T}
	\end{equation}
	for some positive constant $c',$ which is dependent of $\xi$ but independent of $T.$
\end{proof}

\begin{remark}
	One can find a sufficient condition  such that the martingale property for Eq.\eqref{eqn:f_xx_exp_repn}, Eq.\eqref{eqn:drift_f_eps^0} and Eq.\eqref{eqn:drift_f_eps^0_vega} holds true by the same way stated in Proposition \ref{prop:1st_initial_sen_linear}.
\end{remark}

\subsection{Sensitivities of the drift and diffusion terms}
\label{sec:sen_drift_diffusion}

We conduct long-term sensitivity analysis with respect to perturbations of the drift and diffusion terms. 
The variable $\epsilon$ below is a perturbation parameter.

\begin{assume}\label{assume:epsilon} Let  $I$ be an open neighborhood of $0$ and let  $b^{(\epsilon)}(x),\sigma^{(\epsilon)}(x),$ $r^{(\epsilon)}(x),$  $h^{(\epsilon)}(x)$ be  functions of two variables $(\epsilon,x)$ on $I\times\mathcal{D}$  such that for each $x$ they are continuously differentiable in $\epsilon$ on $I$ and $b^{(0)}=b,$ $\sigma^{(0)}=\sigma,$ $r^{(0)}=r,h^{(0)}=h.$
\end{assume}

\begin{assume}\label{assume:epsilon_a}
	For each $\epsilon\in I,$ the quadruple $(b^{(\epsilon)},\sigma^{(\epsilon)},r^{(\epsilon)},h^{(\epsilon)})$ satisfies
	Assumptions \ref{assume:a}--\ref{assume:h}.
\end{assume}

\noindent From the above assumptions,   the notations
$$X^{(\epsilon)},\,\mathcal{P}^{(\epsilon)},\,(\lambda^{(\epsilon)},\phi^{(\epsilon)}),\,M^{(\epsilon)},\,f^{(\epsilon)},\,(\hat{\mathbb{P}}^{(\epsilon)})_{t\ge0}\,,\; (\hat{B}_t^{(\epsilon)})_{t\ge0} $$  are self-explanatory.
The process $X^{(\epsilon)}$ satisfies
\begin{equation}
\begin{aligned}
dX_t^{(\epsilon)}
&=b^{(\epsilon)}(X_t^{(\epsilon)})\,dt+\sigma(X_t^{(\epsilon)})\,dB_t \\
&=\kappa^{(\epsilon)}(X_t^{(\epsilon)})\,dt+\sigma(X_t^{(\epsilon)})\,d\hat{B}_t^{(\epsilon)} 
\end{aligned}
\end{equation} 
where
$$\kappa^{(\epsilon)}:=b^{(\epsilon)}+\sigma^{(\epsilon)2}\phi_x^{(\epsilon)}/\phi^{(\epsilon)}\,.$$
For notational simplicity, we define
$$\ell^{(\epsilon)}=\partial_\epsilon\kappa^{(\epsilon)}\,,\;\ell=\ell^{(0)}\,,\;\Sigma^{(\epsilon)}=\partial_\epsilon\sigma^{(\epsilon)}\,,\;\Sigma=\Sigma^{(0)}\,.$$
The perturbed pricing operator is 
$$\mathcal{P}_T^{(\epsilon)} h(x)=\mathbb{E}_x^{\mathbb{P}}[e^{-\int_0^T r^{(\epsilon)}(X_s^{(\epsilon)})\,ds} h(X_T^{(\epsilon)})]$$
and 
the remainder function is
\begin{equation}
\label{eqn:1st_order_f_eps}
f^{(\epsilon)}(t,x)=\mathbb{E}_x^{\hat{\mathbb{P}}^{(\epsilon)}}
[(h^{(\epsilon)}/\phi^{(\epsilon)})(X_t^{(\epsilon)})]\,.
\end{equation}

\begin{assume}\label{assume:f_eps_pert}
	For each $t\in [0,\infty)$ and $x\in\mathcal{D},$ the  functions $\phi^{(\epsilon)}(x),$ $\phi_x^{(\epsilon)}(x)$  and the  functions $f^{(\epsilon)}(t,x),$ $f_x^{(\epsilon)}(t,x),$  $f_{xx}^{(\epsilon)}(t,x),$ $f_t^{(\epsilon)}(t,x)$ are continuously differentiable in $\epsilon$ on $I.$
\end{assume}

For $\epsilon\in I,$ consider the expectation
$$p_T^{(\epsilon)}=\mathbb{E}_\xi^{\mathbb{P}}[e^{-\int_0^T r^{(\epsilon)}(X_s^{(\epsilon)})\,ds} h^{(\epsilon)}(X_T^{(\epsilon)})]=\mathcal{P}_T^{(\epsilon)}h^{(\epsilon)}(\xi)\,.$$
We are interested in the large-time behavior of
$$\frac{\partial}{\partial \epsilon}\Big|_{\epsilon=0}p_T^{(\epsilon)}\,,$$ 
which measures the sensitivity with respect to the perturbation parameter $\epsilon.$
The main idea is the Hansen--Scheinkman decomposition, which allows the expectation $p_T^{(\epsilon)}$ to be expressed as
$$p_T^{(\epsilon)}=\phi^{(\epsilon)}(\xi)\,e^{-\lambda^{(\epsilon)} T}f^{(\epsilon)}(T,\xi)\,.$$
By taking the differentiation in $\epsilon$ on $I,$
after some manipulations,
we have
\begin{equation}
\label{eqn:p_T^eps_deriva}
\frac{1}{T}\frac{\partial}{\partial \epsilon}\Big|_{\epsilon=0} \ln p_T^{(\epsilon)}+\frac{\partial}{\partial \epsilon}\Big|_{\epsilon=0}\lambda^{(\epsilon)} =\frac{1}{T}\left(\frac{\frac{\partial}{\partial \epsilon}\big|_{\epsilon=0}\phi^{(\epsilon)}(\xi)}{\phi(\xi)}+\frac{\frac{\partial}{\partial \epsilon}\big|_{\epsilon=0}f^{(\epsilon)}(T,\xi)}{f(T,\xi)} \right)\,.
\end{equation} 
We will find sufficient conditions for the term
$$\frac{\frac{\partial}{\partial \epsilon}\big|_{\epsilon=0}f^{(\epsilon)}(T,\xi)}{f^{(\epsilon)}(T,\xi)}$$
to be uniformly bounded in $T.$
Then, Eq.\eqref{eqn:p_T^eps_deriva} gives
\begin{equation}
\label{eqn:para_pert_aim}
\left|\frac{1}{T}\frac{\partial}{\partial \epsilon}\Big|_{\epsilon=0} \ln p_T^{(\epsilon)}+\frac{\partial}{\partial \epsilon}\Big|_{\epsilon=0}\lambda^{(\epsilon)}\right|
\leq  \frac{c}{T}
\end{equation} 
for some positive constant $c,$ which is independent of $T.$
Thus
$$\frac{1}{T}\frac{\partial}{\partial \epsilon}\Big|_{\epsilon=0} \ln p_T^{(\epsilon)}\to -\frac{\partial}{\partial \epsilon}\Big|_{\epsilon=0}\lambda^{(\epsilon)}$$
as $T\to\infty$ and its convergence rate is $O(1/T).$
Sufficient conditions with respect to the drift perturbation and the diffusion perturbation are investigated in Section \ref{sec:drift} and \ref{sec:diffusion}, respectively.

\subsubsection{Drift-term sensitivity}
\label{sec:drift}

This section investigates the long-term sensitivity  with respect to a perturbation of the drift term. 
We provide a sufficient condition  for Eq.\eqref{eqn:para_pert_aim} to hold.
Let $(b^{(\epsilon)},\sigma,r^{(\epsilon)},f^{(\epsilon)})$  be a quadruple  of functions   satisfying Assumptions  \ref{assume:a} -- \ref{assume:h}
on the consistent probability space $(\Omega,\mathcal{F},(\mathcal{F}_t)_{t\ge0},(\mathbb{P}_t)_{t\ge0})$ having Brownian motion $(B_t)_{t\ge0}.$ The diffusion function $\sigma$ is not perturbed.
Recall that if
$(\kappa+\sigma'\sigma,\sigma,-\kappa',(h/\phi)')$
satisfies Assumptions \ref{assume:a} -- \ref{assume:h}
on the consistent probability space $(\Omega,\mathcal{F},(\mathcal{F}_t)_{t\ge0},(\hat{\mathbb{P}}_t)_{t\ge0})$ having Brownian motion $(\hat{B}_t)_{t\ge0},$
we can construct the corresponding objects $$\hat{X},\,\hat{\mathcal{P}},\,(\hat{\lambda},\hat{\phi}),\,\hat{M},\,\hat{f},\,(\tilde{\mathbb{P}}_t)_{t\ge0}\,,(\tilde{B}_t)_{t\ge0}\,.$$

\begin{thm}\label{thm:drift_pert}
	Suppose that the following conditions hold.
	\begin{enumerate}[(i)]
		\item The quadruple $(b^{(\epsilon)},\sigma,r^{(\epsilon)},h^{(\epsilon)})$
		satisfies Assumptions \ref{assume:epsilon}-- \ref{assume:f_eps_pert} on the consistent probability space $(\Omega,\mathcal{F},(\mathcal{F}_t)_{t\ge0},(\mathbb{P}_t)_{t\ge0})$  having Brownian motion $(B_t)_{t\ge0}.$ 
		\item The quadruple $(\kappa+\sigma'\sigma,\sigma,-\kappa',(h/\phi)')$ satisfies Assumptions \ref{assume:a} -- \ref{assume:h}
		 on the consistent probability space
		$(\Omega,\mathcal{F},(\mathcal{F}_t)_{t\ge0},(\hat{\mathbb{P}}_t)_{t\ge0})$ having Brownian motion $(\hat{B}_t)_{t\ge0}.$
	\end{enumerate} 
Then, for each $T>0$ the process $$\left(f_{\epsilon }^{(0)}(T-t,X_t)+\int_0^t\ell(X_s)f_x(T-s,X_s)\,ds\right)_{0\le t\le T}$$ is a local martingale  under the probability measure $\hat{\mathbb{P}}_T.$ 
If this process is a martingale for each  $T>0$, then   
 \begin{equation}\label{eqn:drift_f_eps^0}
 \begin{aligned}
 f_{\epsilon}^{(0)}(T,\xi)
 =\mathbb{E}_\xi^{\hat{\mathbb{P}}}\Big[\int_0^T\ell(X_t) f_x(T-t,X_t)\,dt+\partial_\epsilon|_{\epsilon=0}(h^{(\epsilon)}/\phi^{(\epsilon)})(X_T)\Big] \,.
 \end{aligned}
 \end{equation}
Therefore, if two expectations 
 $\mathbb{E}_\xi^{\hat{\mathbb{P}}}[\int_0^T\ell(X_s) f_x(T-s,X_s)\,ds]$ and  
 $\mathbb{E}_\xi^{\hat{\mathbb{P}}}[\partial_\epsilon|_{\epsilon=0}(h^{(\epsilon)}/\phi^{(\epsilon)})(X_T)]$ are uniformly bounded in $T$ on $[0,\infty),$ 
	then 
$$\left|\frac{1}{T}\frac{\partial}{\partial \epsilon}\Big|_{\epsilon=0} \ln p_T^{(\epsilon)}+\frac{\partial}{\partial \epsilon}\Big|_{\epsilon=0}\lambda^{(\epsilon)}\right|
\leq  \frac{c}{T}$$
for some positive constant $c,$ which is dependent of $\xi$ but independent of $T.$ \end{thm}

\begin{proof}
	Applying the Feynman--Kac formula to Eq.\eqref{eqn:1st_order_f_eps}, we get
	$$-f_t^{(\epsilon)}+\frac{1}{2}\sigma^{2}(x)f_{xx}^{(\epsilon)}+\kappa^{(\epsilon)}(x)f_x^{(\epsilon)}=0\,,\;f^{(\epsilon)}(0,x)=(h^{(\epsilon)}/\phi^{(\epsilon)})(x)$$
	for
	$\kappa^{(\epsilon)}=b^{(\epsilon)}+\sigma^{2}\phi_x^{(\epsilon)}/\phi^{(\epsilon)}.$
Let us differentiate this PDE in $\epsilon$ and evaluate it at $\epsilon=0$ (Assumptions \ref{assume:epsilon} and  \ref{assume:f_eps_pert}). Then
\begin{equation}
\label{eqn:f_eps_PDE}
-f_{\epsilon t}^{(0)}+\frac{1}{2}\sigma^2(x)f_{\epsilon xx}^{(0)}+\kappa(x)f_{\epsilon x}^{(0)}+\ell(x)f_x^{(0)}=0\,,\;f_\epsilon^{(0)}(0,x)=\partial_\epsilon|_{\epsilon=0}(h^{(\epsilon)}/\phi^{(\epsilon)})(x)\,.
\end{equation}
Applying the Ito formula to the process $(f_{\epsilon }^{(0)}(T-t,X_t)+\int_0^t\ell(X_s)f_x(T-s,X_s)\,ds)_{0\le t\le T},$ it follows that
\begin{equation}\label{eqn:drift_sen_Ito}
\begin{aligned}
&d\Big(f_{\epsilon}^{(0)}(T-t,{X}_t)+\int_0^t\ell({X}_s)f_x(T-s,{X_s})\,ds\Big)\\
=&\;\big(-f_{\epsilon t}^{(0)}+\frac{1}{2}\sigma^2(X_t)f_{\epsilon xx}^{(0)}+\kappa({X}_t)f_{\epsilon x}^{(0)} +\ell({X}_t)f_x\big)\,dt+\sigma({X}_t)f_{\epsilon x}^{(0)}  \,d\hat{B}_t=\sigma({X}_t)f_{\epsilon x}^{(0)}  \,d\hat{B}_t\,.
\end{aligned}
\end{equation}
Thus, the process $(f_{\epsilon }^{(0)}(T-t,X_t)+\int_0^t\ell(X_s)f_x(T-s,X_s)\,ds)_{0\le t\le T}$ is a local martingale under the probability measure $\hat{\mathbb{P}}_T.$

Now, assume that
\begin{equation} 
\begin{aligned}
f_{\epsilon}^{(0)}(T,\xi)
=\mathbb{E}_\xi^{\hat{\mathbb{P}}}\Big[\int_0^T\ell(X_t) f_x(T-t,X_t)\,dt+\partial_\epsilon|_{\epsilon=0}(h^{(\epsilon)}/\phi^{(\epsilon)})(X_T)\Big]\,. 
\end{aligned}
\end{equation}
This equality also holds if the process
$(f_{\epsilon }^{(0)}(T-t,X_t)+\int_0^t\ell(X_s)f_x(T-s,X_s)\,ds)_{0\le t\le T}$ is a  martingale because
\begin{equation}
\begin{aligned}
f_{\epsilon}^{(0)}(T,\xi)
&=\mathbb{E}_\xi^{\hat{\mathbb{P}}}\Big[\int_0^T\ell(X_t) f_x(T-t,X_t)\,dt+f_{\epsilon}^{(0)}(0,X_T)\Big]\\
&=\mathbb{E}_\xi^{\hat{\mathbb{P}}}\Big[\int_0^T\ell(X_t) f_x(T-t,X_t)\,dt+\partial_\epsilon|_{\epsilon=0}(h^{(\epsilon)}/\phi^{(\epsilon)})(X_T)\Big]\,.
\end{aligned}
\end{equation}
Since two expectations $\mathbb{E}_\xi^{\hat{\mathbb{P}}}[\int_0^T\ell(X_t) f_x(T-t,X_t)\,dt]$ and  
$\mathbb{E}_\xi^{\hat{\mathbb{P}}}[\partial_\epsilon|_{\epsilon=0}(h^{(\epsilon)}/\phi^{(\epsilon)})(X_T)]$ are uniformly bounded in $T$ on $[0,\infty)$,  the function
$f_\epsilon^{(0)}(T,\xi)$ is also uniformly bounded in $T$ on $[0,\infty).$ 
Using Eq.\eqref{eqn:p_T^eps_deriva},
we obtain the desired result. 
\end{proof}

\begin{remark}	Under the same hypothesis of Theorem \ref{thm:drift_pert},
	we obtain   the probabilistic representation 
	$$f_{\epsilon}^{(0)}(T,\xi)
	=\mathbb{E}_\xi^{\hat{\mathbb{P}}}\Big[\int_0^T\ell(X_t) f_x(T-t,X_t)\,dt+\partial_\epsilon|_{\epsilon=0}(h^{(\epsilon)}/\phi^{(\epsilon)})(X_T)\Big] $$
	presented in Eq.\eqref{eqn:drift_f_eps^0} if
	$$\mathbb{E}_\xi^{\hat{\mathbb{P}}}\Big[\int_0^T(\sigma({X}_t)f_{\epsilon x}^{(0)}(T-t,X_t))^2  \,dt\Big]<\infty\,.$$ 
	This is evident from Eq.\eqref{eqn:drift_sen_Ito} because the process $(f_{\epsilon }^{(0)}(T-t,X_t)+\int_0^t\ell(X_s)f_x(T-s,X_s)\,ds)_{0\le t\le T}$ is a martingale under the probability measure $\hat{\mathbb{P}}_T.$
\end{remark}

\subsubsection{Diffusion-term sensitivity}
\label{sec:diffusion}

This section investigates the long-term sensitivity  with respect to a perturbation of the diffusion term. 
We provide a sufficient condition  for Eq.\eqref{eqn:para_pert_aim} to hold. Let $(b^{(\epsilon)},\sigma^{(\epsilon)},r^{(\epsilon)},f^{(\epsilon)})$  be a quadruple  of functions   satisfying Assumptions  \ref{assume:epsilon} -- \ref{assume:f_eps_pert}
on the consistent probability space $(\Omega,\mathcal{F},(\mathcal{F}_t)_{t\ge0},(\mathbb{P}_t)_{t\ge0})$ having Brownian motion $(B_t)_{t\ge0}.$
The three  quadruples  listed
in Table \ref{Tab:Tcr} are used in this section.
We  recall that
$\Sigma^{(\epsilon)}=\partial_\epsilon\sigma^{(\epsilon)}$ and $\Sigma=\Sigma^{(0)}.$

\begin{thm}
Suppose that the following conditions hold.
	\begin{enumerate}[(i)]
	\item The quadruple $(b^{(\epsilon)},\sigma^{(\epsilon)},r^{(\epsilon)},f^{(\epsilon)})$ satisfies Assumptions \ref{assume:epsilon}-- \ref{assume:f_eps_pert} on the consistent probability space $(\Omega,\mathcal{F},(\mathcal{F}_t)_{t\ge0},(\mathbb{P}_t)_{t\ge0})$ having Brownian motion $(B_t)_{t\ge0}.$  
	\item The quadruple $(\kappa+\sigma'\sigma,\sigma,-\kappa',(h/\phi)')$ satisfies Assumptions \ref{assume:a} -- \ref{assume:h}
	on the consistent probability space
	$(\Omega,\mathcal{F},(\mathcal{F}_t)_{t\ge0},(\hat{\mathbb{P}}_t)_{t\ge0})$ having Brownian motion $(\hat{B}_t)_{t\ge0}.$  
	\item The quadruple   $(\gamma+\sigma'\sigma,\sigma,-\gamma',((h/\phi)'/\hat{\phi})')$ satisfies Assumptions \ref{assume:a} -- \ref{assume:h} 	on the consistent probability space $(\Omega,\mathcal{F},(\mathcal{F}_t)_{t\ge0},(\tilde{\mathbb{P}}_t)_{t\ge0})$ having Brownian motion $(\tilde{B}_t)_{t\ge0}.$ 
\end{enumerate}  
Then, for each $T>0$
 the process $$\Big(f_{\epsilon }^{(0)}(T-t,X_t)+\int_0^t((\sigma\Sigma)(X_s)f_{xx}(T-s,X_s)+\ell(X_s)f_x(T-s,X_s))\,ds\Big)_{0\le t\le T}$$ is a
 local martingale 
 under the probability measure $\hat{\mathbb{P}}_T.$
If this process is a martingale for each $T>0$,
 \begin{equation}\label{eqn:drift_f_eps^0_vega}
 \begin{aligned}
 f_{\epsilon}^{(0)}(T,\xi)
 &=\mathbb{E}_\xi^{\hat{\mathbb{P}}}\Big[\int_0^T((\sigma\Sigma)(X_t) f_{xx}(T-t,X_t)+\ell(X_t) f_x(T-t,X_t))\,dt\Big]+\mathbb{E}_\xi^{\hat{\mathbb{P}}}[\partial_\epsilon|_{\epsilon=0}(h^{(\epsilon)}/\phi^{(\epsilon)})(X_T)] \,.
 \end{aligned}
 \end{equation}
Therefore, if
three expectations $\mathbb{E}_\xi^{\hat{\mathbb{P}}}[\int_0^T(\sigma\Sigma)(X_s) f_{xx}(t-s,X_s)\,ds],$
$\mathbb{E}_\xi^{\hat{\mathbb{P}}}[\int_0^T\ell(X_s) f_x(t-s,X_s)\,ds]$ and  
$\mathbb{E}_\xi^{\hat{\mathbb{P}}}[\partial_\epsilon|_{\epsilon=0}(h^{(\epsilon)}/\phi^{(\epsilon)})(X_T)]$ are uniformly bounded in $T$ on $[0,\infty),$
then
	$$\left|\frac{1}{T}\frac{\partial}{\partial \epsilon}\Big|_{\epsilon=0} \ln p_T^{(\epsilon)}+\frac{\partial}{\partial \epsilon}\Big|_{\epsilon=0}\lambda^{(\epsilon)}\right|
	\leq  \frac{c}{T}$$
	for some positive constant $c,$ 
which is dependent of $\xi$ but independent of $T.$
 \end{thm}

\begin{proof}
	  
Applying the Feynman--Kac formula to Eq.\eqref{eqn:1st_order_f_eps}, we get
$$-f_t^{(\epsilon)}+\frac{1}{2}\sigma^{(\epsilon)2}(x)f_{xx}^{(\epsilon)}+\kappa^{(\epsilon)}(x)f_x^{(\epsilon)}=0\,,\;f^{(\epsilon)}(0,x)=(h^{(\epsilon)}/\phi^{(\epsilon)})(x)$$
for
$\kappa^{(\epsilon)}=b^{(\epsilon)}+\sigma^{(\epsilon)2}\phi_x^{(\epsilon)}/\phi^{(\epsilon)}.$
Let us differentiate this PDE in $\epsilon$ and evaluate it at $\epsilon=0,$ then 
$$-f_{\epsilon t}^{(0)}+\frac{1}{2}\sigma^{2}(x)f_{\epsilon xx}^{(0)}+\kappa(x)f_{\epsilon x}^{(0)}+(\sigma\Sigma)(x)f_{xx}^{(0)} +\ell(x)f_x^{(0)}=0\,,\;f_\epsilon^{(0)}(0,x)=\partial_\epsilon|_{\epsilon=0}(h^{(\epsilon)}/\phi^{(\epsilon)})(x)\,.$$
From the Ito formula, one can show that 
the process $$\Big(f_{\epsilon }^{(0)}(T-t,X_t)+\int_0^t((\sigma\Sigma)(X_s)f_{xx}(T-s,X_s)+\ell(X_s)f_x(T-s,X_s))\,ds\Big)_{0\le t\le T}$$ is a
local martingale under the probability measure $\hat{\mathbb{P}}_T$ by checking that the $dt$ term vanishes. 

Suppose that  $f_\epsilon$ satisfies
\begin{equation}
\begin{aligned}
f_{\epsilon}^{(0)}(T,\xi)
&=\mathbb{E}_\xi^{\hat{\mathbb{P}}}\Big[\int_0^T((\sigma\Sigma)(X_t) f_{xx}(T-t,X_t)+\ell(X_t) f_x(T-t,X_t))\,dt\Big]+\mathbb{E}_\xi^{\hat{\mathbb{P}}}[\partial_\epsilon|_{\epsilon=0}(h^{(\epsilon)}/\phi^{(\epsilon)})(X_T)] \,.
\end{aligned}
\end{equation}
This equality  holds if the local martingale above is a martingale.
Since three expectations
$\mathbb{E}_\xi^{\hat{\mathbb{P}}}[\int_0^T(\sigma\Sigma)(X_t) f_{xx}(T-t,X_t)\,dt],$ $\mathbb{E}_\xi^{\hat{\mathbb{P}}}[\int_0^T\ell(X_t) f_x(T-t,X_t)\,dt]$ and  
$\mathbb{E}_\xi^{\hat{\mathbb{P}}}[\partial_\epsilon|_{\epsilon=0}(h^{(\epsilon)}/\phi^{(\epsilon)})(X_T)]$ are uniformly bounded in $T$ on $[0,\infty)$,  the function
$f_\epsilon^{(0)}(T,\xi)$ is also uniformly bounded in $T$ on $[0,\infty).$ 
Using Eq.\eqref{eqn:p_T^eps_deriva},
we obtain the desired result. 
\end{proof}

\section{Applications}
\label{sec:appli}

We present applications of the previous results  to three practical problems: utility maximization, entropic risk measures and bond prices.

\subsection{Utility maximization}

This section discusses the classical utility maximization problem in complete markets as an application.
We first describe the general Ito process models to formulate the utility maximization problem.
Then, more specific models, namely factor models and  local volatility models, are discussed in Sections \ref{sec:factor_model} and \ref{sec:lvm}, respectively.

Let 
$(\Omega,\mathcal{F},(\mathcal{F}_t)_{t\ge0},({\bf P}_t)_{t\ge0})$
be a consistent probability space   having a $d$-dimensional Brownian motion $Z=(Z_t^{(1)},\cdots,Z_t^{(d)})_{t\ge0}^\top.$ 
The filtration $(\mathcal{F}_t)_{t\ge0}$ satisfies the usual condition.  The probability measures $({\bf P}_t)_{t\ge0}$ are referred to as the physical measures of the market.
An Ito process model is describes as follows.  
\begin{enumerate}[(i)]
	\item The bank account, denoted as $(G_t)_{t\ge0},$ is a process given by 
	\begin{equation}
	\label{eqn:Ito_G}
	G_t=e^{\int_0^tr_u\,du}\,,\;t\ge0
	\end{equation} 
	where  $(r_t)_{t\ge0}$ is a progressively measurable process taking values in $[0,\infty)$ such that $\int_0^Tr_u\,du<\infty$ ${\bf P}_T$-almost surely for each $T\ge0.$
	\item There are $d$ stocks $(S_t)_{t\ge0}=(S_t^{(1)},\cdots,S_t^{(d)})_{t\ge0}^\top$  described as
	\begin{equation}
	\label{eqn:Ito stocks}
	S_t^{(i)}=S_0^{(i)}e^{\int_0^t(\mu_u^{(i)}-\frac{1}{2}|\sigma_u^{(i)}|^2)\,du+\int_0^t\sigma_u^{(i)} dZ_u}\,,\;S_0^{(i)}>0
	\end{equation} 
	for $i=1,2,\cdots,d,$
	where $(\mu_t^{(i)})_{t\ge0}$ is a progressively measurable process with $\int_0^T|\mu_u^{(i)}|\,du<\infty$ ${\bf P}_T$-almost surely for each $T\ge0,$ and $(\sigma_t^{(i)})_{t\ge0}$ is a $d$-dimensional progressively measurable row process  with $\int_0^T|\sigma_u^{(i)}|^2\,du<\infty$ ${\bf P}_T$-almost surely for each $T\ge0.$
\end{enumerate}
For simpler expressions, we define a $d$-dimensional column process
\begin{equation}\label{eqn:Ito_model_vol_drift}
(\mu_t)_{t\ge0}:=(\mu_t^{(1)},\cdots,\mu_t^{(d)})_{t\ge0}^\top
\end{equation} 
and a $d\times d$-matrix process
\begin{equation}\label{eqn:Ito_model_vol_matrix}
\begin{aligned}
(\sigma_t)_{t\ge0}=
\begin{pmatrix}
   \sigma_t^{(1)}\\
 \sigma_t^{(2)}\\
   \vdots\\
   \sigma_t^{(d)}
\end{pmatrix}\,.
\end{aligned}
\end{equation}
In the SDE form, we can write
$$dG_t=r_tG_t\,dt$$
and
\begin{equation}
\label{eqn:Ito stocks_SDE}
dS_t=D(S_t)\mu_t\,dt+D(S_t)\sigma_t\,dZ_t
\end{equation} 
where $D(x)$ for $x=(x_1,\cdots,x_d)^\top$ is the diagonal matrix whose $i$-th diagonal entry is $x_i$ for $i=1,2\cdots,d$ and off-diagonal entries are zero.

A self-financing portfolio is a pair $(x,\pi)$ of a real number $x$ and  a $d$-dimensional  progressively measurable  process $(\pi_t)_{t\ge0}$ that is $(S_t/G_t)_{0\le t\le T}$-integrable  for each $T\ge0.$
The value process $\Pi=(\Pi_t)_{t\ge0}=(\Pi_t^{(x,\pi)})_{t\ge0}$ 
of the portfolio  $(x,\pi)$ is given by 
\begin{equation}
\label{eqn:port_value}
\Pi_t^{(x,\pi)}=xG_t+G_t\int_0^t\pi_u\,d(S/G)_u\,.
\end{equation} 
For a positive number $x,$ let  $\mathcal{X}_T(x)$ denote
the family of nonnegative value processes with  initial value $x$, i.e.,
$$\mathcal{X}_T(x)=\big\{(\Pi_t^{(x,\pi)})_{0\le t\le T}: \exists \textnormal{ a self-financing portfolio  } (x,\pi) \textnormal{ such that } \Pi_t^{(x,\pi)} \ge 0 \textnormal{ for all } 0\le t\le T  \big\}\,.$$ 
Define 
\begin{equation}
\label{eqn:wealth}\mathcal{X}_T=\mathcal{X}_T(1)\,.
\end{equation}

We now construct
a consistent family of risk-neutral measures $({\bf Q}_t)_{t\ge0}.$
\begin{assume}\label{assume:theta} Consider the following conditions.
\begin{enumerate}[(i)]
	\item The $d\times d$-matrix process
	$\sigma:[0,T]\times \Omega\to\mathbb{R}^{d\times d}$  
	is invertible $\textnormal{Leb}\otimes \mathbb{P}_T$-almost surely
	for each $T\ge0.$ 
	\item The process 
	\begin{equation}
	\label{eqn:theta_utility_max}
	\theta_t:=\sigma_t^{-1}(\mu_t-r_t{\bf 1})\,,\;0\le t\le T
	\end{equation}
 is square-integrable, i.e., $\int_0^T|\theta_u|^2\,du<\infty,$ ${\bf P}_T$-almost surely for each $T\ge0.$ This process $\theta$ is referred to as the market price of risk.
	\item A local martingale 
	$$(e^{-\int_0^t\theta_u\,dZ_u-\frac{1}{2}\int_0^t|\theta_u|^2\,du})_{0\le t\le T}$$
	is a martingale  under the probability measure ${\bf P}_T$ for each $T\ge0.$
\end{enumerate}
\end{assume}

\begin{remark}\label{rmk:NFLVR} 
	The fundamental theorem of asset pricing states that Assumption \ref{assume:theta} is equivalent to the NFLVR (no free lunch with vanishing risk) condition on time interval $[0,T]$ for each $T>0$. Refer to \cite{delbaen1994general}.
\end{remark}

\noindent Define a probability measure ${\bf Q}_T$ on $\mathcal{F}_T$ by
\begin{equation}\label{eqn:LN}
L_T:=\frac{d{\bf Q}_T}{d{\bf P}_{T}}=e^{-\int_0^T\theta_u\,dZ_u-\frac{1}{2}\int_0^T|\theta_u|^2\,du}\,.
\end{equation}
Then, the probability measure ${\bf Q}_T$ is equivalent to ${\bf P}_T$ and the processes $S^{(1)}/G,\cdots,S^{(d)}/G$    are ${\bf Q}_T$-local martingales. 
As is well known, this measure ${\bf Q}_T$ is called the {\em risk-neutral measure}. 
It is evident that
the family of risk-neutral measures 
$({\bf Q}_t)_{t\ge0}$ is consistent.
Define a process $W$ as
$$W_t=\int_0^t\theta_u\,du+Z_t\,,\;t\ge0\,.$$
Then, $W$ is a Brownian motion 
on the consistent probability space $(\Omega,\mathcal{F},(\mathcal{F}_t)_{t\ge0},({\bf Q}_t)_{t\ge0})$
by the Girsanov theorem.

In this market, an agent    has a utility function
$U:[0,\infty)\to \mathbb{R}$ for wealth. Assume that the utility function is a power function of the form
\begin{equation}
\label{eqn:utility_ftn}
u(x)=x^\nu/\nu
\end{equation} 
for $\nu<0.$
This utility function is increasing, strictly concave, continuously differentiable and satisfies the Inada conditions. For given unit initial endowment,  the agent wants to maximize the  
expected utility 
\begin{equation}
\label{eqn:utility_max}
\max_{\Pi\in \mathcal{X}_T}\mathbb{E}^{\bf{P}}[U(\Pi_T)]\,,
\end{equation} 
of which the long-term sensitivity with respect to small changes of the underlying process is of interest to us.

\cite{kramkov1999asymptotic} states that the optimal portfolio value is
$$\Pi_T^{\textnormal{(opt)}}=c_T(U')^{-1}(L_T/G_T)=c_T(L_T/G_T)^{1/{(\nu-1)}}$$
where 
$c_T$ is a constant satisfying the budget constraint
$$1=\mathbb{E}^{{\bf Q}}[\Pi_T^{\textnormal{(opt)}}/G_T]
=c_T\mathbb{E}^{\bf Q}[(L_T/G_T^\nu)^{1/(\nu-1)}]\,.$$
Thus, the optimal expected utility is
\begin{equation}\label{eqn:opt_exp}
\begin{aligned}
\mathbb{E}^{{\bf P}}[U(\Pi_T^{\textnormal{(opt)}})]
&=\mathbb{E}^{\bf P}[U(c_T(L_T/G_T)^{1/{(\nu-1)}})]=\frac{1}{\nu}c_T^\nu\mathbb{E}^{\bf P}[(L_T/G_T)^{\nu/(\nu-1)}]\\
&=\frac{1}{\nu}\frac{\mathbb{E}^{\bf P}[(L_T/G_T)^{\nu/(\nu-1)}]}{(\mathbb{E}^{{\bf Q}}[(L_T/G_T^\nu)^{1/(\nu-1)}])^\nu}
=\frac{1}{\nu}\frac{\mathbb{E}^{{\bf Q}}[(L_T/G_T^\nu)^{1/(\nu-1)}]}{(\mathbb{E}^{{\bf Q}}[(L_T/G_T^\nu)^{1/(\nu-1)}])^\nu}
=\frac{1}{\nu}(\mathbb{E}^{{\bf Q}}[(L_T/G_T^\nu)^{1/(\nu-1)}])^{1-\nu}\;.
\end{aligned}
\end{equation}
Defining
$$u_T=\mathbb{E}^{{\bf Q}}[(L_T/G_T^\nu)^{1/(\nu-1)}]\,,$$
the  optimal expected utility satisfies 
$$\max_{\Pi\in \mathcal{X}}\mathbb{E}^{\bf P}[U(\Pi_T)]=\mathbb{E}^{\bf P}[U(\Pi_T^{\textnormal{(opt)}})]= {u_T^{1-\nu}}/{\nu}\,.$$
Thus, the problem of studying the long-term sensitivity  of the  optimal expected utility boils down to analyzing the expectation $u_T.$  
 
We can analyze  the large-time behavior of $u_T$  for small changes of the underlying process as follows.
From Eq.\eqref{eqn:LN},  the Radon--Nikodym derivative $L_T$ is
$$L_T=e^{-\int_0^T\theta_u\,dZ_u-\frac{1}{2}\int_0^T|\theta_u|^2\,du}=e^{-\int_0^T\theta_u\,dW_u+\frac{1}{2}\int_0^T|\theta_u|^2\,du}\,,$$
thus
\begin{equation}
\begin{aligned}
u_T=\mathbb{E}^{{\bf Q}}[(L_T/G_T^\nu)^{1/(\nu-1)}]
&=\mathbb{E}^{\bf Q}[e^{-\frac{1}{\nu-1}\int_0^T\theta_u\, dW_u+\frac{1}{2(\nu-1)}\int_0^T|\theta_u|^2\,du-\frac{\nu}{\nu-1}\int_0^Tr_u\,du}]\\
&=\mathbb{E}^\mathbb{P}[e^{\frac{\nu}{2(\nu-1)^2}\int_0^T|\theta_u|^2\,du-\frac{\nu}{\nu-1}\int_0^Tr_u\,du}]
\end{aligned}
\end{equation}
where $(\mathbb{P}_t)_{t\ge0}$ is a consistent family of probability measures defined  from $({\bf Q}_t)_{t\ge0}$ by the Girsanov kernel $\frac{1}{\nu-1}\theta,$ i.e., $\mathbb{P}_t$ is a probability measure on $\mathcal{F}_t$ given as
\begin{equation}
\label{eqn:applil_P_T}
\frac{d{\mathbb{P}_t}}{d{\bf Q}_t}=e^{-\frac{1}{\nu-1}\int_0^t\theta_u  dW_u-\frac{1}{2(\nu-1)^2}\int_0^t|\theta_u|^2\,du}\,.
\end{equation}
Here, the following condition was assumed for this change of measures.
\begin{assume}\label{assume:application_martingale}
	 A local martingale
	 $$(e^{-\frac{1}{\nu-1}\int_0^t\theta_u  dW_u-\frac{1}{2(\nu-1)^2}\int_0^t|\theta_u|^2\,du})_{0\le t\le T}$$
	  is a martingale under the probability measure ${\bf Q}_T$ for each $T\ge0.$
\end{assume}
\noindent Then, a process $B$ defined by
\begin{equation}
\label{eqn:B}
B_t:=\frac{1}{\nu-1}\int_0^t\theta_u\,du+W_t,\,t\ge0
\end{equation}
is a Brownian motion
on the consistent probability space 
$(\Omega,\mathcal{F},(\mathcal{F}_t)_{t\ge0},(\mathbb{P}_t)_{t\ge0}).$

We can summarize the arguments of this section as follows.
\begin{prop}\label{prop:final_utility_max}
	Consider the Ito process model with  the bank account and the $d$ stocks stated in Eq.\eqref{eqn:Ito_G} and Eq.\eqref{eqn:Ito stocks}, respectively. 
	Under Assumptions \ref{assume:theta} and \ref{assume:application_martingale}, we define a consistent
	family of probability  measures  $(\mathbb{P}_t)_{t\ge0}$ 	
	by Eq.\eqref{eqn:applil_P_T}.
	Then, for the power utility function $u(x)=x^\nu/\nu,$  
  $\nu<0$ and the family of nonnegative wealth processes $\mathcal{X}$ in Eq.\eqref{eqn:wealth},	the  optimal expected utility  
 is	$$\max_{\Pi\in \mathcal{X}}\mathbb{E}^{\bf P}[U(\Pi_T)]= {u_T^{1-\nu}}/{\nu}$$
		where
		\begin{equation}
		\label{eqn:final_u_T}
		u_T=\mathbb{E}^\mathbb{P}[e^{\frac{\nu}{2(\nu-1)^2}\int_0^T|\theta_u|^2\,du-\frac{\nu}{\nu-1}\int_0^Tr_u\,du}]\,.
		\end{equation} 
\end{prop}

\subsubsection{Factor models}
\label{sec:factor_model}

In this section, we investigate the long-term sensitivity of the optimal expected utility under one-factor models driven by a multi-dimensional Brownian motion.
Recall that 
$Z=(Z_t^{(1)},\cdots,Z_t^{(d)})_{t\ge0}^\top$
is a   $d$-dimensional Brownian motion   on 
a consistent probability space  
$(\Omega,\mathcal{F},(\mathcal{F}_t)_{t\ge0},({\bf P}_t)_{t\ge0}).$   The probability measures $({\bf P}_t)_{t\ge0}$ are referred to as physical measures of the market.
\begin{enumerate}[(i)] 
	\item Let $\mathcal{D}$ be an open interval in $\mathbb{R}.$ A factor process is a one-dimensional Markov diffusion process $X$ with state space $\mathcal{D}$ satisfying 
\begin{equation}\label{eqn:factor_model} 
dX_t=k(X_t)\,dt+v(X_t)\,dZ_t\,,\;X_0=\xi
\end{equation}
where
the function 
$k:\mathcal{D}\to\mathbb{R}$
and the row vector function $v:\mathcal{D}\to\mathbb{R}^d$  
are continuously differentiable.
\item The process $(r_t)_{t\ge0}$ in Eq.\eqref{eqn:Ito_G} is given as
\begin{equation}\label{eqn:utiliy_factor_r}
r_t=r(X_t)
\end{equation} 
for continuous function $r(\cdot):\mathcal{D}\to\mathbb{R}.$ 
\item The processes $(\mu_t)_{t\ge0}$ and $(\sigma_t)_{t\ge0}$ in 
Eq.\eqref{eqn:Ito stocks} are given as
\begin{equation}\label{eqn:factor_mu_sigma}
\mu_t=\mu(X_t)\,,\;\sigma_t=\sigma(X_t)
\end{equation}
for continuously differentiable functions $\mu(\cdot):\mathcal{D}\to\mathbb{R}^d$ and $\sigma(\cdot):\mathcal{D}\to\mathbb{R}^{d\times d}.$ 
\end{enumerate} 
This gives the full description of the one-factor model.

The long-term sensitivity of the optimal expected utility in a factor model
can be manipulated   
to fit the underlying framework of this paper.
As discussed in Proposition \ref{prop:final_utility_max},
it suffices to analyze the expectation $u_T$ in Eq.\eqref{eqn:final_u_T}.
Under Assumptions \ref{assume:theta} and \ref{assume:application_martingale}, 
we have
\begin{equation} 
u_T=\mathbb{E}^\mathbb{P}[e^{\frac{\nu}{2(\nu-1)^2}\int_0^T|\theta(X_u)|^2\,du-\frac{\nu}{\nu-1}\int_0^Tr(X_u)\,du}]
\end{equation} 
where $\theta:\mathcal{D}\to\mathbb{R}$ is a function defined as  
$\theta(\cdot)=\sigma^{-1}(\cdot)(\mu(\cdot)-r(\cdot){\bf 1}).$
Note that the process $(\theta(X_t))_{t\ge0}$ is the market price of risk presented in Eq.\eqref{eqn:theta_utility_max}.
Observe that the process $(B_t)_{t\ge0}$ in Eq.\eqref{eqn:B} is a Brownian motion  under the consistent family of probability measures $(\mathbb{P}_t)_{t\ge0}$ and that $X$ satisfies
\begin{equation} 
\begin{aligned}
dX_t
&=k(X_t)\,dt+v(X_t)\,dZ_t\\
&=(k-v\theta)(X_t)\,dt+v(X_t)\,dW_t\\
&=\big(k-\frac{\nu}{\nu-1}v\theta\big)(X_t)\,dt+v(X_t)\,dB_t\,.
\end{aligned}
\end{equation}
By defining
$${\bf r}(\cdot)=-\frac{\nu}{2(\nu-1)^2}|\theta(\cdot)|^2+\frac{\nu}{\nu-1}r(\cdot)\,,$$
the expectation $u_T$ is expressed as
$u_T =\mathbb{E}^\mathbb{P}[e^{-\int_0^T{\bf r}(X_u)\,du}].$
Thus, the quadruple of functions 
\begin{equation}
\label{eqn:quadruple_utility_max}
\Big(k(\cdot)-\frac{\nu}{\nu-1}v(\cdot)\theta(\cdot),|v(\cdot)|, {\bf r}(\cdot),1\Big)
\end{equation}
and the above expectation $u_T$   fit  the underlying framework of this paper
  on the consistent probability space  
$(\Omega,\mathcal{F},(\mathcal{F}_t)_{t\ge0},(\mathbb{P}_t)_{t\ge0})$ having Brownian motion 
$$(B_t)_{t\ge0}:=\Big(\frac{1}{\nu-1}\int_0^t\theta_u\,du+W_t\Big)_{t\ge0}\,.$$

We now investigate several specific examples.
As one can see, the utility maximization problem is 
specified by the factor process $X,$ the short rate function $r$ and the function $\theta$ representing the market price of risk. Thus, it is more convenient to specify these rather than the functions $\mu$ and $\sigma$ in Eq.\eqref{eqn:factor_mu_sigma}.

\begin{ex}
\textnormal{(The Heston model)} Let $\mathcal{D}=(0,\infty).$ Assume that the short rate is zero and the stock price follows
\begin{equation}
\begin{aligned} 
&dS_t=\mu X_t S_t\,dt+\sqrt{X_t}S_t\,dZ_t^{(1)}\,,\;&&S_0>0\\
&dX_t= k(m-X_t)\,dt+ v\rho\sqrt{X_t}\,dZ_t^{(1)}+ v\sqrt{1-\rho^2}\sqrt{X_t}\,dZ_t^{(2)} \,,\;&&X_0=\xi 
\end{aligned}
\end{equation}
and that the market price of risk is $\theta(X_t)=(\mu\sqrt{X_t},0)^\top.$
Here, the parameters satisfy
$\mu,v\in\mathbb{R},$ $k,m,\xi>0$ and $-1\le \rho \le 1.$ 
Then, the quadruple of functions in Eq.\eqref{eqn:quadruple_utility_max} is $$\Big(km-(k+\frac{\nu}{\nu-1}v\rho\mu) x,v\sqrt{x},-\frac{\nu\mu^2}{2(\nu-1)^2}x+\frac{\nu r}{\nu-1},1\Big)\,,\;x\in\mathcal{D}$$ and the  optimal expected utility   $u_T$ is 
$$u_T =\mathbb{E}^\mathbb{P}[e^{\frac{\nu\mu^2}{2(\nu-1)^2}\int_0^TX_u\,du}]e^{-\frac{\nu r}{\nu-1}T}\,.$$
We can simplify this problem as follows. Define 
$$a=k+\frac{\nu}{\nu-1}v\rho\mu\,,\;b=km\,,\;\sigma=v\,,\;q=-\frac{\nu\mu^2}{2(\nu-1)^2}$$
and $p_T:=\mathbb{E}^\mathbb{P}[e^{-q\int_0^TX_u\,du}]$ (thus, $u_T=p_Te^{-\frac{\nu r}{\nu-1}T}$).
Then, the quadruple of functions
 $$(b-ax,\sigma\sqrt{x},qx,1)\,,\;x\in\mathcal{D}$$
and the expectation $p_T$ fit the underlying framework of this paper. 
The details of the long-term sensitivities are discussed in Section \ref{sec:CIR}.
\end{ex}

\begin{ex}
\textnormal{(The $3/2$ model)}
Let $\mathcal{D}=(0,\infty).$ Assume that the short rate is zero and the stock price follows
\begin{equation}
\begin{aligned} 
&dS_t=\mu X_t S_t\,dt+\sqrt{X_t}S_t\,dZ_t^{(1)}\,,\;&&S_0>0\\
&dX_t= k(m-X_t)X_t\,dt+ v\rho X_t^{3/2}\,dZ_t^{(1)}+ v\sqrt{1-\rho^2} X_t^{3/2}\,dZ_t^{(2)} \,,\;&&X_0=\xi>0
\end{aligned}
\end{equation}
and that the market price of risk is $\theta(X_t)=(\mu\sqrt{X_t},0)^\top.$
Then, the quadruple of functions in Eq.\eqref{eqn:quadruple_utility_max} is
$$\Big(km-(k+\frac{\nu}{\nu-1}v\rho\mu) x^2,vx^{3/2},-\frac{\nu\mu^2}{2(\nu-1)^2}x+\frac{\nu r}{\nu-1},1\Big)\,,\;x\in\mathcal{D}$$ and the   optimal expected utility   $u_T$ is 
$$u_T =\mathbb{E}^\mathbb{P}[e^{\frac{\nu\mu^2}{2(\nu-1)^2}\int_0^TX_u\,du}]e^{-\frac{\nu r}{\nu-1}T}\,.$$
We can simplify this problem as follows. Define 
$$a=k+\frac{\nu}{\nu-1}v\rho\mu\,,\;b=km\,,\;\sigma=v\,,\;q=-\frac{\nu\mu^2}{2(\nu-1)^2}$$
and $p_T:=\mathbb{E}^\mathbb{P}[e^{-q\int_0^TX_u\,du}]$ (thus, $u_T=p_Te^{-\frac{\nu r}{\nu-1}T}$).
Then, the quadruple of functions
$$(b-ax^2,\sigma x^{3/2},qx,1)\,,\;x\in\mathcal{D}$$
and the expectation $p_T$ fit the underlying framework of this paper.
The details of the long-term sensitivities are discussed in Section \ref{sec:ex_3/2}.
\end{ex}

\subsubsection{Local volatility models}
\label{sec:lvm}

We investigate the utility maximization problem in local volatility models.
Consider a consistent probability space  
$(\Omega,\mathcal{F},(\mathcal{F}_t)_{t\ge0},({\bf P}_t)_{t\ge0})$  having   one-dimensional Brownian motion $Z=(Z_t)_{t\ge0}.$ The filtration $(\mathcal{F}_t)_{t\ge0}$ satisfies the usual condition.  The probability measures $({\bf P}_t)_{t\ge0}$ are referred to as  physical measures of the market. A local volatility model is described as follows. 
\begin{enumerate}[(i)]
	\item Assume that the short interest rate in Eq.\eqref{eqn:Ito_G} is a constant $r\ge0.$
	\item Let $\mathcal{D}=(0,\infty).$  A stock price is a Markov diffusion process  $(S_t)_{t\ge0}$ with state space $\mathcal{D}.$ Assume that $S$ satisfies 
	\begin{equation} 
	\begin{aligned}
	\frac{dS_t}{S_t}= \mu(S_t)\,dt+\sigma(S_t)\,dZ_t 
	\end{aligned}
	\end{equation} 	for continuously differentiable functions $\mu(\cdot):\mathcal{D}\to\mathbb{R}$ and $\sigma(\cdot):\mathcal{D}\to\mathbb{R}.$
	In other words, the processes $(\mu_t)_{t\ge0}$ and $(\sigma_t)_{t\ge0}$ in 
	Eq.\eqref{eqn:Ito stocks} are given as
	\begin{equation} \label{eqn:local_vol_mu_sigma}
	\mu_t=\mu(S_t)\,,\;\sigma_t=\sigma(S_t)\,.
	\end{equation} 
\end{enumerate} 
This gives the full description of the local volatility model used in this paper.

The long-term sensitivity of the optimal expected utility
can be analyzed by the same argument in Section \ref{sec:factor_model}.
Under Assumptions \ref{assume:theta} and \ref{assume:application_martingale}, 
the  optimal expected utility   $u_T$ in Proposition \ref{prop:final_utility_max} is
\begin{equation} 
u_T=\mathbb{E}^\mathbb{P}[e^{\frac{\nu}{2(\nu-1)^2}\int_0^T\theta^2(S_u)\,du}]e^{-\frac{\nu r}{\nu-1}T}
\end{equation} 
where
$\theta(\cdot):=\sigma^{-1}(\cdot)(\mu(\cdot)-r).$
For convenience, we define  
$$p_T:=\mathbb{E}^\mathbb{P}[e^{\frac{\nu}{2(\nu-1)^2}\int_0^T \theta^2(S_u) \,du}]$$
so that $u_T=p_Te^{-\frac{\nu r}{\nu-1}T}.$
The process
$$(B_t)_{t\ge0}:=\Big(\frac{\nu}{\nu-1}\int_0^t\theta_u\,du+Z_t\Big)_{t\ge0}$$  
is a Brownian motion on the consistent probability space  
$(\Omega,\mathcal{F},(\mathcal{F}_t)_{t\ge0},(\mathbb{P}_t)_{t\ge0})$, and the stock price $S$ follows
\begin{equation} 
\begin{aligned}
\frac{dS_t}{S_t} 
&=\big(-\frac{1}{\nu-1}\mu(S_t)+\frac{\nu r}{\nu-1}\big)\,dt+\sigma(S_t)\,dB_t\,.
\end{aligned}
\end{equation}
Thus, the quadruple of functions 
\begin{equation} \label{eqn:quadruple_utility_max_lv}
\Big(-\frac{1}{\nu-1}\mu(\cdot)\cdot+\frac{\nu r}{\nu-1},\sigma(\cdot)\cdot, -\frac{\nu}{2(\nu-1)^2} \theta^2(\cdot) ,1\Big)
\end{equation}
and the above expectation $p_T$  fit the underlying framework of this paper
on the consistent probability space  
$(\Omega,\mathcal{F},(\mathcal{F}_t)_{t\ge0},(\mathbb{P}_t)_{t\ge0})$ having Brownian motion 
$(B_t)_{t\ge0}.$

\begin{ex} 
\textnormal{(The CEV model)}
Let $\mathcal{D}=(0,\infty).$  Assume that the stock price follows the CEV model, which is given as a solution of
$$\frac{dS_t}{S_t}=k\,dt+\sigma {S_t}^{\beta}\,dB_t\,,\;X_0=\xi$$
for $\beta,k,\sigma,\xi>0.$ Then, the market price of risk is $\theta(S_t)=\frac{k-r}{\sigma}S_t^{-\beta}.$
The quadruple of functions in Eq.\eqref{eqn:quadruple_utility_max_lv} is
$$\Big(\frac{\nu r-k}{\nu-1}x,\sigma x^{\beta+1}, -\frac{(k-r)^2\nu}{2\sigma^2(\nu-1)^2} x^{-2\beta} ,1\Big)\,,\;x\in\mathcal{D}$$ and the expectation $p_T$ is 
$$p_T=\mathbb{E}^\mathbb{P}[e^{\frac{(k-r)^2\nu}{2\sigma^2(\nu-1)^2}\int_0^T S_u^{-2\beta} \,du}]$$
(thus, the  optimal expected utility  is $u_T =p_Te^{-\frac{\nu r}{\nu-1}T}$).
We can simplify this problem as follows. Define 
$$\mu=\frac{\nu r-k}{\nu-1}\,,\;q=-\frac{(k-r)^2\nu}{2\sigma^2(\nu-1)^2} \,.$$
Then, the quadruple of functions
$$(\mu x ,\sigma x^{\beta+1},qx^{-2\beta},1)\,,\;x\in\mathcal{D}$$
and the expectation $p_T=\mathbb{E}^\mathbb{P}[e^{-q\int_0^TS_u^{-2\beta}\,du}]$ fit the underlying framework of this paper.
The details of the long-term sensitivities are discussed in Section \ref{sec:CEV_n}.
\end{ex}

\subsection{Entropic risk measures}

In this section, we investigate the long-term sensitivity of  entropic risk measures.
The entropic risk measure of a portfolio value $\Pi_T$ 
is defined as
$$\rho(\Pi_T)=\frac{1}{\nu}\ln\mathbb{E}(e^{-\nu \Pi_T})$$
for the risk aversion parameter $\nu>0.$
The main purpose of this section is to measure
the extent to which the entropic risk measure
is affected by small perturbations of the underlying model.
First, we discuss how to formulate the entropic risk measure
under general Ito process models. Then, we consider specific models in Sections \ref{sec:ent_factor}
and \ref{sec:ent_local}.

The entropic risk measure can be expressed in a manageable manner as follows.
Recall the Ito process model described in Eq.\eqref{eqn:Ito_G}
and Eq.\eqref{eqn:Ito stocks}.
In this section, the short rate is assumed to be zero so that the bank account in Eq.\eqref{eqn:Ito_G} is identically equal to one.
Without loss of generality, we consider only portfolios with zero initial capital.
For given  self-financing portfolio $\pi,$  
the value process $(\Pi_t)_{t\ge0}=(\Pi_t^{(\pi)})_{t\ge0}$ 
  is  
\begin{equation} \label{eqn:value_process}
\Pi_t^{(\pi)}=\int_0^t\pi_u\,dS_u
\end{equation} 
as presented in Eq.\eqref{eqn:port_value}.
Then,
\begin{equation}
\begin{aligned}
\rho(\Pi_T)
&=\frac{1}{\nu}\ln \mathbb{E}^{\bf P}(e^{-\nu \Pi_T})\\
&=\frac{1}{\nu}\ln \mathbb{E}^{\bf P}(e^{-\nu \int_0^T\pi_u D(S_u)\mu_u\,du-\nu\int_0^T \pi _u D(S_u)\sigma_u\,dZ_u})\\
&=\frac{1}{\nu}\ln \mathbb{E}^{\mathbb{P}}(e^{-\nu \int_0^T(\pi_u D(S_u)\mu_u-\frac{1}{2}\nu |\pi _u D(S_u)\sigma_u|^2)\,du})
\end{aligned}
\end{equation}
where $(\mathbb{P}_t)_{t\ge0}$ is a consistent family of probability measures defined by
\begin{equation}
\label{eqn:appli_entropic}
\frac{d\mathbb{P}_t}{d{\bf P}_t}=e^{-\nu\int_0^t \pi _u D(S_u)\sigma_u\,dZ_u-\frac{1}{2}\nu^2\int_0^t | \pi _u D(S_u)\sigma_u|^2\,du}\,.
\end{equation} 
Here, the following condition was assumed for this change of measures.
\begin{assume}\label{assume:entropic}
 A local martingale
$$(e^{-\nu\int_0^t \pi _u D(S_u)\sigma_u\,dZ_u-\frac{1}{2}\nu^2\int_0^t | \pi _u D(S_u)\sigma_u|^2\,du})_{0\le t\le T}$$
is a martingale under the  physical measure ${\bf P}_T$ for each $T\ge0.$ 
\end{assume}

We can summarize the above-mentioned arguments as follows.

\begin{prop}\label{prop:entropic}
	Consider the Ito process model with zero short rate and   the $d$ stocks in   Eq.\eqref{eqn:Ito stocks},   and let $\Pi_T^{(\pi)}$ be the value at time $T$ of a portfolio $\pi$  presented in Eq.\eqref{eqn:value_process}.
	Under Assumption  \ref{assume:entropic},   we define a 
	consistent family of probability  measures  $(\mathbb{P}_t)_{t\ge0}$ 	
	by Eq.\eqref{eqn:appli_entropic}. 
Then, the entropic risk measure of  $\Pi_T^{(\pi)}$ 
	is 
	\begin{equation}
	\label{eqn:final_entropic}
	\rho(\Pi_T^{(\pi)})=\frac{1}{\nu}\ln \mathbb{E}^{\mathbb{P}}(e^{-\nu \int_0^T(\pi_u D(S_u)\mu_u-\frac{1}{2}\nu |\pi _u D(S_u)\sigma_u|^2)\,du})\,.
	\end{equation} 
    The process
	\begin{equation}
	B_t:= Z_t+ \nu\int_0^t\sigma_u^\top D(S_u)\pi _u^\top\,du\,,\;t\ge0
	\end{equation}
	is a Brownian motion
	on the consistent probability space  
	$(\Omega,\mathcal{F},(\mathcal{F}_t)_{t\ge0},(\mathbb{P}_t)_{t\ge0}).$ 
\end{prop}

\subsubsection{Factor models}
\label{sec:ent_factor}

We consider the factor model as a specific case to 
investigate the entropic risk measure of portfolios.
The factor process $X$ is given by Eq.\eqref{eqn:factor_model},
\begin{equation} 
dX_t=k(X_t)\,dt+v(X_t)\,dZ_t\,,\;X_0=\xi,
\end{equation}
and the drift and volatility functions, $\mu(\cdot)$ and $\sigma(\cdot)$, respectively, are given by Eq.\eqref{eqn:factor_mu_sigma}. The short rate  in Eq.\eqref{eqn:utiliy_factor_r} is assumed to be zero.
We consider a portfolio $\pi$ such that $\pi_t D(S_t)$ is determined by the factor $X_t$; more precisely, there is a
continuously differentiable function   $\eta:\mathbb{R}\to\mathbb{R}^d$ such that 
$$\pi_tD(S_t)=\eta(X_t)\,,\;t\ge0$$ 
and $\eta(X) D^{-1}(S)$ is $S$-integrable. 
By Eq.\eqref{eqn:final_entropic}, the entropic risk measure is
	\begin{equation}
	\begin{aligned}
\label{eqn:factor_entropic}
\rho(\Pi_T^{(\pi)})
&=\frac{1}{\nu}\ln \mathbb{E}^{\mathbb{P}}(e^{-\nu \int_0^T(\eta(X_u)\mu(X_u)-\frac{1}{2}\nu |\eta(X_u)\sigma(X_u)|^2)\,du})\\ 
&=\frac{1}{\nu}\ln \mathbb{E}^{\mathbb{P}}(e^{-\nu \int_0^T{\bf r}(X_u)\,du})  
\end{aligned}
\end{equation} 
where
$${\bf r}(\cdot)=\eta(\cdot)\mu(\cdot)-\frac{1}{2}\nu |\eta(\cdot)\sigma(\cdot)|^2$$
and the factor process $X$ satisfies
\begin{equation}
dX_t=(k(X_t)- \nu v(X_t)\sigma^\top(X_t)\eta^\top(X_t)) \,dt+v(X_t)\,dB_t 
\end{equation}
for $(\mathbb{P}_t)_{t\ge0}$-Brownian motion $(B_t)_{t\ge0}.$
In conclusion, the quadruple of functions 
\begin{equation}\label{eqn:ent_factor}
\Big(k(\cdot)-\nu v(\cdot)\sigma^\top(\cdot)\eta^\top(\cdot),|v(\cdot)|, \nu {\bf r}(\cdot),1\Big)
\end{equation}
and the  expectation 
$u_T:=\mathbb{E}^{\mathbb{P}}[e^{-\nu \int_0^T{\bf r}(X_u)\,du}]$ (thus, the entropic risk measure is $\rho(\Pi_T^{(\pi)})=\frac{1}{\nu}\ln u_T$)   fit  the underlying framework of this paper.

\begin{ex}
	 \textnormal{(Constant proportion portfolios in an affine model)} We consider the entropic risk measure of a constant proportion portfolio $\eta(\cdot)=\overline{\eta}=(\overline{\eta}_1,\cdots,\overline{\eta}_d)\in\mathbb{R}^d$ in an affine model.
	 Assume that the factor process $X$ satisfies
	 $$dX_t= k(m-X_t)\,dt+ \sqrt{X}_t v \,dZ_t  \,,\;X_0=\xi $$
	  for  $k,m,\xi>0$ and $v=(v_1,\cdots,v_d)
	  \in\mathbb{R}^d$ and that the drift and the volatility of the $d$ stocks described in    Eq.\eqref{eqn:factor_mu_sigma} are given as
	 $$\mu(X_t):=(\mu_i+\gamma_iX_t)_{1\le i\le d}\,,\;\sigma(X_t):=(\sqrt{\delta_{ij}+\varsigma_{ij}X_t})_{1\leq i,j\leq d}$$
	 for constants $\mu_i,\gamma_i\in\mathbb{R},\delta_{ij},\varsigma_{ij}\ge0,1\le i,j\leq d.$ Then, the quadruple in Eq.\eqref{eqn:ent_factor} is
	$$\Big(k(m-x)-\nu v\sqrt{x}\sigma^\top(x)\overline{\eta}^\top,|v|\sqrt{x}, \nu (\overline{\eta}\mu(x)-\frac{1}{2}\nu |\overline{\eta}\sigma(x)|^2),1\Big)\,,\;x>0$$ 	 
	 and the  entropic risk measure is $\rho(\Pi_T^{(\pi)})=\frac{1}{\nu}\ln u_T$
	 where
	 $$u_T:=\mathbb{E}^{\mathbb{P}}(e^{-\nu \int_0^T(\overline{\eta}\mu(X_u)-\frac{1}{2}\nu |\overline{\eta}\sigma(X_u)|^2)\,du})\,.  $$
	 As a specific example, let $\delta_{ij}=0$ for  $1\le i,j\leq d.$ Then, this problem can be simplified as follows. Define 
		 $$b=mk,\,a=k+\nu\sum_{i,j=1}^dv_i\sqrt{\varsigma_{ij}}\,\overline{\eta}_j,\,\sigma=|v|,\,q=\nu \big(\sum_{i=1}^d\overline{\eta}_i\gamma_i-\frac{1}{2}\nu\sum_{j}\big(\sum_{i}\overline{\eta}_i\sqrt{\varsigma_{ij}}\big)^2\big)\,,$$
	and we consider the case $q>0.$ 
	 The quadruple of functions
	 $$(b-ax,\sigma\sqrt{x},qx,1)\,,\;x>0$$
	 and the expectation $p_T:=\mathbb{E}^\mathbb{P}[e^{-q\int_0^TX_u\,du}]$ (thus, $u_T=p_Te^{-T\nu \sum_{i=1}^d\overline{\eta}_i\mu_i}$) fit the underlying framework of this paper. The details of the long-term sensitivities are discussed in Section \ref{sec:CIR}. 
\end{ex}

\subsubsection{Local volatility models}
\label{sec:ent_local}

This section investigates the entropic risk measure of  portfolios  in  the local volatility model presented in Eq.\eqref{eqn:local_vol_mu_sigma}.
We consider  portfolios  determined by the stock 
price, i.e., $\pi_t=\pi(S_t)$ for continuously differentiable function $\pi(\cdot).$  
By Eq.\eqref{eqn:final_entropic}, the entropic risk measure is
\begin{equation}
\begin{aligned}
\label{eqn:local_vol_entropic}
\rho(\Pi_T^{(\pi)})&=\frac{1}{\nu}\ln \mathbb{E}^{\mathbb{P}}(e^{-\nu \int_0^T(\pi(S_u) \mu(S_u)S_u-\frac{1}{2}\nu \pi^2(S_u)\sigma^2(S_u)S_u^2)\,du})\\
&=\frac{1}{\nu}\ln \mathbb{E}^{\mathbb{P}}(e^{-\nu \int_0^T{\bf r}(S_u)\,du})  
\end{aligned}
\end{equation} 
where
$${\bf r}(\cdot)=\pi(\cdot) \mu(\cdot)\cdot-\frac{1}{2}\nu \pi^2(\cdot)\sigma^2(\cdot)\,\cdot^2$$
and the stock price process $S$ satisfies
\begin{equation}\label{eqn:ent_risk_local_model}
dS_t=(\mu(S_t)- \nu \pi(S_t) \sigma^2(S_t)S_t)S_t \,dt+\sigma(S_t)S_t\,dB_t 
\end{equation}
for $(\mathbb{P}_t)_{t\ge0}$-Brownian motion $(B_t)_{t\ge0}.$
In conclusion, the   quadruple of functions 
\begin{equation}\label{eqn:ent_risk_local_vol}
\Big((\mu(\cdot)- \nu   \pi(\cdot) \sigma^2(\cdot)\,\cdot\,)\,\cdot,\sigma(\cdot)\,\cdot\,, {\bf r}(\cdot),1\Big)
\end{equation}
and the  expectation 
$u_T:=\mathbb{E}^{\mathbb{P}}[e^{-\nu \int_0^T{\bf r}(X_u)\,du}]$ fit the underlying framework of this paper.

\begin{ex}
	\textnormal{(Constant proportion portfolios I)} 
	We investigate the entropic risk measure of a portfolio
	in the $3/2$ model.
	Assume that the stock price follows
	\begin{equation}
	\label{eqn:ent_risk_ex}
	dS_t= k(m-S_t)S_t\,dt+ v  S_t^{3/2}\,dZ_t  \,,\;S_0>0
	\end{equation} 
	for $m,k,\xi>0$ and $v\neq0.$ In this example, we consider a constant proportion portfolio  $\pi(\cdot)\,\cdot^2= \overline{\eta}\in\mathbb{R}$ satisfying $0<\overline{\eta}<\frac{km}{\nu v^2}.$
	Then, the quadruple in Eq.\eqref{eqn:ent_risk_local_vol} is
	$$\Big((km- \nu v^2 \overline{\eta}-kx)x,vx^{3/2}, -k\overline{\eta}+\overline{\eta}(km -\frac{1}{2}\nu v^2 \overline{\eta})x^{-1},1\Big)\,,\;x>0$$ 
	and the  entropic risk measure is $\rho(\Pi_T^{(\pi)})=\frac{1}{\nu}\ln u_T$
	where
	$$u_T:=\mathbb{E}^{\mathbb{P}}(e^{-\nu\overline{\eta}(km -\frac{1}{2}\nu v^2 \overline{\eta})\int_0^T  S_u^{-1}\,du}) e^{\nu k\overline{\eta}T}\,.  $$	
    It is noteworthy that the process $S$ satisfies
	$$dS_t=(km- \nu v^2 \overline{\eta}-kS_t)S_t\,dt+ v  S_t^{3/2}\,dB_t$$
	as presented in Eq.\eqref{eqn:ent_risk_local_model} and  the process $X:=1/S$ satisfies 
	$$dX_t=(k+v^2-(km- \nu v^2 \overline{\eta})X_t)\,dt-v\sqrt{X_t}\,dB_t\,.$$
This problem can be simplified on the basis of the process $X:=1/S$ being a CIR model. Define 
$$b=k+v^2,\,a=km- \nu v^2 \overline{\eta},\,\sigma=-v,\,q=\nu\overline{\eta}(km -\frac{1}{2}\nu v^2 \overline{\eta})\,,$$
then the quadruple of functions
$$(b-ax,\sigma\sqrt{x},qx,1)\,,\;x>0$$
and the expectation $p_T:=\mathbb{E}^\mathbb{P}[e^{-q\int_0^TX_u\,du}]$ (thus, $u_T=p_Te^{\nu k\overline{\eta}T}$) fit the underlying framework of this paper. The details of the long-term sensitivities are discussed in Section \ref{sec:CIR}. 
\end{ex}

\begin{ex}
	\textnormal{(Constant proportion portfolios II)} 
	We investigate the entropic risk measure of another  portfolio  
	in the $3/2$ model.
	Assume that the stock price follows Eq.\eqref{eqn:ent_risk_ex}
		for $m,k,\xi>0$ and $v\neq0.$ 
	In this example, we consider a constant proportion portfolio  $\pi(\cdot)\,\cdot= \overline{\eta}\in\mathbb{R}$ satisfying  
	$-\frac{k}{\nu v^2}<\overline{\eta}<0.$
	Then, the quadruple in Eq.\eqref{eqn:ent_risk_local_vol} is
	$$\Big((km-(k+ \nu v^2 \overline{\eta}) x)x,vx^{3/2}, mk\overline{\eta}-\overline{\eta}(k+\frac{1}{2}\nu v^2 \overline{\eta})x,1\Big)\,,\;x>0$$ 	 
	and the  entropic risk measure is $\rho(\Pi_T^{(\pi)})=\frac{1}{\nu}\ln u_T$
	where
	$$u_T:=\mathbb{E}^{\mathbb{P}}(e^{ \nu\overline{\eta}(k+\frac{1}{2}\nu v^2 \overline{\eta})\int_0^T  S_u\,du}) e^{-\nu mk\overline{\eta}T}\,.  $$	
We can simplify this problem as follows. Define 
	$$a=k+ \nu v^2 \overline{\eta}\,,\;b=km\,,\;\sigma=v\,,\;q=-\nu\overline{\eta}(k+\frac{1}{2}\nu v^2 \overline{\eta})\,.$$
	Then, the quadruple of functions
	$$(b-ax^2,\sigma x^{3/2},qx,1)\,,\;x\in\mathcal{D}$$
	and the expectation   $p_T:=\mathbb{E}^\mathbb{P}[e^{-q\int_0^TS_u\,du}]$ (thus, $u_T=p_Te^{-\nu mk\overline{\eta}T}$) fit the underlying framework of this paper. The details of the  long-term sensitivities are discussed in Section \ref{sec:ex_3/2}.
\end{ex}

\subsection{Bond prices}

In this section, we study the long-term sensitivity of bond prices whose underlying short rate is modeled by a Markov diffusion.
Let $(\mathbb{P}_t)_{t\ge0}$ be a consistent family of risk-neutral measures and let $X$ be a short interest rate process given by 
$$dX_t=k(X_t)\,dt+v(X_t)\,dB_t\,,\;X_0=\xi $$
for a $(\mathbb{P}_t)_{t\ge0}$-Brownian motion $(B_t)_{t\ge0}.$
Then, the bond price with maturity $T$ is given by
$$p_T:=\mathbb{E}^\mathbb{P}[e^{-\int_0^TX_s\,ds}]\,.$$
The quadruple of functions
$(k(\cdot),v(\cdot),\,\cdot\,,1)$
and $p_T$ satisfy the underlying framework of this paper.

\begin{ex}
	\textnormal{(The CIR model)} Assume that the short rate follows the CIR model
	$$dX_t=(b-aX_t)\,dt+\sigma \sqrt{X_t}\,dB_t\,,\;X_0=\xi$$
	for $a,\xi>0,\sigma\neq0,2b>\sigma^2.$ Then, the bond price is
	$p_T=\mathbb{E}^\mathbb{P}[e^{-\int_0^TX_s\,ds}].$
	The long-term sensitivity of $p_T$ is analyzed in Section \ref{sec:CIR}.
\end{ex}

\begin{ex}
	\textnormal{(The $3/2$ model)} Assume that the short rate follows the $3/2$ model
$$dX_t=(b-aX_t)X_t\,dt+\sigma {X_t}^{3/2}\,dB_t\,,\;X_0=\xi$$
for $b,\sigma,\xi>0$ and $a>-\sigma^2/2.$  Then, the bond price is
$p_T=\mathbb{E}^\mathbb{P}[e^{-\int_0^TX_s\,ds}].$
The long-term sensitivity of $p_T$ is analyzed in Section \ref{sec:ex_3/2}.
\end{ex}

\section{Examples}
\label{sec:ex}

Concrete examples are studied in this section.
We describe three specific models: the CIR model, the 3/2 model and the CEV model.

\subsection{CIR model}
\label{sec:CIR}

We consider the CIR model  
$$dX_t=(b-aX_t)\,dt+\sigma \sqrt{X_t}\,dB_t\,,\;X_0=\xi$$
for $a,\xi>0,\sigma\neq0,2b>\sigma^2.$   As is well known, the process $X$ stays positive, thus we put the domain  $\mathcal{D}=(0,\infty).$  
The expectation $p_T$ of our interest is
$$p_T=\mathbb{E}_\xi^\mathbb{P}[e^{-q\int_0^TX_s\,ds}h(X_T)] $$
where $q>0$
and $h:\mathcal{D}\to\mathbb{R}$ is a nonzero, nonnegative, twice differentiable function with polynomial growth $h,$ $h',$ $h''.$
The process $X$ and the expectation $p_T$ are 
specified by
the quadruple of functions $(b-ax,\sigma\sqrt{x},qx,h(x)),x\in\mathcal{D}.$

We are interested in the large-time behavior of $p_T$ for small perturbations of the parameters $\xi,$ $b,$ $a$ and $\sigma.$ 
 It can be shown that the corresponding recurrent eigenpairs are
	$$(\lambda,\phi(x))=(b\eta, e^{-\eta x})\,,\;(\hat{\lambda},\hat{\phi}(x))=(\tilde{\lambda},\tilde{\phi}(x))=(\alpha,1)$$  
	where  
	$$\alpha:=\sqrt{a^2+2q\sigma^2}\,,\;
	\eta:=\frac{\alpha-a}{\sigma^2}\,.$$ 
For the long-term first-order and second-order sensitivities with respect to the initial-value, we have
	\begin{equation} 
	\begin{aligned}
	\left|\frac{\partial_\xi p_T}{p_T}+\eta\right|
\leq c e^{-\alpha T}\,,\;\left|\frac{{\partial_{\xi\xi}} p_T}{p_T}-\eta^2\right|
\leq ce^{-\alpha T} 
	\end{aligned}
	\end{equation}
 for some positive constant $c.$ 
We can also provide  a higher-order convergence rate of the second-order sensitivity 
as  
$$\left|\frac{{\partial_{\xi\xi}} p_T}{p_T}-\Big(\frac{\partial_\xi p_T}{p_T}\Big)^2\right|
\le ce^{-2\alpha T}  $$
for some positive constant $c.$
The long-term sensitivities with respect to the parameters $b,$ $a$ and $\sigma$ are described as
\begin{equation}
\begin{aligned}
&\left|\frac{1}{T}\frac{\partial_b p_T}{p_T}+\eta \right|\leq \frac{c}{T}\\
&\left|\frac{1}{T}\frac{\partial_a p_T}{p_T}+\frac{b}{\sigma^2}\Big(\frac{a}{\alpha}-1\Big)\right| \leq \frac{c}{T}\\
&\left|\frac{1}{T}\frac{\partial_\sigma p_T}{p_T}+2b\Big(\frac{q}{\alpha\sigma}-\frac{\alpha-a}{\sigma^3}\Big)\right| \leq \frac{c}{T}
\end{aligned}
\end{equation}
for some positive constant $c.$
For the proofs of these asymptotic behaviors, see Appendix \ref{app:CIR_model}.


\subsection{3/2 model}
\label{sec:ex_3/2}

In this section, we 
consider the $3/2$ model  
$$dX_t=(b-aX_t)X_t\,dt+\sigma {X_t}^{3/2}\,dB_t\,,\;X_0=\xi$$
for $b,\sigma,\xi>0$ and $a>-\sigma^2/2.$ 
As is well known, the process $X$ stays positive, thus we put the domain  $\mathcal{D}=(0,\infty).$  
The expectation $p_T$ of our interest is
$$p_T=\mathbb{E}_\xi^\mathbb{P}[e^{-q\int_0^TX_s\,ds}h(X_T)] $$
where $q>0$
and 
$h:\mathcal{D}\to\mathbb{R}$ is a nonzero, nonnegative, twice differentiable function with polynomial growth $h,$ $h',$ $h''.$
The process $X$ and the expectation $p_T$ are 
specified by
the quadruple of functions $((b-ax)x,\sigma x^{3/2},qx,h(x)),x\in\mathcal{D}.$

We are interested in the large-time behavior of $p_T$ for small perturbations of the parameters $\xi,$ $b,$ $a$ and $\sigma.$
 It can be shown that the corresponding recurrent eigenpairs are
$$(\lambda,\phi(x)):=(b\eta, x^{-\eta})\,,\;(\hat{\lambda},\hat{\phi}(x))=(\tilde{\lambda},\tilde{\phi}(x))=(b,x^{-2})$$  
where  
$$\eta:=\frac{\sqrt{(a+\sigma^2/2)^2+2q\sigma^2}-(a+\sigma^2/2)}{\sigma^2}\,.$$
For the long-term first-order and second-order  sensitivities with respect to the initial-value, we have
\begin{equation} \label{eqn:ex_3/2_initial}
\begin{aligned}
\left|\frac{\partial_\xi p_T}{p_T}+\frac{\,\eta\,}{\xi}\right|
\leq c  e^{-bT}\,,\;\left|\frac{{\partial_{\xi\xi}} p_T}{p_T}
-\frac{\eta(\eta+1)}{\xi^2}\right|
\leq ce^{-bT}
\end{aligned}
\end{equation}
for some positive constant $c.$ 
We can also provide  a higher-order convergence  rate of the second-order sensitivity 
as  
\begin{equation}\label{eqn:ex_3/2_initial_second}
\begin{aligned}
\left|\frac{\partial_{\xi\xi}p_T}{p_T}-\Big(\frac{\partial_\xi p_T}{p_T}\Big)^2+\frac{2}{\xi}\frac{\partial_\xi p_T}{p_T}+\frac{\eta}{\xi^2}\right|
\leq ce^{-2bT} 
\end{aligned}
\end{equation}
for some positive constant $c.$
The long-term sensitivities with respect to the parameters $b,$ $a$ and $\sigma$ are described as
\begin{equation}\label{eqn:ex_3/2_para_sen}
\begin{aligned}
&\left|\frac{1}{T}\frac{\partial_b p_T}{p_T}+\eta \right|\leq \frac{c}{T}\\
&\left|\frac{1}{T}\frac{\partial_a p_T}{p_T}-\frac{b}{\sigma^2}\frac{\sqrt{(a+\sigma^2/2)^2+2q\sigma^2}-(a+\sigma^2/2)}{\sqrt{(a+\sigma^2/2)^2+2q\sigma^2}}
\right|\leq \frac{c}{T}\\
&\left|\frac{1}{T}\frac{\partial_\sigma p_T}{p_T}+\frac{b(a+{\sigma^2}/{2}+2 q-\sqrt{(a+{\sigma^2}/{2})^2+2 q\sigma^2})}{\sigma \sqrt{(a+{\sigma^2}/{2})^2+2 q\sigma^2}} -\frac{2b}{\sigma^3}(\sqrt{(a+{\sigma^2}/{2})^2+2 q \sigma^2}-a- {\sigma^2}/{2})a \right|\leq \frac{c}{T}
\end{aligned}
\end{equation}
for some positive constant $c.$
For the proofs of these asymptotic behaviors, see Appendix \ref{sec:3/2}.

\subsection{CEV model I}
\label{sec:CEV_n}

In this section, we 
consider the CEV model
$$\frac{dX_t}{X_t}=(\mu-\theta X_t^{2\beta})\,dt+\sigma {X_t}^{\beta}\,dB_t\,,\;X_0=\xi$$ 
for $\beta,\mu,\xi>0,$ $\sigma\neq0,$ $\theta\ge0.$
As is well known, the process $X$ stays positive, thus we put the domain  $\mathcal{D}=(0,\infty).$  
When $\theta=0,$ the process $X$ is the standard CEV model.
The expectation $p_T$ of our interest is
$$p_T:=\mathbb{E}^\mathbb{P}[e^{-q\int_0^TX_t^{-2\beta}\,dt}]$$
for $q>0.$
The process $X$ and the expectation $p_T$ are 
specified by
the quadruple of functions  $((\mu-\theta x^{2\beta})x ,\sigma x^{\beta+1},qx^{-2\beta},1),x\in\mathcal{D}.$

The CEV model can be transformed into the CIR model by defining $Y_t=X_t^{-2\beta}.$ By the Ito formula,
$$dY_t=(b-aY_t)\,dt+\Sigma \sqrt{Y_t}\,dB_t\,,\;Y_0=\xi^{-2\beta}$$
where
$$b=2\beta\theta+\beta(2\beta+1)\sigma^2\,,\;a=2\beta\mu\,,\;\Sigma=-2\beta\sigma\,.$$
Note that $b>\Sigma^2/2$ so that the Feller condition is satisfied.
The expectation $p_T$ is   $$p_T=\mathbb{E}^\mathbb{P}[e^{-q\int_0^TY_t\,dt}]\,.$$
The sensitivities of the expectation $p_T$ 
for the CIR model $Y$ have already been analyzed in Section \ref{sec:CIR},  thus we use the previous results. 
The corresponding recurrent eigenpairs are
	$$(\lambda,\phi(x))=\big(\theta\eta+\big(\beta+\frac{1}{2}\big)\sigma^2\eta, e^{-\frac{\eta}{2\beta}x^{-2\beta}}\big)\,,\;(\hat{\lambda},\hat{\phi}(x))=(\tilde{\lambda},\tilde{\phi}(x))=(2\beta\sqrt{\mu^2+2q\sigma^2},1)$$
where    
$\eta:=\frac{\sqrt{\mu^2+2q\sigma^2}-\mu}{\sigma^2}.$

For the long-term first-order and second-order sensitivities with respect to the  initial-value, we have
\begin{equation}  
\begin{aligned}
	\left|\frac{\partial_\xi p_T}{p_T}-2\beta\eta\xi^{-2\beta-1}\right|
	\leq c  e^{-2\beta\sqrt{\mu^2+2q\sigma^2} T}\,,\;\left|\frac{{\partial_{\xi\xi}} p_T}{p_T}
	-2\beta\eta(2\beta\eta\xi^{-2\beta}-2\beta-1)\xi^{-2\beta-2}\right|
	\leq ce^{-2\beta\sqrt{\mu^2+2q\sigma^2}T}
\end{aligned}
\end{equation}
for some positive constant $c.$
We can also provide  a higher-order convergence rate of the second-order sensitivity 
as  
\begin{equation} 
\begin{aligned}
\left|\frac{{\partial_{\xi\xi}} p_T}{p_T}-\Big(\frac{\partial_\xi p_T}{p_T}\Big)^2 
+2\beta\eta(2\beta+1)\xi^{-2\beta-2}\right|
\leq ce^{-4\beta\sqrt{\mu^2+2q\sigma^2}T} \,.
\end{aligned}
\end{equation}
The long-term sensitivities with respect to the parameters $\mu,$ $\theta,$ $\sigma$ and $\beta$ are described as 
\begin{equation} 
\begin{aligned}
&\left|\frac{1}{T}\frac{\partial_\mu p_T}{p_T}-\Big(\frac{\theta}{\sigma^2}+\beta+\frac{1}{2}\Big)\frac{\sqrt{\mu^2+2q\sigma^2}-\mu}{\sqrt{\mu^2+2q\sigma^2}} \right|\leq \frac{c}{T}\\
&\left|\frac{1}{T}\frac{\partial_\theta p_T}{p_T}+\frac{\sqrt{\mu^2+2q\sigma^2}-\mu}{\sigma^2}
\right|\leq \frac{c}{T}\\
&\left|\frac{1}{T}\frac{\partial_\sigma p_T}{p_T}+\frac{q\sigma^4(2\beta+1)-2\theta(\mu^2+q\sigma^2)}{\sigma^3\sqrt{\mu^2+2q\sigma^2}}+\frac{2\theta\mu}{\sigma^3} \right|\leq \frac{c}{T}\\
&\left|\frac{1}{T}\frac{\partial_\beta p_T}{p_T}+\sqrt{\mu^2+2q\sigma^2}-\mu\right|\leq \frac{c}{T}
\end{aligned}
\end{equation}
for some positive constant $c.$

\subsection{CEV model II}
\label{sec:CEV_p}

We consider the CEV model
$$\frac{dX_t}{X_t}=(\mu-\theta X_t^{2\beta})\,dt+\sigma {X_t}^{\beta}\,dB_t\,,\;X_0=\xi$$ 
for $\beta,\mu,\xi>0,$ $\sigma\neq0,$ $\theta\ge0.$
Since the process $X$ stays positive, we put the domain  $\mathcal{D}=(0,\infty).$  
The expectation $p_T$ of our interest is
$$p_T:=\mathbb{E}^\mathbb{P}[e^{-q\int_0^TX_t^{2\beta}\,dt}]$$
for $q>0.$
The process $X$ and 
the expectation $p_T$ 
are specified by
the quadruple of functions $((\mu-\theta x^{2\beta})x ,\sigma x^{\beta+1},qx^{2\beta},1),x\in\mathcal{D}.$

The CEV model can be transformed into the $3/2$ model by defining $Y_t=X_t^{2\beta}.$ By the Ito formula, 
$$dY_t=(b-aY_t)Y_t\,dt+\Sigma Y_t^{3/2}\,dB_t\,,\;Y_0=\xi^{2\beta}$$
where
$$b=2\beta\mu\,,\;a=2\beta\theta-\beta(2\beta-1)\sigma^2\,,\;\Sigma=2\beta\sigma\,.$$
Note that $a>-\Sigma^2/2.$
The expectation $p_T$ is $$p_T=\mathbb{E}^\mathbb{P}[e^{-q\int_0^TY_t\,dt}]\,.$$
The sensitivities of the expectation $p_T$ 
for the $3/2$ model $Y$ have already been analyzed in Section \ref{sec:ex_3/2},  thus we use the previous results. 
The corresponding recurrent eigenpairs are
$$(\lambda,\phi(x))=(\mu\eta, x^{-2\beta\eta})\,,\;(\hat{\lambda},\hat{\phi}(x))=(\tilde{\lambda},\tilde{\phi}(x))=(2\beta\mu,x^{-4\beta})$$   
where  
$$\eta:=\frac{\sqrt{(\theta+\sigma^2/2)^2+2q\sigma^2}-(\theta+\sigma^2/2)}{\sigma^2}\,.$$

For the long-term first-order and second-order   sensitivities with respect to the initial-value, we have
\begin{equation}  
\begin{aligned}
\left|\frac{\partial_\xi p_T}{p_T}+\frac{2\beta\eta}{\xi}\right|
\leq c  e^{-2\beta\mu T}\,,\;\left|\frac{{\partial_{\xi\xi}} p_T}{p_T}
-\frac{2\beta\eta(2\beta\eta+1)}{\xi^2}\right|
\leq ce^{-2\beta\mu T}
\end{aligned}
\end{equation}
for some positive constant $c.$
We can also provide  a higher-order convergence rate of the second-order sensitivity 
as  
\begin{equation} 
\begin{aligned}
	\left|\frac{{\partial_{\xi\xi}} p_T}{p_T}-\Big(\frac{\partial_\xi p_T}{p_T}\Big)^2 
+\frac{4\beta}{\xi} \frac{\partial_\xi p_T}{p_T} + \frac{2\beta\eta(4\beta-1)}{\xi^2}\right|
\leq ce^{-4\beta\mu T} \,.
\end{aligned}
\end{equation}
The long-term sensitivities with respect to the parameters $\mu,$ $\theta,$ $\sigma$ and $\beta$ are described as 
\begin{equation} 
\begin{aligned}
&\left|\frac{1}{T}\frac{\partial_\mu p_T}{p_T}+\eta \right|\leq \frac{c}{T}\\
&\left|\frac{1}{T}\frac{\partial_\theta p_T}{p_T}-\frac{\mu}{\sigma^2}\frac{\sqrt{(\theta+\sigma^2/2)^2+2q\sigma^2}-(\theta+\sigma^2/2)}{\sqrt{(\theta+\sigma^2/2)^2+2q\sigma^2}}
\right|\leq \frac{c}{T}\\
&\left|\frac{1}{T}\frac{\partial_\sigma p_T}{p_T}+\frac{\mu(\theta+{\sigma^2}/{2}+2 q-\sqrt{(\theta+{\sigma^2}/{2})^2+2 q\sigma^2})}{\sigma \sqrt{(\theta+{\sigma^2}/{2})^2+2 q\sigma^2}} -\frac{2\mu}{\sigma^3}(\sqrt{(\theta+{\sigma^2}/{2})^2+2 q \sigma^2}-\theta- {\sigma^2}/{2})\theta \right|\leq \frac{c}{T}\\
&\left|\frac{\partial_\beta p_T}{p_T} \right|\leq c
\end{aligned}
\end{equation}
for some positive constant $c.$

\section{Conclusion}
\label{sec:con}

This paper investigated the large-time asymptotic behavior of the sensitivities of cash flows.
The price of cash flows 
is given 
in expectation form as
\begin{equation}
\label{eqn:conclusion_p_T}
p_T=\mathbb{E}_\xi^{\mathbb{P}}[e^{-\int_0^T r(X_s)\,ds} h(X_T)]\,.\end{equation}
We studied the extent to which this expectation is affected by small changes of the underlying Markov diffusion $X.$
The main idea is a PDE representation of the expectation 
by 
incorporating  the Hansen--Scheinkman decomposition  method.
The sensitivities of long-term cash flows and their large-time convergence rates can be represented via simple expressions
in terms of   eigenvalues and eigenfunctions of the pricing operator $h\mapsto \mathbb{E}_\xi^{\mathbb{P}}[e^{-\int_0^T r(X_s)\,ds} h(X_T)] $.

Essentially, we demonstrated two types of long-term sensitivities.
First, the first-order and second-order sensitivities with respect to the initial value $\xi=X_0$ were investigated.
Using the Hansen--Scheinkman decomposition, 
we can express the expectation $p_T$ as 
\begin{equation*} 
\begin{aligned}
p_T=\phi(\xi)\,e^{-\lambda T}f(T,\xi) 
\end{aligned}
\end{equation*} 
with recurrent eigenpair $(\lambda,\phi)$ and remainder function
$f(T,\xi).$
Applying 
the Hansen--Scheinkman decomposition repeatedly, the derivative  $f_x(T,\xi)$ has the decomposition 
\begin{equation}
\begin{aligned}
f_x(T,\xi)
&=\hat{\phi}(\xi)e^{-\hat{\lambda} T}\hat{f}(T,\xi)
\end{aligned}
\end{equation}
with recurrent eigenpair $(\hat{\lambda},\hat{\phi})$ and  remainder function
$\hat{f}(T,\xi).$
Under appropriate conditions,  
the first-order sensitivity 
and its convergence rate with respect to the initial value are given by
\begin{equation} 
\left|\frac{\partial_\xi p_T}{p_T}
-\frac{\phi'(\xi)}{\phi(\xi)}\right|\leq  ce^{-\hat{\lambda} T}\,,\;T\ge0
\end{equation} 
for some positive constant $c,$ which is independent of $T.$
For the second-order sensitivity respect to the initial value, a similar expression  is obtained. We have
\begin{equation}
\left|\frac{{\partial_{\xi\xi}} p_T}{p_T}-\Big(\frac{\partial_\xi p_T}{p_T}\Big)^2-\frac{\phi''(\xi)}{\phi(\xi)}+\Big(\frac{\phi'(\xi)}{\phi(\xi)}\Big)^2
-\frac{\hat{\phi}'(\xi)}{\hat{\phi}(\xi)} \frac{\partial_\xi p_T}{p_T} +\frac{\hat{\phi}'(\xi)}{\hat{\phi}(\xi)} \frac{\phi'(\xi)}{\phi(\xi)}\right|\leq c(e^{-\tilde{\lambda} T}+e^{-\hat{\lambda} T})e^{-\hat{\lambda} T}
\end{equation}
for some positive constant $c,$ which is independent of $T,$ where $\tilde{\lambda}$ is a recurrent eigenvalue.

Second, the sensitivities
with respect to the drift and diffusion terms were demonstrated.
From the Hansen--Scheinkman decomposition,
the perturbed expectation $p_T^\epsilon=\mathbb{E}_\xi^{\mathbb{P}}[e^{-\int_0^T r(X_s^\epsilon)\,ds} f(X_T^\epsilon)]$ induced by the perturbed process $X^\epsilon$ is 
expressed as 

$$p_T^{(\epsilon)}=\phi^{(\epsilon)}(\xi)\,e^{-\lambda^{(\epsilon)} T}f^{(\epsilon)}(T,\xi)$$
with recurrent eigenpair $(\lambda^{(\epsilon)},\phi^{(\epsilon)})$ and remainder function  $f^{(\epsilon)}(T,\xi).$ 
The long-term sensitivity of $p_T^{(\epsilon)}$ with respect to the perturbation parameter $\epsilon$  can be expressed in a simple form as 
$$\left|\frac{1}{T}\frac{\partial}{\partial \epsilon}\Big|_{\epsilon=0} \ln p_T^{(\epsilon)}+\frac{\partial}{\partial \epsilon}\Big|_{\epsilon=0}\lambda^{(\epsilon)}\right|
\leq  \frac{c}{T}$$
for some positive constant $c,$ 
which is independent of $T.$

We presented applications of these results to three practical problems: utility maximization, entropic risk measures and bond prices.
Under factor models and local volatility models, these problems can be transformed into the expectation form in Eq.\eqref{eqn:conclusion_p_T}. 
As specific examples, explicit formulas for several market models, namely the CIR model, the 3/2 model and the CEV model, were investigated. \newline

\noindent\textbf{Acknowledgement.}\\ 
The author sincerely appreciates the valuable suggestions received from
the	Editor, the Assistant Editor and the anonymous referee for their helpful comments and insights
that have greatly improved the quality of the paper.
This research was supported by the National Research Foundation of Korea (NRF) grants funded by the Ministry of Science and ICT  (No. 2018R1C1B5085491 and No. 2017R1A5A1015626) 
and the Ministry of Education   (No. 2019R1A6A1A10073437) through Basic Science Research Program.

\appendices

\section{CIR model}
\label{app:CIR_model}

Let  $(\Omega,\mathcal{F},(\mathcal{F}_t)_{t\ge0},(\mathbb{P}_t)_{t\ge0})$ be a consistent probability space that has  
a one-dimensional  Brownian motion $B=(B_t)_{t\ge0}.$ 
The filtration $(\mathcal{F}_t)_{t\ge0}$ is the completed filtration generated by $B.$
The CIR model is a process given as a solution of
$$dX_t=(b-aX_t)\,dt+\sigma \sqrt{X_t}\,dB_t\,,\;X_0=\xi$$
for $a,\sigma,\xi>0$ and $2b>\sigma^2.$
For $q>0$ and a nonzero nonnegative function $h$ with polynomial growth, we define
$$p_T=\mathbb{E}^\mathbb{P}[e^{-q\int_0^TX_s\,ds} h(X_T)]\,.$$
It can be shown that the quadruple of functions
$$(b-ax,\sigma\sqrt{x},qx,h(x))\,,\;x>0$$
satisfies Assumptions \ref{assume:a} -- \ref{assume:h}.
The recurrent eigenpair is
$$(\lambda,\phi(x)):=(b\eta, e^{-\eta x})$$
where  
$$\alpha:=\sqrt{a^2+2q\sigma^2}\,,\;
\eta:=\frac{\alpha-a}{\sigma^2}\,.$$ 
Under the consistent family of recurrent eigen-measures $(\hat{\mathbb{P}}_t)_{t\ge0},$ 
the process
$$\hat{B}_t=\sigma\eta\int_0^t\sqrt{X_s}\,ds+B_t\,,\;t\ge0$$
is a Brownian motion, and $X$ follows
$$dX_t= (b-\alpha X_t)\,dt+\sigma\sqrt{X_t}\,d\hat{B}_t\,.$$

Using the Hansen--Scheinkman decomposition, we have
\begin{equation}\label{eqn:CIR_decompo}
\begin{aligned}
p_T&=\mathbb{E}_\xi^\mathbb{P}[e^{-q\int_0^TX_s\,ds} h(X_T)]
=\mathbb{E}_\xi^{\hat{\mathbb{P}}}[h(X_T)e^{\eta X_T}] \,e^{-\eta\xi}\,e^{-\lambda T}\,.
\end{aligned}
\end{equation}
For $t\in [0,\infty)$ and $x>0,$ the remainder function is 
\begin{equation}\label{eqn:CIR_reminder}
f(t,x)=\mathbb{E}_x^{\hat{\mathbb{P}}}[h(X_t)e^{\eta X_t}] 
\end{equation} 
so that $p_T=f(T,\xi)e^{-\eta\xi}e^{-\lambda T}.$
For nonzero and nonnegative $h$ with polynomial growth, it is easy to show that
$f(T,\xi)$ converges to a positive constant as $T\to\infty$
by using Lemma \ref{lem:CIR_MMG}.
It is also easy to check that $f$ is $C^{1,2}$ 
by considering the density function of $X_t.$  
We will investigate the behavior of the function $f(T,\xi)$ by expressing this function as a 
solution of a second-order differential equation.
Using the Feynman--Kac formula, the function $f$ satisfies
\begin{equation}\label{eqn:CIR_FK}
-f_t+\frac{1}{2}\sigma^2xf_{xx}+(b-\alpha x)f_x=0\,,\;f(0,x)=h(x)e^{\eta x}\,.
\end{equation}

\begin{lemma}\label{lem:CIR_MMG} Let $\hat{B}$ be a Brownian motion on the consistent probability space
	  $(\Omega,\mathcal{F},(\mathcal{F}_t)_{t\ge0},(\hat{\mathbb{P}}_t)_{t\ge0}).$
Suppose that $X$ is a solution of
$$dX_t= (b-\alpha X_t)\,dt+\sigma\sqrt{X_t}\,d\hat{B}_t\,,\;X_0=\xi$$
where $\alpha,\sigma,\xi>0$ and $2b>\sigma^2.$ Then,
for $\beta<2\alpha/\sigma^2$, we have
$$\mathbb{E}^{\hat{\mathbb{P}}}[e^{\beta X_T}]=\Big(\frac{1}{1-\beta c(T)}\Big)^{2b/\sigma^2}e^{\frac{\beta }{1-\beta c(T)}e^{-\alpha T}\xi}$$
where $c(T):=\sigma^2 (1-e^{-\alpha T})/2\alpha.$ 
Thus, in this case,
$$\mathbb{E}^{\hat{\mathbb{P}}}[e^{\beta X_T}]\leq \Big(\frac{1}{1-\beta \sigma^2/2\alpha}\Big)^{2b/\sigma^2}e^{\frac{\beta }{1-\beta\sigma^2/2\alpha}e^{-\alpha T}\xi}\,,$$
and
$$\lim_{T\to\infty}\mathbb{E}^{\hat{\mathbb{P}}}[e^{\beta X_T}]=\Big(\frac{1}{1-\beta\sigma^2/2\alpha}\Big)^{2b/\sigma^2}\,.$$
\end{lemma} 
\noindent Refer to  \cite[Corollary 6.3.4.4]{jeanblanc2009mathematical} for the proof.


\subsection{First-order sensitivity of $\xi$}
\label{app:delta}

We estimate the large-time asymptotic behavior of the first-order sensitivity of $p_T$ with respect to the initial value $\xi.$
In this section, assume that $h$ is continuously differentiable and that $h$ and $h'$ have polynomial growth. 
From Eq.\eqref{eqn:CIR_decompo}, it follows that
\begin{equation} \label{eqn:CIR_delta}
\begin{aligned}
\frac{\partial_\xi  p_T}{p_T} 
=\frac{  f_x(T,\xi)}{f(T,\xi)} -\eta \,.
\end{aligned}
\end{equation}
The function $f_x(t,x)$ satisfies
 \begin{equation}
 \label{eqn:CIR_delta_PDE}
 -f_{xt}+\frac{1}{2}\sigma^2x f_{xxx}+\Big(b+\frac{1}{2}\sigma^2-\alpha x\Big)f_{xx}
 -\alpha f_x=0\,,\;f_x(0,x)=(h'(x)+\eta h(x))e^{\eta x}\,,
 \end{equation} 
which is  obtained from Eq.\eqref{eqn:CIR_FK}  by taking the differentiation in $x.$ 
Note that since $f$ is $C^{1,2}$   and every coefficient is continuously differentiable in $x$ in  Eq.\eqref{eqn:CIR_FK}, the function $f$ is thrice continuously differentiable in $x.$ It is easy to show that
the quadruple of functions $$(b+\sigma^2/2-\alpha x,\sigma\sqrt{x},\alpha, (h'+\eta h)e^{\eta x})\,,\;x>0$$ satisfies Assumptions \ref{assume:a} -- \ref{assume:h}.
The corresponding process 
  $\hat{X}$ is the solution of 
$$d\hat{X}_t=(b+\sigma^2/2-\alpha \hat{X}_t)\,dt+\sigma {\hat{X}_t}^{1/2}\,d\hat{B}_t\,. $$

\begin{lemma} \label{lemma:CIR_x_remainder}
	The remainder function $f$ satisfies
	\begin{equation}
	\label{eqn:CIR_delta_decom}
	f_x(t,x)=\mathbb{E}^{\hat{\mathbb{P}}}[(h'(\hat{X}_t)+\eta h(\hat{X}_t))e^{\eta \hat{X}_t}|\hat{X}_0=x]e^{-\alpha t}
	\end{equation}
	for $x>0$ and $t\ge0.$
\end{lemma}

\begin{proof} 
Define
$g(t,x):=f_x(t,x)\phi(x).$
Eq.\eqref{eqn:CIR_delta_PDE} gives \begin{equation}
\begin{aligned}
&-g_t+\frac{1}{2}\sigma^2xg_{xx}+(b+\sigma^2/2-a x)g_x+\Big(-\frac{1}{2}(\alpha+a)\eta x+(b+\sigma^2/2)\eta-\alpha\Big)g=0\\
&g(0,x)=h'(x)+\eta h(x)\,.
\end{aligned}
\end{equation}
Consider a consistent family $(\mathbb{Q}_t)_{t\ge0}$ of probability measures  where each 
$\mathbb{Q}_t$ is a probability measure on $\mathcal{F}_t$ defined as 
\begin{equation} 
\frac{d\mathbb{Q}_t}{d\hat{\mathbb{P}}_t}=\frac{\phi(\hat{X}_0)}{\phi(\hat{X}_t)}e^{\int_0^t\frac{\hat{\mathcal{L}}\phi(\hat{X}_s)}{\phi(\hat{X}_s)}\,ds}=e^{\eta(\hat{X}_t-x)+\int_0^t\frac{1}{2}(\alpha+a)\eta \hat{X}_s
	-(b+\sigma^2/2)\eta\,ds}=e^{-\frac{1}{2}\sigma^2\eta^2\int_0^t\hat{X}_s\,ds+\sigma\eta \int_0^t \hat{X}_s^{1/2}\,d\hat{B}_s}\,.
\end{equation} 
It is easy to check that  $\hat{X}$ satisfies
\begin{equation}
\begin{aligned}
d\hat{X}_t
&=(b+\sigma^2/2-\alpha \hat{X}_t)\,dt+\sigma {\hat{X}_t}^{1/2}\,d\hat{B}_t\\
&=(b+\sigma^2/2-a\hat{X}_t)\,dt+\sigma {\hat{X}_t}^{1/2}\,d{B}_t^{\mathbb{Q}} 
\end{aligned}
\end{equation}
for a $(\mathbb{Q}_t)_{t\ge0}$-Brownian motion $(B_t^{\mathbb{Q}})_{t\ge0}.$
By \cite[Theorem 5.1.8]{pinsky1995positive}, since this process does not reach the boundaries under the consistent family of probability measures $(\mathbb{Q}_t)_{t\ge0},$ the $\mathbb{Q}_T$-local martingale 
$(e^{-\frac{1}{2}\sigma^2\eta^2\int_0^t\hat{X}_s\,ds+\sigma\eta \int_0^t \hat{X}_s^{1/2}\,d\hat{B}_s})_{0\le t\le T}$
is a  $\mathbb{Q}_T$-martingale.
Observe that the operator $\hat{\mathcal{L}}$ given as
$$\hat{\mathcal{L}}f=\frac{1}{2}\sigma^2xf''(x)
+(b+\sigma^2/2-ax)f'(x)$$
is the infinitesimal generator of $\hat{X}$ under the consistent family of probability measures $(\mathbb{Q}_t)_{T\ge0},$
and for $\phi(x)=e^{-\eta x},$ we get $$\hat{\mathcal{L}}\phi(x)=\big(\frac{1}{2}(\alpha+a)\eta x
-(b+\sigma^2/2)\eta\big)e^{-\eta x}\,.$$
The Feynman--Kac formula (Remark \ref{remark:FK} or Proposition \ref{prop:1st_initial_sen_linear})
gives that
\begin{equation}
\begin{aligned}
f_x(t,x)\phi(x)=g(t,x)&=\mathbb{E}^{{\mathbb{Q}}}\big[
e^{\int_0^t-\frac{1}{2}(\alpha+a)\eta \hat{X}_s+(b+\sigma^2/2)\eta-\alpha\,ds}(h'(\hat{X}_t)+\eta h(\hat{X}_t))\big|\hat{X}_0=x\big]\\
&=\mathbb{E}^{{\mathbb{Q}}}\Big[
e^{-\int_0^t\frac{\hat{\mathcal{L}}\phi}{\phi}(\hat{X}_s)\,ds}\,\frac{\phi(\hat{X}_{t})}{\phi(\hat{X}_{0})}
e^{\int_0^t-\alpha\,ds}(h'(\hat{X}_t)+\eta h(\hat{X}_t))e^{\eta \hat{X}_t}\Big|\hat{X}_0=x\Big]\phi(x)\\
&=\mathbb{E}^{{\mathbb{Q}}}\Big[
\frac{d\hat{\mathbb{P}}_t}{d\mathbb{Q}_t}
e^{-\alpha t}(h'(\hat{X}_t)+\eta h(\hat{X}_t))e^{\eta \hat{X}_t}\Big|\hat{X}_0=x\Big]\phi(x)\\
&=\mathbb{E}^{\hat{\mathbb{P}}}\big[
(h'(\hat{X}_t)+\eta h(\hat{X}_t))e^{\eta \hat{X}_t}\big|\hat{X}_0=x\big]\phi(x)e^{-\alpha t}\,,
\end{aligned}
\end{equation}
which implies that
$$f_x(t,x)=\mathbb{E}^{\hat{\mathbb{P}}}\big[(h'(\hat{X}_t)+\eta h(\hat{X}_t))e^{\eta \hat{X}_t}\big|\hat{X}_0=x\big]e^{-\alpha t}\,.$$  
Condition (iv) of Proposition \ref{prop:1st_initial_sen_linear} can be confirmed from Lemma \ref{lem:CIR_MMG} and the density function of $X_t,$ and the other conditions are trivial.
\end{proof}

Using Eq.\eqref{eqn:CIR_delta_decom}, we can obtain the large-time behavior of $\partial_\xi p_T.$
 Since $h$ and $h'$ have polynomial growth,
for $\eta<\beta<2\alpha/\sigma^2,$ there is a positive constant $c_0=c_0(\beta)$ such that $|h'(x)+\eta h(x)|e^{\eta x}\leq c_0 e^{\beta x}$ for $x>0.$ From Lemma \ref{lem:CIR_MMG}, we have
\begin{equation}\label{eqn:CIR_delta_reminder}
\begin{aligned}
|f_x(t,x)|e^{\alpha t}\leq\mathbb{E}^{\hat{\mathbb{P}}}\Big[|h'(\hat{X}_t)+\eta h(\hat{X}_t)|e^{\eta \hat{X}_t}\Big|\hat{X}_0=x\Big]
&\le c_0\mathbb{E}^{\hat{\mathbb{P}}}[e^{\beta \hat{X}_t}|\hat{X}_0=x]\\
&\leq
c_0\Big(\frac{1}{1-\beta \sigma^2/2\alpha}\Big)^{1+2b/\sigma^2}e^{\frac{\beta }{1-\beta\sigma^2/2\alpha}e^{-\alpha t}x} \leq c_1
\end{aligned}
\end{equation}
for some positive constant $c_1$ which depends on $x$ but does not depend on $t.$
It follows that
\begin{equation}
\label{eqn:CIR_deri_x}
|f_x(t,x)|\leq c_1e^{-\alpha t}\,.
\end{equation} 
Therefore,
\begin{equation} 
\begin{aligned}
\left|\frac{\partial_\xi p_T}{p_T}+\eta\right|
=\left|\frac{  f_x(T,\xi)}{f(T,\xi)}  \right|\leq c_2 e^{-\alpha T}
\end{aligned}
\end{equation}
for some positive constant $c_2.$ This gives the desired result.

\subsection{Second-order sensitivity of $\xi$}

We analyze the second-order sensitivity  with respect to the initial value $\xi.$
In this section, assume that $h$ is twice continuously differentiable and that $h,$ $h',$ $h''$  have polynomial growth.
From Eq.\eqref{eqn:CIR_delta}, we know that
\begin{equation} \label{eqn:CIR_gamma}
\begin{aligned}
\frac{{\partial_{\xi\xi}} p_T}{p_T}
=\frac{  f_{xx}(T,\xi)}{f(T,\xi)} -2\eta \frac{f_x(T,\xi)}{f(T,\xi)} +\eta^2 \,.
\end{aligned}
\end{equation}
Since we already estimated the large-time asymptotic behavior of 
$f_x(T,\xi),$ it suffices to investigate the second-order derivative $f_{xx}(T,\xi).$
Define $\tilde{\mathbb{P}}:=\hat{\mathbb{P}}$ (to be consistent with the notations in Table \ref{Tab:Tcr}) and 
$$\hat{f}(t,x):=\mathbb{E}^{\tilde{\mathbb{P}}}[(h'(\hat{X}_t)+\eta h(\hat{X}_t))e^{\eta \hat{X}_t}|\hat{X}_0=x]$$
so that $f_x(t,x)=\hat{f}(t,x)e^{-\alpha t},$ which gives
\begin{equation}
\label{eqn:f_xx_hat_f}
f_{xx}(t,x)=\hat{f}_x(t,x)e^{-\alpha t}\,.
\end{equation} 
By the Feynman--Kac formula, we have 
$$-\hat{f}_t+\frac{1}{2}\sigma^2x\hat{f}_{xx}+(b+\sigma^2/2-\alpha x)\hat{f}_x=0\,,\;\hat{f}(0,x)=(h'(x)+\eta h(x))e^{\eta x}\,.$$
Since $\hat{f}$ is $C^{1,2}$   and every coefficient is continuously differentiable in $x$, the function $\hat{f}$ is thrice continuously differentiable in $x.$
Differentiate this PDE in $x,$ 
then
\begin{equation}\label{eqn:cir_2nd_initial_sens_hat_f_PDE}
\begin{aligned}
&-\hat{f}_{xt}+\frac{1}{2}\sigma^2x\hat{f}_{xxx}+(b+\sigma^2-\alpha x)\hat{f}_{xx}-\alpha \hat{f}_x=0\,,\;\hat{f}_x(0,x)=(h''(x)+2\eta h'(x)+\eta^2h(x))e^{\eta x}\,.
\end{aligned}
\end{equation}
It is easy to show that
the quadruple of functions $$(b+\sigma^2-\alpha x,\sigma\sqrt{x},\alpha, (h''(x)+2\eta h'(x)+\eta^2h(x))e^{\eta x})\,,\;x>0$$ satisfies Assumptions \ref{assume:a} -- \ref{assume:h}.
The corresponding process 
$\tilde{X}$ is a solution of 
$$d\tilde{X}_t=(b+\sigma^2-\alpha \tilde{X}_t)\,dt+\sigma {\tilde{X}_t}^{1/2}\,d\tilde{B}_t\,. $$
We observe that $\hat{f}_x(t,x)$ satisfies
\begin{equation}
\label{eqn:CIR_hat_f_x}
\hat{f}_x(t,x)=\mathbb{E}^{\tilde{\mathbb{P}}}[(h''(\tilde{X}_t)+2\eta h'(\tilde{X}_t)+\eta^2 h(\tilde{X}_t))e^{\eta \tilde{X}_t}|\tilde{X}_0=x]e^{-\alpha t}\,.
\end{equation} 
This is directly obtained from 
Proposition \ref{prop:1st_initial_sen_linear} by the same argument used in the derivation of Eq.\eqref{eqn:CIR_delta_decom}.

To analyze $\hat{f}_x(t,x),$
we apply the same argument used  in Eq.\eqref{eqn:CIR_delta_reminder} and Eq.\eqref{eqn:CIR_deri_x}. 
For $\eta<\beta<2\alpha/\sigma^2,$ there is a positive constant $c_0=c_0(\beta)$ such that $|(h''(x)+2\eta h'(x)+\eta^2 h(x))|e^{\eta x}\leq c_0 e^{\beta x}$ for $x>0.$ Using Lemma \ref{lem:CIR_MMG}, we have
\begin{equation} \label{eqn:CIR_gamma_h}
\begin{aligned}
|\hat{f}_x(t,x)|e^{\alpha t}
&\leq \mathbb{E}^{\tilde{\mathbb{P}}}\Big[|h''(\tilde{X}_t)+2\eta h'(\tilde{X}_t)+\eta^2 h(\tilde{X}_t)|e^{\eta \tilde{X}_t}\Big|\tilde{X}_0=x\Big]\\
&\le c_0\mathbb{E}^{\tilde{\mathbb{P}}}[e^{\beta \tilde{X}_t}|\tilde{X}_0=x]\\
&\leq
c_0\Big(\frac{1}{1-\beta \sigma^2/2\alpha}\Big)^{1+2b/\sigma^2}e^{\frac{\beta }{1-\beta\sigma^2/2\alpha}e^{-\alpha t}x} \leq c_1
\end{aligned}
\end{equation}
for some positive constant $c_1$ which depends on $x$ but does not depend on $t.$
It follows that $$|\hat{f}_x(t,x)|\leq c_1e^{-\alpha t}\,,$$
which gives
\begin{equation}
\label{eqn:CIR_f_xx}
|f_{xx}(t,x)|=|\hat{f}_x(t,x)|e^{-\alpha t}\leq c_1e^{-2\alpha  t}\,.
\end{equation} 
Therefore,
\begin{equation} 
\begin{aligned}
\left|\frac{{\partial_{\xi\xi}} p_T}{p_T}-\eta^2\right|
\le \left| \frac{  f_{\xi \xi}(T,\xi)}{f(T,\xi)}\right| +2\eta\left| \frac{f_\xi(T,\xi)}{f(T,\xi)}\right| \leq c_2e^{-\alpha T} 
\end{aligned}
\end{equation}
for some positive constant $c_2.$
We can also provide  a higher-order convergence rate as follows.
From Eq.\eqref{eqn:CIR_delta} and Eq.\eqref{eqn:CIR_gamma}, 
$$\left|\frac{{\partial_{\xi\xi}} p_T}{p_T}-\Big(\frac{\partial_\xi p_T}{p_T}\Big)^2\right|
\le \left| \frac{  f_{\xi \xi}(T,\xi)}{f(T,\xi)}\right| +\left( \frac{f_\xi(T,\xi)}{f(T,\xi)}\right)^2\le c_3e^{-2\alpha T}  $$
for some positive constant $c_3.$

\subsection{Sensitivity of $b$}

We investigate  the large-time asymptotic behavior of the  sensitivity  with respect to the parameter $b.$
In this section, assume that $h$ is continuously differentiable and that $h$ and $h'$ have polynomial growth. 
Using Eq.\eqref{eqn:CIR_decompo} and Eq.\eqref{eqn:CIR_reminder},
since $\eta$ is independent of $b$ and
$\lambda=b\eta,$ we have 
$$\frac{\partial_b p_T}{p_T}=\frac{f_b(T,\xi)}{f(T,\xi)}-\eta T\,.$$
It can be easily shown that $f$ is continuously differentiable in  $b$ by considering the density function of $X_t$ or  by using \cite[Theorem 4.8]{park2018sensitivity}.
We focus on the large-time behavior of $f_b(T,\xi).$  
Differentiate Eq.\eqref{eqn:CIR_FK} in $b,$ then
\begin{equation} \label{eqn:f_b_PDE_FK}
-f_{bt}+\frac{1}{2}\sigma^2xf_{bxx}+(b-\alpha x)f_{bx}+f_x=0\,,\;f_b(0,x)=0\,.
\end{equation}
From the Feynman--Kac formula, one can show that
\begin{equation}
\label{eqn:CIR_f_b_prob_expression}
f_b(t,x)
=\mathbb{E}_x^{\hat{\mathbb{P}}}\Big[\int_0^tf_x(t-s,X_s)\,ds\Big]
\end{equation} 
by the same method used in the proof of Lemma \ref{lemma:CIR_x_remainder}.

We can estimate the expectation on the right-hand side by using the same method in Section \ref{app:delta}.
For $\eta<\beta<\alpha/\sigma^2,$ Eq.\eqref{eqn:CIR_delta_decom} and
Eq.\eqref{eqn:CIR_delta_reminder} implies that
\begin{equation}\label{eqn:CIR_sens_b}  
\begin{aligned}
|f_x(t-s,x)|
&\leq
c_0\Big(\frac{1}{1-\beta \sigma^2/2\alpha}\Big)^{1+2b/\sigma^2}e^{\frac{\beta }{1-\beta\sigma^2/2\alpha}e^{-\alpha (t-s)}x}  e^{-\alpha (t-s)} \leq c_1e^{\gamma x}  e^{-\alpha (t-s)}
\end{aligned}
\end{equation}
where
$$c_1:=c_0\Big(\frac{1}{1-\beta \sigma^2/2\alpha}\Big)^{1+2b/\sigma^2},\;\gamma:=\frac{\beta }{1-\beta\sigma^2/2\alpha}\,.$$
Then,
$$|f_b(t,x)|
\le\mathbb{E}_x^{\hat{\mathbb{P}}}\Big[\int_0^t|f_x(t-s,X_s)|\,ds\Big]
= \int_0^t\mathbb{E}_x^{\hat{\mathbb{P}}}|f_x(t-s,X_s)|\,ds
\leq c_1\int_0^te^{-\alpha (t-s)} \mathbb{E}_x^{\hat{\mathbb{P}}}[e^{\gamma X_s}]\,ds\,.$$
Since $\gamma<2\alpha/\sigma^2,$ by Lemma \ref{lem:CIR_MMG}, the expectation  $\mathbb{E}_x^{\hat{\mathbb{P}}}[e^{\gamma X_s}]$
is bounded in $s$ on $[0,\infty).$
Thus, there is a positive constant $c_2$ such that $\mathbb{E}_x^{\hat{\mathbb{P}}}[e^{\gamma X_s}]
\leq c_2 \alpha/c_1,$ which gives 
$$|f_b(t,x)|
\leq c_2 \alpha \int_0^te^{-\alpha (t-s)}\,ds=c_2(1-e^{-\alpha t})\,.
$$
Since $f(T,\xi)$ converges to a positive constant as $T\to\infty,$ we conclude that  
$$\left|\frac{1}{T}\frac{\partial_b p_T}{p_T}+\eta \right|=\frac{1}{T}\left|\frac{f_b(T,\xi)}{f(T,\xi)}\right|\leq \frac{c_3}{T}$$
 for some positive constant $c_3.$

\subsection{Sensitivity of $a$}

We investigate the large-time asymptotic behavior of the  sensitivity of $p_T$ with respect to the parameter $a.$
In this section, assume that $h$ is continuously differentiable and that $h$ and $h'$ have polynomial growth. 
From Eq.\eqref{eqn:CIR_decompo} and Eq.\eqref{eqn:CIR_reminder},
it follows that
$$\frac{\partial_a p_T}{p_T}=\frac{f_a(T,\xi)}{f(T,\xi)}-\xi\partial_a\eta-T\partial_a\lambda \,.$$
It can be easily shown that $f$ is continuously differentiable in $a$ by considering the density function of $X_t$ or by using \cite[Theorem 4.8]{park2018sensitivity}.
We focus on the  large-time behavior of $f_a(T,\xi).$ 
Differentiate Eq.\eqref{eqn:CIR_FK} in $a,$ then
\begin{equation}\label{eqn:f_a_PDE_FK}
-f_{at}+\frac{1}{2}\sigma^2xf_{axx}+(b-\alpha x)f_{ax}-\frac{a}{\alpha} xf_x=0\,,\;f_a(0,x)=-\frac{\eta}{\alpha}h(x)xe^{\eta x}\,.
\end{equation}
Here, we used   $\alpha=\sqrt{a^2+2q\sigma^2}$ and $\eta=\frac{\alpha-a}{\sigma^2}.$
By the Feynman--Kac formula in  Remark \ref{remark:FK}, it follows that
\begin{equation}\label{eqn:sen_f_a}
\begin{aligned}
f_a(t,x)&=\mathbb{E}_x^{\hat{\mathbb{P}}}\Big[-\frac{a}{\alpha}\int_0^tX_sf_x(t-s,X_s)\,ds-\frac{\eta}{\alpha}h(X_t)X_te^{\eta X_t}\Big]\\
&=-\frac{a}{\alpha}\int_0^t\mathbb{E}_x^{\hat{\mathbb{P}}}[X_sf_x(t-s,X_s)]\,ds-\frac{\eta}{\alpha}\mathbb{E}_x^{\hat{\mathbb{P}}}[h(X_t)X_te^{\eta X_t}]\,.
\end{aligned}
\end{equation}
Note that Remark \ref{remark:FK} cannot be applied directly because the two terms $-\frac{a}{\alpha} xf_x$ and $f_a(0,x)=-\frac{\eta}{\alpha}h(x)xe^{\eta x}$  neither have polynomial growth nor are nonnegative. 
To overcome this problem, define $g(t,x)=-f_a(t,x),$ then
\begin{equation} \label{eqn:app_CIR_a}
-f_{t}+\frac{1}{2}\sigma^2xg_{xx}+(b-\alpha x)g_{x}+\frac{a}{\alpha} xf_x=0\,,\;g_a(0,x)=\frac{\eta}{\alpha}h(x)xe^{\eta x}\,.
\end{equation} 
Since this equation satisfies the hypothesis of Remark \ref{remark:FK}, we obtain Eq.\eqref{eqn:sen_f_a} by the same method  used in the proof of Lemma \ref{lemma:CIR_x_remainder}.

Now, we use Eq.\eqref{eqn:sen_f_a} to estimate $f_a(t,x).$
From Eq.\eqref{eqn:CIR_sens_b}, we know that
\begin{equation}
\label{eqn:CIR_a_xf_x}
|xf_x(t-s,x)|
\leq  c_1e^{\delta x}  e^{-\alpha (t-s)}
\end{equation} 
for some positive constants $c_1$ and $\gamma<\delta<2\alpha/\sigma^2.$
By Lemma \ref{lem:CIR_MMG}, the expectation  $\mathbb{E}_x^{\hat{\mathbb{P}}}[e^{\delta X_s}]$
is bounded in $s$ on $[0,\infty).$
Thus, there is a positive constant $c_2$ such that $\mathbb{E}_x^{\hat{\mathbb{P}}}[e^{\delta X_s}]
\leq c_2 \alpha/c_1,$ which gives 
\begin{equation}
\label{eqn:CIR_a_Xf}
\int_0^t\mathbb{E}_x^{\hat{\mathbb{P}}}|X_sf_x(t-s,X_s)|\,ds
\leq c_1\int_0^t e^{-\alpha (t-s)}\mathbb{E}_x^{\hat{\mathbb{P}}}[e^{\delta X_s}]\,ds\leq
c_2 \alpha\int_0^t e^{-\alpha (t-s)}\,ds=c_2(1-e^{-\alpha t})\leq c_2\,. 
\end{equation} 
The expectation 
$\mathbb{E}_x^{\hat{\mathbb{P}}}[h(X_t)X_te^{\eta X_t}]$ is also bounded in $t$ on $[0,\infty).$ We conclude
that 
$$\left|\frac{1}{T}\frac{\partial_a p_T}{p_T}+\partial_a\lambda\right|=\frac{1}{T}\left|\frac{f_a(T,\xi)}{f(T,\xi)}-\xi\partial_a\eta \right|\leq \frac{c_3}{T}$$
for some positive constant $c_3.$
Direct calculation gives 
$\partial_a\lambda=\frac{b}{\sigma^2}(\frac{a}{\alpha}-1).$

\subsection{Sensitivity of $\sigma$}

We study the large-time asymptotic behavior of the  sensitivity of $p_T$ with respect to the parameter $\sigma.$
In this section, assume that $h$ is twice continuously differentiable and that $h,$ $h',$ $h''$  have polynomial growth.
From Eq.\eqref{eqn:CIR_decompo} and Eq.\eqref{eqn:CIR_reminder},
it follows that
$$\frac{\partial_\sigma p_T}{p_T}=\frac{f_\sigma(T,\xi)}{f(T,\xi)}-\xi\partial_\sigma\eta-T\partial_\sigma\lambda \,.$$
It can be easily shown that $f$ is continuously differentiable in
$\sigma$ by considering the density function of $X_t$ or by using \cite[Theorem 4.13]{park2018sensitivity}.
We focus on the large-time behavior of $f_\sigma(T,\xi).$ 
Differentiate Eq.\eqref{eqn:CIR_FK} in $\sigma,$ then
\begin{equation}\label{eqn:f_sigma_CIR_PDE_FK}
-f_{\sigma t}+\frac{1}{2}\sigma^2xf_{\sigma xx}+(b-\alpha x)f_{\sigma x}+\sigma x f_{xx}-\frac{2q\sigma}{\alpha}xf_x=0\,,\;f_\sigma(0,x)=h(x)xe^{\eta x}\partial_\sigma\eta\,.
\end{equation}
By the same method used in the proof of Lemma \ref{lemma:CIR_x_remainder}, we have 
\begin{equation}\label{eqn:f_sigma_CIR_PDE}
\begin{aligned}
f_\sigma(t,x)&=\mathbb{E}_x^{\hat{\mathbb{P}}}\Big[\sigma\int_0^t
X_sf_{xx}(t-s,X_s)\,ds-\frac{2q\sigma}{\alpha}\int_0^tX_sf_x(t-s,X_s)\,ds+h(X_t)X_te^{\eta X_t}\partial_\sigma\eta\Big] \\
&=\sigma\int_0^t\mathbb{E}_x^{\hat{\mathbb{P}}}[
X_sf_{xx}(t-s,X_s)]\,ds-\frac{2q\sigma}{\alpha}\int_0^t
\mathbb{E}_x^{\hat{\mathbb{P}}}[X_sf_x(t-s,X_s)]\,ds+\partial_\sigma\eta
\mathbb{E}_x^{\hat{\mathbb{P}}}[h(X_t)X_te^{\eta X_t}] \,.
\end{aligned}
\end{equation}

We claim that $|f_\sigma(t,x)|$ is bounded in $t$ on $[0,\infty)$ by estimating the three terms on the right-hand side.
From Eq.\eqref{eqn:CIR_gamma_h} and Eq.\eqref{eqn:CIR_f_xx},
for $\eta<\beta<\alpha/\sigma^2,$ 
$$|f_{xx}(t-s,x)|\leq c_0\Big(\frac{1}{1-\beta \sigma^2/2\alpha}\Big)^{1+2b/\sigma^2}e^{\frac{\beta }{1-\beta\sigma^2/2\alpha}x}e^{-2\alpha (t-s)}\,.$$
For $\delta$ with 
$\frac{\beta }{1-\beta\sigma^2/2\alpha}<\delta<2\alpha/\sigma^2,$
there is a positive constant $c_1$ such that
$$|xf_{xx}(t-s,x)|\leq c_1e^{\delta x}e^{-2\alpha (t-s)}\,.$$
By the same analysis used in Eq.\eqref{eqn:CIR_a_xf_x} and Eq.\eqref{eqn:CIR_a_Xf}, it follows that
$$\int_0^t\mathbb{E}_x^{\hat{\mathbb{P}}}[
X_sf_{xx}(t-s,X_s)]\,ds$$
is bounded in $t$ on $[0,\infty).$
By Eq.\eqref{eqn:CIR_a_Xf}, the integral 
$$\int_0^t
\mathbb{E}_x^{\hat{\mathbb{P}}}[X_sf_x(t-s,X_s)]\,ds$$
is bounded in $t$ on $[0,\infty).$
The expectation 
$\mathbb{E}_x^{\hat{\mathbb{P}}}[h(X_t)X_te^{\eta X_t}]$ is also bounded in $t$ on $[0,\infty).$
Therefore, we have 
$|f_\sigma(t,x)|\leq c_2$ for some positive constant $c_2.$ 
In conclusion,
$$\left|\frac{1}{T}\frac{\partial_\sigma p_T}{p_T}+\partial_\sigma\lambda\right|=\frac{1}{T}\left|\frac{f_\sigma(T,\xi)}{f(T,\xi)}-\xi\partial_\sigma\eta\right|\leq \frac{c_3}{T} $$
for some positive constant $c_3.$ Direct calculation gives $\partial_\sigma\lambda=2b(\frac{q}{\alpha\sigma}-\frac{\alpha-a}{\sigma^3}).$

\section{3/2 model}
\label{sec:3/2}

The $3/2$ model is a process given as a solution of
$$dX_t=(b-aX_t)X_t\,dt+\sigma {X_t}^{3/2}\,dB_t\,,\;X_0=\xi$$
for $b,\sigma,\xi>0$ and $a>-\sigma^2/2.$ 
For $q>0$ and a nonzero, nonnegative  Borel function $h$ with linear growth at most, we define
$$p_T=\mathbb{E}^\mathbb{P}[e^{-q\int_0^TX_s\,ds} h(X_T)]\,.$$
It can be shown that the quadruple of functions
$$((b-ax)x,\sigma x^{3/2},qx,h)\,,\;x>0$$
satisfies Assumptions \ref{assume:a} -- \ref{assume:h}.
The recurrent eigenpair is 
$$(\lambda,\phi(x)):=(b\eta, x^{-\eta})$$
where  
$$\eta:=\frac{\sqrt{(a+\sigma^2/2)^2+2q\sigma^2}-(a+\sigma^2/2)}{\sigma^2}\,.$$ 
Under the consistent family of recurrent eigen-measures $(\hat{\mathbb{P}}_t)_{t\ge0},$ 
the process
$$\hat{B}_t=\sigma\eta\int_0^t\sqrt{X_s}\,ds+B_t\,,\;t\ge0$$
is a Brownian motion, and $X$ follows
\begin{equation}
\begin{aligned}
dX_t=(b-\alpha X_t)X_t\,dt+\sigma {X_t}^{3/2}\,d\hat{B}_t
\end{aligned}
\end{equation}
where
$\alpha:=a+\sigma^2\eta.$

Using this consistent family of recurrent eigen-measures, we have the Hansen--Scheinkman decomposition
\begin{equation} \label{eqn:3/2_HS_decompo}
\begin{aligned}
p_T&=\mathbb{E}_\xi^\mathbb{P}[e^{-q\int_0^TX_s\,ds} h(X_T)]
=\mathbb{E}_\xi^{\hat{\mathbb{P}}}[h(X_T)X_T^{\eta}] \,\xi^{-\eta}\,e^{-\lambda T}\,.
\end{aligned}
\end{equation}
For $t\in [0,\infty)$ and $x>0,$ we define
\begin{equation}\label{eqn:3/2_reminder}
f(t,x)=\mathbb{E}_x^{\hat{\mathbb{P}}}[h(X_t)X_t^{\eta}] 
\end{equation} 
so that \begin{equation}
\label{eqn:3/2_decom}
p_T=f(T,\xi)\phi(\xi)e^{-\lambda T}=f(T,\xi)\xi^{-\eta}e^{-\lambda T}.
\end{equation} 
For nonzero, nonnegative Borel function $h$ with linear growth at most, it is easy to show that $f(T,\xi) $ converges to a positive constant 
as $T\to\infty$
by Lemma \ref{lem:3/2_MMG}.
We will investigate the large-time behavior of the function $f(T,\xi)$ by expressing this function as a 
solution of a second-order differential equation.
By  the Feynman--Kac formula,   $f$ satisfies
\begin{equation}\label{eqn:3/2_FK}
-f_t+\frac{1}{2}\sigma^2x^3f_{xx}+(b-\alpha x)xf_x=0\,,\;f(0,x)=h(x)x^{\eta}\,.
\end{equation}

\begin{lemma}\label{lem:3/2_MMG} Let $\hat{B}=(\hat{B}_t)_{t\ge0}$ be a Brownian motion on the consistent probability space $(\Omega,\mathcal{F},(\mathcal{F}_t)_{t\ge0},(\hat{\mathbb{P}}_t)_{t\ge0}).$
	Suppose that $X$ is a solution of
	$$dX_t=(b-\alpha X_t)X_t\,dt+\sigma {X_t}^{3/2}\,d\hat{B}_t\,,\;X_0=\xi$$
	where $\alpha,b,\sigma,\xi>0.$
Then, for $A<\frac{2\alpha}{\sigma^2}+2,$ we have
\begin{equation}\label{eqn:3/2_moment}
\begin{aligned}
\mathbb{E}_\xi(X_T^A)
&=\frac{\Gamma(\frac{2\alpha}{\sigma^2}+2-A)}{\Gamma(\frac{2\alpha}{\sigma^2}+2)}\Big(\frac{2b}{\sigma^2}\frac{1}{1-e^{-bT}}\Big)^{A}
F\Big(A,\frac{2\alpha}{\sigma^2}+2,-\frac{2b}{\sigma^2}\frac{1}{(e^{bT}-1)\xi}\Big)\,,
\end{aligned}
\end{equation} 
and the expectation
$\mathbb{E}_\xi(X_T^A)$
converges to $$\frac{\Gamma(\frac{2\alpha}{\sigma^2}+2-A)}{\Gamma(\frac{2\alpha}{\sigma^2}+2)}\Big(\frac{2b}{\sigma^2}\Big)^{A}$$ as $T\to\infty$
where $F$ is the confluent hypergeometric function.
Moreover,
if $0<A<\frac{2\alpha}{\sigma^2}+2,$ then 
the map $H:[0,\infty)\times(0,\infty)\to \mathbb{R}$ defined by $H(t,x)=\mathbb{E}_x(X_t^A)$ is uniformly bounded on the domain $[0,\infty)\times(0,\infty).$
\end{lemma}

\begin{proof}
	Define a process $Y$ as $Y=1/X,$ then 
$$dY_t=(\theta-bY_t)\,dt-\sigma\sqrt{Y_t}\,d\hat{B}_t\,,\;Y_0=\zeta\,,$$
where $\theta:=\alpha+\sigma^2$ and $\zeta:=1/\xi.$ Since $\theta>\sigma^2/2,$ the Feller condition is satisfied.
From \cite[Theorem 3.1]{hurd2008explicit}
or \cite[Section 3]{dereich2011euler},
we have for
$A<\frac{2\theta}{\sigma^2}=\frac{2\alpha}{\sigma^2}+2,$
\begin{equation}
\begin{aligned}
\mathbb{E}(X_t^A)=\mathbb{E}(Y_t^{-A})
&=\frac{\Gamma(\frac{2\theta}{\sigma^2}-A)}{\Gamma(\frac{2\theta}{\sigma^2})}\Big(\frac{2b}{\sigma^2}\frac{1}{1-e^{-bt}}\Big)^{A}
F\Big(A,\frac{2\theta}{\sigma^2},-\frac{2b}{\sigma^2}\frac{\zeta}{e^{bt}-1}\Big)\,.
\end{aligned}
\end{equation} 
Since the  confluent hypergeometric function $F$ satisfies $$\lim_{t\to\infty}F(A,\frac{2\theta}{\sigma^2},-\frac{2b}{\sigma^2}\frac{\zeta}{e^{bt}-1})=1\,,$$ we obtain the desired result.
Moreover, if $A$ and $\frac{2\alpha}{\sigma^2}+2$ are positive, 
the function 
$z\mapsto F(A,\frac{2\alpha}{\sigma^2}+2,z)$ is uniformly bounded in $z$ on $(0,\infty),$ which is the direct result from \cite[13.1.5 on page 504]{abramowitz1965handbook}. This means that the map $H$ is uniformly bounded on the domain $[0,\infty)\times(0,\infty).$ 
\end{proof}

\subsection{First-order sensitivity of $\xi$}
\label{app:3/2_delta}

We estimate the large-time asymptotic behavior of the first-order sensitivity of $p_T$ with respect to the initial value $\xi.$
In this section, 
assume that $h$ is continuously differentiable and that $h$ and $h'$ have linear growth at most and are nonnegative (we are mainly interested in the case $h=1$).
From Eq.\eqref{eqn:3/2_HS_decompo}, it follows that
\begin{equation}  
\begin{aligned}
\frac{\partial_\xi  p_T}{p_T} 
=\frac{  f_x(T,\xi)}{f(T,\xi)} -\frac{\,\eta\,}{\xi} \,.
\end{aligned}
\end{equation}
We focus on the term $f_x(T,\xi).$
Since $f$ is $C^{1,2}$   and every coefficient is continuously differentiable in $x$ in Eq.\eqref{eqn:3/2_FK}, the function $f$ is thrice continuously differentiable in $x.$
This gives 
\begin{equation}
\label{eqn:app_f_x_pde}
-f_{xt}+\frac{1}{2}\sigma^2x^3f_{xxx}+(b-(\alpha-\frac{3}{2}\sigma^2)x)xf_{xx}+(b-2\alpha x)f_x=0\,,\;f_x(0,x)=(xh'(x)+\eta h(x))x^{\eta-1}\,.
\end{equation}
It is easy to show that
the quadruple of functions
\begin{equation}
\label{eqn:3/2_quad_initial_sen}
((b-(\alpha-\frac{3}{2}\sigma^2)x)x,\sigma {x}^{3/2},-b+2\alpha x, (xh'(x)+\eta h(x))x^{\eta-1})\,,\;x>0
\end{equation} 
satisfies Assumptions \ref{assume:a} -- \ref{assume:h}.
The corresponding process 
$\hat{X}$ is the solution of 
$$d\hat{X}_t=(b-(\alpha-\frac{3}{2}\sigma^2)\hat{X}_t)\hat{X}_t\,dt+\sigma \hat{X}_t^{3/2}\,d\hat{B}_t \,.$$

We show that the remainder function $f$ satisfies
\begin{equation}
\label{eqn:3/2_delta_decom}
f_x(t,x)=\mathbb{E}^{\hat{\mathbb{P}}}[(\hat{X}_th'(\hat{X}_t)+\eta h(\hat{X}_t))\hat{X}_t^{\eta-1}e^{-2\alpha\int_0^t\hat{X}_s\,ds}|\hat{X}_0=x]  e^{bt}\,.
\end{equation}
To achieve this equality, we use the  Feynman--Kac formula in Remark \ref{remark:FK}. However, Remark \ref{remark:FK} cannot be applied directly because the
volatility function $\sigma x^{3/2}$ of $\hat{X}$ does not have linear growth. Instead, we define $Y:=1/\hat{X}$ and  $g(t,y):=f_x(t,1/y).$ Then, $Y$ is a CIR process  
$$dY_t=((\alpha-\sigma^2/2)-bY_t)\,dt-\sigma\sqrt{Y_t}\,d\hat{B}_t\,,$$
and Eq.\eqref{eqn:app_f_x_pde} becomes
$$-g_t+\frac{1}{2}\sigma^2yg_{yy}+((\alpha-\sigma^2/2)-by)g_y+(b-2\alpha /y)g=0\,,\;g(0,y)=(y^{-1}h'(y^{-1})+\eta h(y^{-1}))y^{1-\eta}\,.$$
Now, the quadruple of functions
\begin{equation}\label{eqn:app_reciprocal_3/2}
\Big((\alpha-\sigma^2/2)-by,-\sigma\sqrt{y},-b+2\alpha /y,(y^{-1}h'(y^{-1})+\eta h(y^{-1}))y^{1-\eta}\Big)\,,\;y>0
\end{equation} 
satisfies all the conditions in Remark \ref{remark:FK}. We know that the function $\max_{0\le t\le T}|f(t,1/y)|$ (thus, $\max_{0\le t\le T}|g(t,y)|=\max_{0\le t\le T}|f_x(t,1/y)|$ by the Schauder estimation) is bounded in $y$ by Lemma \ref{lem:3/2_MMG}.
It follows that
\begin{equation}
\begin{aligned}
g(t,y)&=\mathbb{E}^{\hat{\mathbb{P}}}[(Y_t^{-1}h'(Y_t^{-1})+\eta h(Y_t^{-1}))Y_t^{1-\eta}e^{-2\alpha\int_0^tY_t^{-1}\,ds}|Y_0=y]  e^{bt}\\
&=\mathbb{E}^{\hat{\mathbb{P}}}[(\hat{X}_th'(\hat{X}_t)+\eta h(\hat{X}_t))\hat{X}_t^{\eta-1}e^{-2\alpha\int_0^t\hat{X}_s\,ds}|\hat{X}_0=1/y]  e^{bt}\,.
\end{aligned}
\end{equation}
Thus, 
$$f_x(t,x)=g(t,1/x)=\mathbb{E}^{\hat{\mathbb{P}}}[(\hat{X}_th'(\hat{X}_t)+\eta h(\hat{X}_t))\hat{X}_t^{\eta-1}e^{-2\alpha\int_0^t\hat{X}_s\,ds}|\hat{X}_0=x]  e^{bt}\,,$$
which is the desired result.

To analyze 
$$f_x(t,x)e^{-bt}=\mathbb{E}^{\hat{\mathbb{P}}}[(\hat{X}_th'(\hat{X}_t)+\eta h(\hat{X}_t))\hat{X}_t^{\eta-1}e^{-2\alpha\int_0^t\hat{X}_s\,ds}|\hat{X}_0=x]\,,$$
we apply the Hansen--Scheinkman decomposition.
It can be shown that the pair
$(2b,x^{-2})$
is the recurrent eigenpair by considering the generator of $\hat{X}$ with killing rate $2\alpha \hat{X}_t,$ which is given as
$$\mathcal{L}\psi(x):=\frac{1}{2}\sigma^2x^3\psi''(x)+(b-(\alpha-\frac{3}{2}\sigma^2)x)x\psi'(x)-2\alpha x\psi(x)\,. $$
Consider a consistent family 
$(\tilde{\mathbb{P}}_t)_{t\ge0}$  of probability measures 
where each $\tilde{\mathbb{P}}_t$ is defined on   $\mathcal{F}_t$ as
$$\frac{d\tilde{\mathbb{P}}_t}{d\hat{\mathbb{P}}_t}
=e^{2bt-2\alpha\int_{0}^{t}\hat{X}_s\,ds}\,\frac{\hat{X}_0^2}{\hat{X}_t^2}\,.$$
Then,
$$\tilde{B}_t:=2\sigma\int_0^t\hat{X}_s^{1/2}\,ds+\hat{B}_t\,,\;t\ge0$$
is a $(\tilde{\mathbb{P}}_t)_{t\ge0}$-Brownian motion, and $\hat{X}$ follows
$$d\hat{X}_t=(b-(\alpha+\frac{1}{2}\sigma^2)\hat{X}_t)\hat{X}_t\,dt+\sigma \hat{X}_t^{3/2}\,d\tilde{B}_t\,.$$
Using this consistent family of probability measures, we have
\begin{equation}\label{eqn:3/2_reminder_f_y}
\begin{aligned}
f_x(t,x)
&=\mathbb{E}^{\hat{\mathbb{P}}}[(\hat{X}_th'(\hat{X}_t)+\eta h(\hat{X}_t))\hat{X}_t^{\eta-1}e^{-2\alpha\int_0^t\hat{X}_s\,ds}|\hat{X}_0=x]  e^{bt}\\
&=\mathbb{E}^{\tilde{\mathbb{P}}}[(\hat{X}_th'(\hat{X}_t)+\eta h(\hat{X}_t))\hat{X}_t^{\eta+1}|\hat{X}_0=x]  e^{-bt}x^{-2}\,.
\end{aligned}
\end{equation}

We can obtain the large-time behavior of $\partial_\xi p_T$ by 
using Eq.\eqref{eqn:3/2_reminder_f_y}.
Since $h$ and $h'$ have linear growth at most, there is a positive constant $c_0$ such that
$|\hat{X}_th'(\hat{X}_t)+\eta h(\hat{X}_t)|\hat{X}_t^{\eta+1}\leq c_0\hat{X}_t^{\eta+3}.$
By Lemma \ref{lem:3/2_MMG},
the expectation
\begin{equation}
\label{eqn:3/2_remainder}
\mathbb{E}^{\tilde{\mathbb{P}}}\big[|\hat{X}_th'(\hat{X}_t)+\eta h(\hat{X}_t)|\hat{X}_t^{\eta+1}\big|\hat{X}_0=x\big]
\end{equation}
is uniformly bounded in $(t,x)$ since the constant $A$ in the lemma satisfies
$A=\eta+3<\frac{2\alpha}{\sigma^2}+3.$
Therefore,
\begin{equation} \label{eqn:3/2_delta_decay}
\begin{aligned}
\left|\frac{\partial_\xi p_T}{p_T}+\frac{\,\eta\,}{\xi}\right|
=\left|\frac{  f_\xi(T,\xi)}{f(T,\xi)}  \right|\leq c_1  e^{-bT}
\end{aligned}
\end{equation}
for some positive constant $c_1.$ This gives the desired result.

\subsection{Second-order sensitivity of $\xi$}
\label{app:3/2_gamma}

We investigate the large-time asymptotic behavior of the second-order sensitivity  with respect to the initial value $\xi.$
In this section, assume that
$h$ is twice continuously differentiable, $h$ and  $h'$ have linear growth at most and $h''$ is bounded, moreover $h,$ $h',$ $h''$ are nonnegative (we are mainly interested in the case $h=1$).
From Eq.\eqref{eqn:3/2_HS_decompo}, 
\begin{equation} \label{eqn:3/2_gamma}
\begin{aligned}
\frac{{\partial_{\xi\xi}} p_T}{p_T}
=\frac{  f_{xx}(T,\xi)}{f(T,\xi)} -\frac{2\eta}{\xi} \frac{f_x(T,\xi)}{f(T,\xi)} +\frac{\eta(\eta+1)}{\xi^2} \,.
\end{aligned}
\end{equation}
Since we already estimated the large-time asymptotic behavior of 
$f_x(T,\xi),$ we focus on the second-order derivative $f_{xx}(T,\xi).$
Define
$$\hat{f}(t,x):=\mathbb{E}^{\tilde{\mathbb{P}}}[(\hat{X}_th'(\hat{X}_t)+\eta h(\hat{X}_t))\hat{X}_t^{\eta+1}|\hat{X}_0=x] $$
so that Eq.\eqref{eqn:3/2_reminder_f_y} gives $f_x(t,x)=\hat{f}(t,x)e^{-bt}x^{-2}.$ Thus,
\begin{equation}
\label{eqn:3/2_f_yy}
f_{xx}(t,x)=\hat{f}_x(t,x)e^{-bt}x^{-2}-2\hat{f}(t,x)e^{-bt}x^{-3}\,.
\end{equation} 

We need to estimate the large-time behavior of $\hat{f}_x(t,x).$
The Feynman--Kac formula states that  
$$-\hat{f}_t+\frac{1}{2}\sigma^2x^3\hat{f}_{xx}+(b-(\alpha+\frac{1}{2}\sigma^2)x)x\hat{f}_x=0\,,\;\hat{f}(0,x)=(h'(x)x+\eta h(x))x^{\eta+1}\,.$$
Differentiate this equation in $x,$ then 
\begin{equation}
\begin{aligned}
&-\hat{f}_{xt}+\frac{1}{2}\sigma^2x^3\hat{f}_{xxx}+(b-(\alpha-\sigma^2)x)x\hat{f}_{xx}+(b-(2\alpha+\sigma^2)x)\hat{f}_{x}=0\,,\\
&\hat{f}_x(0,x)=(h''(x)x^2+2(\eta+1) x h'(x)+\eta(\eta+1)h(x))x^{\eta}\,.
\end{aligned}
\end{equation}
It is easy to show that
the quadruple of functions $$\Big((b-(\alpha-\sigma^2)x)x,\sigma {x}^{3/2},-b+(2\alpha+\sigma^2) x, (h''(x)x^2+2(\eta+1) x h'(x)+\eta(\eta+1)h(x))x^{\eta}\Big),\,x>0$$ satisfies Assumptions \ref{assume:a} -- \ref{assume:h}.
The corresponding process 
$\tilde{X}$ is the solution of 
$$d\tilde{X}_t=(b-(\alpha-\sigma^2)\tilde{X}_t)\tilde{X}_t\,dt+\sigma \tilde{X}_t^{3/2}\,d\tilde{B}_t \,.$$
We observe that  
$$\hat{f}_x(t,x)=\mathbb{E}^{\tilde{\mathbb{P}}}[(h''(\tilde{X}_t)\tilde{X}_t^2
+2(\eta+1) \tilde{X}_th'(\tilde{X}_t)+\eta(\eta+1)h(\tilde{X}_t))\tilde{X}_t^{\eta}e^{-(2\alpha+\sigma^2)\int_0^t\tilde{X}_s\,ds}|\tilde{X}_0=x]e^{bt}\,.$$
This can be obtained by the same argument used in the derivation of Eq.\eqref{eqn:3/2_delta_decom} by
considering $Y=1/X.$

To analyze the expectation 
above,
we apply the Hansen--Scheinkman decomposition.
It can be shown that the pair
$(2b,x^{-2})$
is the recurrent eigenpair
by considering the generator of $\tilde{X}$ with killing rate $(2\alpha+\sigma^2)\tilde{X}_t,$ which is given as
$$\mathcal{L}\psi(x):=\frac{1}{2}\sigma^2x^3\psi''(x)+(b-(\alpha-\sigma^2)x)x\psi'(x)-(2\alpha+\sigma^2)x\psi(x)\,. $$
Consider a consistent family $(\bar{\mathbb{P}}_t)_{t\ge0}$ of probability measures
where each $\bar{\mathbb{P}}_t$ is defined on $\mathcal{F}_t$
as
$$\frac{d\bar{\mathbb{P}}_t}{d\tilde{\mathbb{P}}_t}
=e^{2bt-(2\alpha+\sigma^2)\int_{0}^{t}\tilde{X}_s\,ds}\,\frac{\tilde{X}_0^2}{\tilde{X}_t^2}\,.$$
Then,
$$\bar{B}_t:=2\sigma\int_0^t\tilde{X}_s^{1/2}\,ds+\tilde{B}_t\,,\;t\ge0$$
is a $(\bar{\mathbb{P}}_t)_{t\ge0}$-Brownian motion, and $\tilde{X}$ follows
$$d\tilde{X}_t=(b-(\alpha+\sigma^2)\tilde{X}_t)\tilde{X}_t\,dt+\sigma \tilde{X}_t^{3/2}\,d\bar{B}_t\,.$$
Using this consistent family of probability measures $(\bar{\mathbb{P}}_t)_{t\ge0},$ we have
\begin{equation}\label{eqn:3/2_reminder_f_yy}
\begin{aligned}
\hat{f}_x(t,x)&=\mathbb{E}^{\tilde{\mathbb{P}}}[(h''(\tilde{X}_t)\tilde{X}_t^2
+2(\eta+1) \tilde{X}_th'(\tilde{X}_t)+\eta(\eta+1)h(\tilde{X}_t))\tilde{X}_t^{\eta}e^{-(2\alpha+\sigma^2)\int_0^t\tilde{X}_s\,ds}|\tilde{X}_0=x]e^{bt}\\
&=\mathbb{E}^{\bar{\mathbb{P}}}[(h''(\tilde{X}_t)\tilde{X}_t^2
+2(\eta+1) \tilde{X}_th'(\tilde{X}_t)+\eta(\eta+1)h(\tilde{X}_t))\tilde{X}_t^{\eta+2}|\tilde{X}_0=x]e^{-bt}x^{-2}\,.
\end{aligned}
\end{equation}

Since  $h,$ $h'$ have linear growth at most and $h''$ is bounded, 
by Lemma \ref{lem:3/2_MMG},
the expectation
$$\mathbb{E}^{\bar{\mathbb{P}}}\big[|h''(\tilde{X}_t)\tilde{X}_t^2
+2(\eta+1) \tilde{X}_th'(\tilde{X}_t)+\eta(\eta+1)h(\tilde{X}_t)|\tilde{X}_t^{\eta+2}\big|\tilde{X}_0=x\big]$$
converges to a positive constant as $t\to\infty$
since the constant $A$ in the lemma satisfies
$A=\eta+4<\frac{2\alpha}{\sigma^2}+4.$
Thus,
\begin{equation}
\begin{aligned}
|f_{xx}(t,x)|
&=|\hat{f}_x(t,x)e^{-bt}x^{-2}-2g(t,x)e^{-bt}x^{-3}|\\
&\leq  |\hat{f}_x(t,x)|e^{-bt}x^{-2}+2|g(t,x)|e^{-bt}x^{-3}\\
&\le  c_1e^{-2bt}x^{-4}+c_1e^{-bt}x^{-3}\\
&\le c_2 e^{-bt}
 \end{aligned}
\end{equation}
 for some positive constants $c_1$ and $c_2,$ which are independent of $t.$
From  Eq.\eqref{eqn:3/2_delta_decay},
we conclude that
\begin{equation} 
\begin{aligned}
\left|\frac{{\partial_{\xi\xi}} p_T}{p_T}
-\frac{\eta(\eta+1)}{\xi^2}\right|
=\left|\frac{  f_{xx}(T,\xi)}{f(T,\xi)} -\frac{2\eta}{\xi} \frac{f_x(T,\xi)}{f(T,\xi)} \right|\leq \left|\frac{  f_{xx}(T,\xi)}{f(T,\xi)}\right| +\frac{2\eta}{\xi}\left| \frac{f_x(T,\xi)}{f(T,\xi)} \right|\leq c_3e^{-bT}
\end{aligned}
\end{equation}
for some positive constant $c_3.$ This gives the convergence rate of the second-order sensitivity.
We can also provide  a higher-order convergence rate of the second-order sensitivity as follows.
From Eq.\eqref{eqn:3/2_f_yy} and $p_T=f(T,\xi)e^{\lambda T}\phi(\xi)$ presented in
Eq.\eqref{eqn:3/2_decom}, it follows that
\begin{equation}
\begin{aligned}
\frac{\partial_{\xi\xi}p_T}{p_T}-\Big(\frac{\partial_\xi p_T}{p_T}\Big)^2+\frac{2}{\xi}\Big(\frac{\partial_\xi p_T}{p_T}-\frac{\phi'(\xi)}{\phi(\xi)}\Big)-\frac{\phi''(\xi)}{\phi(\xi)}+\Big(\frac{\phi'(\xi)}{\phi(\xi)}\Big)^2
&=\frac{g_x(T,\xi)}{g(T,\xi)}\frac{f_x(T,\xi)}{f(T,\xi)}-\Big(\frac{f_x(T,\xi)}{f(T,\xi)}\Big)^2\,.
\end{aligned}
\end{equation}
Using $\phi(\xi)=\xi^{-\eta},$ we have
\begin{equation}
\begin{aligned}
\left|\frac{\partial_{\xi\xi}p_T}{p_T}-\Big(\frac{\partial_\xi p_T}{p_T}\Big)^2+\frac{2}{\xi}\frac{\partial_\xi p_T}{p_T}+\frac{\eta}{\xi^2}\right|
&\leq\left|\frac{g_x(T,\xi)}{g(T,\xi)}\frac{f_x(T,\xi)}{f(T,\xi)}\right|+\Big(\frac{f_x(T,\xi)}{f(T,\xi)}\Big)^2\leq c_4e^{-2bT}
\end{aligned}
\end{equation}
for some positive constant $c_4.$
For the last inequality, we used Eq.\eqref{eqn:3/2_delta_decay} and Eq.\eqref{eqn:3/2_reminder_f_yy}.

\subsection{Sensitivity of $b$}

We investigate the large-time asymptotic behavior of the  sensitivity of $p_T$ with respect to the parameter $b.$
In this section, assume that $h$ is continuously differentiable and that $h,$ $h'$ have linear growth at most and are nonnegative.
Recall the   Hansen--Scheinkman decomposition 
$p_T 
=f(T,\xi)\xi^{-\eta}\,e^{-\lambda T}$
and the remainder function  $f(t,x)=\mathbb{E}_x^{\hat{\mathbb{P}}}[h(X_t)X_t^{\eta}] $
in Eq.\eqref{eqn:3/2_HS_decompo} and Eq.\eqref{eqn:3/2_reminder}.
Since $\eta$ is independent of $b$ and
$\lambda=b\eta,$ we have 
$$\frac{\partial_b p_T}{p_T}=\frac{f_b(T,\xi)}{f(T,\xi)}-\eta T\,.$$
It can be easily shown that $f$ is continuously differentiable in $b$  by considering the density function of $X_t$ or by using \cite[Theorem 4.8]{park2018sensitivity}.
We focus on  the  large-time behavior of $f_b(T,\xi).$ 
Differentiate Eq.\eqref{eqn:3/2_FK} in $b,$ then
\begin{equation}
\label{eqn:f_b_PDE}-f_{bt}+\frac{1}{2}\sigma^2x^3f_{bxx}+(b-\alpha x)xf_{bx}+xf_x=0\,,\;f_b(0,x)=0\,.
\end{equation}
Recall that the quadruple of functions in Eq.\eqref{eqn:3/2_quad_initial_sen}
satisfies Assumptions \ref{assume:a} -- \ref{assume:h}.

From this PDE, we observe that the remainder function $f$ satisfies
$$f_b(t,x)
=\mathbb{E}_x^{\hat{\mathbb{P}}}\Big[\int_0^tX_sf_x(s,X_s)\,ds\Big]\,.$$
However, to obtain this equality,   
the  Feynman--Kac formula in Remark \ref{remark:FK} cannot be applied directly because the
volatility function $\sigma x^{3/2}$  of $X$ does not have linear growth. Instead, we define $Y:=1/X$ and  $g(t,y):=f_b(t,1/y).$ Then, $Y$ is a CIR process  
$$dY_t=(\alpha+\sigma^2-bY_t)\,dt-\sigma\sqrt{Y_t}\,d\hat{B}_t\,,$$
and Eq.\eqref{eqn:f_b_PDE} becomes
$$-g_t+\frac{1}{2}\sigma^2yg_{yy}+(\alpha+\sigma^2-by)g_y+(1/y)f_x(t,1/y)=0\,,\;g(0,y)=0\,.$$
This PDE satisfies all the conditions in Remark \ref{remark:FK}. 
Note that the function $(1/y)f_x(t,1/y)$ has linear growth at most in $y$ by Eq.\eqref{eqn:3/2_reminder_f_y} because Eq.\eqref{eqn:3/2_remainder} is uniformly   bounded in $(t,x).$ 
It follows that
\begin{equation}
\begin{aligned}
g(t,y)&=\mathbb{E}^{\hat{\mathbb{P}}}\Big[\int_0^t (1/Y_s)f_x(s,1/Y_s)\,ds \Big|Y_0=y\Big]   \\
&=\mathbb{E}^{\hat{\mathbb{P}}}\Big[\int_0^t X_sf_x(s,X_s)\,ds \Big|X_0=1/y\Big] \,.
\end{aligned}
\end{equation}
Thus, 
$$f_x(t,x)=g(t,1/x)=\mathbb{E}^{\hat{\mathbb{P}}}\Big[\int_0^t X_sf_x(s,X_s)\,ds \Big|X_0=x\Big] \,,$$
which is the desired result.

Since $h$ and $h'$ have linear growth at most, there is a positive constant $c_0$ such that
$|Y_th'(Y_t)+\eta h(Y_t)|Y_t^{\eta+1}\leq c_0Y_t^{\eta+3}.$
From Eq.\eqref{eqn:3/2_reminder_f_y}, we have
\begin{equation}\label{eqn:3/2_b_f_x}
\begin{aligned}
|f_x(t,x)| 
&\leq\mathbb{E}^{\tilde{\mathbb{P}}}\big[|Y_th'(Y_t)+\eta h(Y_t)|Y_t^{\eta+1}\big|Y_0=x\big]  e^{-bt}x^{-2}\\
&\leq  c_0 \mathbb{E}^{\tilde{\mathbb{P}}}[Y_t^{\eta+3}|Y_0=x]  e^{-bt}x^{-2} \\
&\leq c_1e^{-bt}x^{-2} 
\end{aligned}
\end{equation}
for some positive constant $c_1,$ which is independent of $t$ and $x.$ Here,
we used Lemma \ref{lem:3/2_MMG}, which gives that
the expectation
$\mathbb{E}^{\tilde{\mathbb{P}}}[Y_t^{\eta+3}|Y_0=x]$
is uniformly bounded in $(t,x)$ on $[0,\infty)\times (0,\infty)$
since the constant $A$ in the lemma satisfies
$0<A=\eta+3<\frac{2\alpha}{\sigma^2}+3.$
Thus,
$$|f_b(t,x)|
\leq\mathbb{E}_x^{\hat{\mathbb{P}}}\Big[\int_0^tX_s|f_x(s,X_s)|\,ds\Big]\leq
c_1 \int_0^t
e^{-bs}\mathbb{E}_x^{\hat{\mathbb{P}}}[X_s^{-1}]\,ds \,.$$
By Lemma \ref{lem:3/2_MMG}, the expectation $\mathbb{E}_x^{\hat{\mathbb{P}}}[X_s^{-1}]$ is bounded in $s$ on $[0,\infty)$ since the expectation converges to a positive constant. Therefore,
$$|f_b(t,x)|
\leq
c_2 b \int_0^t
e^{-bs}\,ds \leq c_2(1-e^{-bt})\leq c_2$$
for some positive constant $c_2.$
Since $f(T,\xi)$ converges to a positive constant as $T\to\infty,$ we conclude that  
$$\left|\frac{1}{T}\frac{\partial_b p_T}{p_T}+\eta \right|=\frac{1}{T}\left|\frac{f_b(T,\xi)}{f(T,\xi)}\right|\leq \frac{c_3}{T}$$
for some positive constant $c_3.$

\subsection{Sensitivity of $a$}
\label{sec:3/2_sen_a}

We study the large-time asymptotic behavior of the  sensitivity of $p_T$ with respect to the parameter $a.$
In this section, assume that $h$ is nonzero, nonnegative, continuously differentiable and that $h,$ $h'$ have linear growth at most. 
From the decomposition in 
Eq.\eqref{eqn:3/2_HS_decompo}, it follows that
$$\frac{\partial_a p_T}{p_T}=\frac{f_a(T,\xi)}{f(T,\xi)}
-(\partial_a\eta) \ln\xi-T\partial_a\lambda \,.$$
It can be easily shown that  $f$ is continuously differentiable in $a$ by considering the density function of $X_t$ or by using \cite[Theorem 4.8]{park2018sensitivity}.
We focus on the  large-time behavior of $f_a(T,\xi).$

Differentiate Eq.\eqref{eqn:3/2_FK} in $a,$ then
$$-f_{at}+\frac{1}{2}\sigma^2x^3f_{axx}+(b-\alpha x)xf_{ax}-(\partial_a\alpha)x^2f_x=0\,,\;f_a(0,x)=h(x)x^\eta(\ln x)\partial_a\eta\,.$$
It follows that 
\begin{equation} \label{eqn:3/2_f_a}
f_a(t,x)
=\mathbb{E}_x^{\hat{\mathbb{P}}}\Big[-(\partial_a\alpha)\int_0^tX_s^2f_x(s,X_s)\,ds+h(X_t)X_t^\eta (\ln X_t)\partial_a\eta\Big]\,.
\end{equation}
However, to obtain this equality,  
the  Feynman--Kac formula in Remark \ref{remark:FK} cannot be applied directly because the
volatility function $\sigma x^{3/2}$  of $X$ does not have linear growth. Instead, we define $Y:=1/X$ and  $g(t,y):=-f_a(t,1/y).$ Then, $Y$ is a CIR process  
$$dY_t=(\alpha+\sigma^2-bY_t)\,dt-\sigma\sqrt{Y_t}\,d\hat{B}_t\,,$$
and Eq.\eqref{eqn:f_b_PDE} becomes
$$-g_t+\frac{1}{2}\sigma^2yg_{yy}+(\alpha+\sigma^2-by)g_y+(\partial_a\alpha)(1/y^2)f_x(t,1/y)=0\,,\;g(0,y)=h(1/y)(1/y)^\eta(\ln y)\partial_a\eta\,.$$
This PDE satisfies all the conditions in Remark \ref{remark:FK}. 
Note that the function $(\partial_a\alpha)(1/y^2)f_x(t,1/y)$ is bounded in $y$ by Eq.\eqref{eqn:3/2_reminder_f_y} since Eq.\eqref{eqn:3/2_remainder} is bounded in $(t,x),$ moreover  $g(0,y)$ is bounded below since $\partial_a\eta <0$ and $\ln y<0$ for small $y>0$   and $h(1/y)(1/y)^\eta(\ln y)$ is bounded  for large $y.$
It follows that
\begin{equation}
\begin{aligned}
g(t,y)&= \mathbb{E}^{\hat{\mathbb{P}}}\Big[-(\partial_a\alpha)\int_0^t (1/Y_s^2)f_x(s,1/Y_s)\,ds-h(1/Y_t)(1/Y_t)^\eta(\ln Y_t)\partial_a\eta \Big|Y_0=y\Big]   \\
&=\mathbb{E}^{\hat{\mathbb{P}}}\Big[-(\partial_a\alpha)\int_0^t X_s^2f_x(s,X_s)\,ds+h(X_t)X_t^\eta(\ln X_t)\partial_a\eta \Big|X_0=1/y\Big] \,.
\end{aligned}
\end{equation}
Thus,
$$f_a(t,x)=g(t,1/x)=\mathbb{E}^{\hat{\mathbb{P}}}\Big[-(\partial_a\alpha)\int_0^t X_s^2f_x(s,X_s)\,ds+h(X_t)X_t^\eta(\ln X_t)\partial_a\eta \Big|X_0=x\Big]  \,,$$
which is the desired result.

Using Lemma \ref{lem:3/2_MMG} and Eq.\eqref{eqn:3/2_b_f_x}, we have
\begin{equation} 
|f_a(t,x)|
\leq c_1\mathbb{E}_x^{\hat{\mathbb{P}}}\Big[\int_0^te^{-bs}\,ds\Big]+c_2\leq \frac{c_1}{b}+c_2
\end{equation}
for some positive constants $c_1$ and $c_2.$
Since $f(T,\xi)$ converges to a positive constant as $T\to\infty,$ we conclude that  
$$\left|\frac{1}{T}\frac{\partial_a p_T}{p_T}+\partial_a\lambda \right|=\frac{1}{T}\left|\frac{f_a(T,\xi)}{f(T,\xi)}-(\partial_a\eta) \ln\xi\right|\leq \frac{c_3}{T}$$
for some positive constant $c_3.$ Furthermore,  direct calculation gives
$$\partial_a\lambda=-\frac{b}{\sigma^2}\frac{\sqrt{(a+\sigma^2/2)^2+2q\sigma^2}-(a+\sigma^2/2)}{\sqrt{(a+\sigma^2/2)^2+2q\sigma^2}}
\,.$$

\subsection{Sensitivity of $\sigma$}

We study the large-time asymptotic behavior of the  sensitivity of $p_T$ with respect to the parameter $\sigma.$
In this section, assume that
$h$ is nonzero, nonnegative,  twice continuously differentiable and that $h,$ $h'$ have linear growth at most and $h''$ is bounded.
From the decomposition in Eq.\eqref{eqn:3/2_HS_decompo},
it follows that
$$\frac{\partial_\sigma p_T}{p_T}=\frac{f_\sigma(T,\xi)}{f(T,\xi)}
-(\partial_\sigma\eta) \ln\xi
-T\partial_\sigma\lambda \,.$$
It can be easily shown that $f$ is continuously differentiable in $\sigma$ by considering the density function of $X_t$ or by using \cite[Theorem 4.13]{park2018sensitivity}.
We focus on the large-time behavior of $f_\sigma(T,\xi).$

Differentiate Eq.\eqref{eqn:3/2_FK} in $\sigma,$ then
$$-f_{\sigma t}+\frac{1}{2}\sigma^2x^3f_{\sigma xx}+(b-\alpha x)xf_{\sigma x}+\sigma x^3f_{xx}-(\partial_\sigma \alpha)x^2f_x=0\,,\;f(0,x)=h(x)x^{\eta}(\ln x)\partial_\sigma\eta\,.$$
It follows that 
\begin{equation} 
\begin{aligned}
f_\sigma(t,x)
&=\mathbb{E}_x^{\hat{\mathbb{P}}}\Big[\sigma\int_0^tX_s^3f_{xx}(s,X_s)\,ds-(\partial_\sigma\alpha)\int_0^tX_s^2f_x(s,X_s)\,ds+h(X_t)X_t^\eta (\ln X_t) \partial_\sigma\eta\Big]\\
&=\mathbb{E}_x^{\hat{\mathbb{P}}}\Big[\sigma\int_0^tX_s^3f_{xx}(s,X_s)\,ds\Big]+\mathbb{E}_x^{\hat{\mathbb{P}}}\Big[-(\partial_\sigma\alpha)\int_0^tX_s^2f_x(s,X_s)\,ds+h(X_t)X_t^\eta (\ln X_t) \partial_\sigma\eta\Big]\,.
\end{aligned}
\end{equation}
We claim that the two expectations on the right-hand side are bounded in $t$ on $[0,\infty).$ 
The second expectation is bounded in $t$ on $[0,\infty)$ by the same method used in the analysis of Eq.\eqref{eqn:3/2_f_a}.
To estimate the first expectation, observe that
\begin{equation} 
f_{xx}(t,x)=g_x(t,x)e^{-bt}x^{-2}-2g(t,x)e^{-bt}x^{-3} \,,
\end{equation} 
which is presented in Eq.\eqref{eqn:3/2_f_yy}.
It follows that
$$\mathbb{E}_x^{\hat{\mathbb{P}}}\Big[\int_0^tX_s^3f_{xx}(s,X_s)\,ds\Big]=\mathbb{E}_x^{\hat{\mathbb{P}}}\Big[\int_0^tX_sg_x(s,X_s)e^{-bs}\,ds\Big]-2\mathbb{E}_x^{\hat{\mathbb{P}}}\Big[\int_0^tg(s,X_s)e^{-bs}\,ds\Big]\,.$$
By the same method used in Eq.\eqref{eqn:3/2_b_f_x}, 
there is a positive constant $c_1,$ which is independent of $t$ and $x,$ such that
$$|g(t,x)|\leq \mathbb{E}^{\tilde{\mathbb{P}}}\big[|Y_th'(Y_t)+\eta h(Y_t)|Y_t^{\eta+1}\big|Y_0=x\big]\leq c_1\,.$$
Thus, $\mathbb{E}_x^{\hat{\mathbb{P}}}[\int_0^tg(s,X_s)e^{-bs}\,ds]$ is bounded in $t$ on $[0,\infty).$
From Eq.\eqref{eqn:3/2_reminder_f_yy},
\begin{equation} 
\begin{aligned}
|g_x(t,x)|
&\leq \mathbb{E}^{\bar{\mathbb{P}}}\big[|h''(Z_t)Z_t^2
+2(\eta+1) Z_th'(Z_t)+\eta(\eta+1)h(Z_t)|Z_t^{\eta+2}\big|Z_0=x\big]e^{-bt}x^{-2}\\
&\leq c_2\mathbb{E}^{\bar{\mathbb{P}}}[Z_t^{\eta+4}\big|Z_0=x]e^{-bt}x^{-2} 
\end{aligned}
\end{equation}
for some positive constant $c_2.$
By Lemma \ref{lem:3/2_MMG},
the expectation
$\mathbb{E}^{\bar{\mathbb{P}}}[Z_t^{\eta+4}\big|Z_0=x]$
is uniformly bounded in $(t,x)$ on the domain $[0,\infty)\times(0,\infty),$ which implies that
$$|g_x(t,x)|
\leq  c_3e^{-bt}x^{-2} $$
for some positive constant $c_3.$
Thus,
$$\mathbb{E}_x^{\hat{\mathbb{P}}}\Big[\int_0^tX_s|g_x(s,X_s)|e^{-bs}\,ds\Big]\leq c_3\mathbb{E}_x^{\hat{\mathbb{P}}}\Big[\int_0^tX_s^{-1}e^{-2bs}\,ds\Big]=c_3 \int_0^t\mathbb{E}_x^{\hat{\mathbb{P}}}[X_s^{-1}]e^{-2bs}\,ds\,.$$
By Lemma \ref{lem:3/2_MMG}, the expectation $\mathbb{E}_x^{\hat{\mathbb{P}}}[X_s^{-1}]$ is bounded in $s$ on $[0,\infty)$ since the expectation converges to a positive constant, which gives the desired result.
In conclusion,
$$\left|\frac{1}{T}\frac{\partial_\sigma p_T}{p_T}+\partial_\sigma\lambda \right|=\frac{1}{T}\left|\frac{f_\sigma(T,\xi)}{f(T,\xi)}-(\partial_\sigma \eta) \ln\xi\right|\leq \frac{c_4}{T}$$
 for some positive constant $c_4.$
Furthermore, direct calculation gives
$$\partial_\sigma\lambda=
 \frac{b(a+{\sigma^2}/{2}+2 q-\sqrt{(a+{\sigma^2}/{2})^2+2 q\sigma^2})}{\sigma \sqrt{(a+{\sigma^2}/{2})^2+2 q\sigma^2}} -\frac{2b}{\sigma^3}(\sqrt{(a+{\sigma^2}/{2})^2+2 q \sigma^2}-a- {\sigma^2}/{2})\,.$$

\bibliographystyle{apa}

\bibliography{sensitivity}
\end{document}